\documentclass[11pt,a4paper,reqno,intlimits,sumlimits]{amsart}

\usepackage{amssymb}
\usepackage[dvipsnames]{xcolor}
\usepackage[mathscr]{euscript}
\usepackage{xspace}
\usepackage{enumerate}
\usepackage{epsfig,rotating}
\usepackage{graphicx}
\input{xy}\xyoption{all}
\usepackage[backgroundcolor=orange,linecolor=orange]{todonotes}
\usepackage{amsmath}
\usepackage{mathtools}
\usepackage{lscape}
\usepackage{enumitem}
\usepackage[round,longnamesfirst]{natbib}
\usepackage{upgreek}
\usepackage{bbm}
\usepackage{subfigure} 
\usepackage{epstopdf}
\usepackage{nth}
\usepackage{dsfont}

\usepackage[left=1.20in,right=1.20in,top=1.40in,bottom=1.30in]{geometry}

\DeclareMathAlphabet{\mathpzc}{OT1}{pzc}{m}{it}

\usepackage[bookmarksopen,pdfstartview=FitH]{hyperref}
\hypersetup{colorlinks,%
            citecolor=blue,%
            filecolor=blue,%
            linkcolor=blue,%
            urlcolor=blue}%

\begin{document}


\theoremstyle{plain}
\newtheorem{theorem}{Theorem}[section]
\newtheorem{lemma}[theorem]{Lemma}
\newtheorem{proposition}[theorem]{Proposition}
\newtheorem*{proposition*}{Proposition}
\newtheorem{claim}[theorem]{Claim}
\newtheorem{corollary}[theorem]{Corollary}
\newtheorem{axiom}{Axiom}

\theoremstyle{definition}
\newtheorem{remark}[theorem]{Remark}
\newtheorem{note}{Note}[section]
\newtheorem{definition}[theorem]{Definition}
\newtheorem{example}[theorem]{Example}
\newtheorem*{ackn}{Acknowledgements}
\newtheorem{assumption}{Assumption}
\newtheorem{approach}{Approach}
\newtheorem{critique}{Critique}
\newtheorem{question}{Question}
\newtheorem{aim}{Aim}
\newtheorem*{assucd}{Assumption ($\mathbb{CD}$)}
\newtheorem*{asa}{Assumption ($\mathbb{A}$)}
\newtheorem*{appS}{Approximation ($\mathbb{S}$)}
\newtheorem*{appBS}{Approximation ($\mathbb{BS}$)}
\renewcommand{\theequation}{\thesection.\arabic{equation}}
\numberwithin{equation}{section}

\renewcommand{\thefigure}{\thesection.\arabic{figure}}
\numberwithin{equation}{section}

\newcommand{\Law}{\ensuremath{\mathop{\mathrm{Law}}}}
\newcommand{\loc}{{\mathrm{loc}}}

\let\SETMINUS\setminus
\renewcommand{\setminus}{\backslash}

\def\stackrelboth#1#2#3{\mathrel{\mathop{#2}\limits^{#1}_{#3}}}

%

\newcommand{\prozess}[1][L]{{\ensuremath{#1=(#1_t)_{t\in[0,T]}}}\xspace}
\newcommand{\prazess}[1][L]{{\ensuremath{#1=(#1_t)_{t\ge0}}}\xspace}
\newcommand{\pt}[1][N]{\ensuremath{\P_{#1}}\xspace}
\newcommand{\tk}[1][N]{\ensuremath{T_{#1}}\xspace}
\newcommand{\dd}[1][]{\ensuremath{\ud{#1}}\xspace}

\newcommand{\scal}[2]{\ensuremath{\left\langle #1, #2 \right\rangle}}
\newcommand{\bscal}[2]{\ensuremath{\big\langle #1, #2 \big\rangle}}
\newcommand{\set}[1]{\ensuremath{\left\{#1\right\}}}
\newcommand{\brac}[1]{\ensuremath{\left( #1\right)}}
\newcommand{\brak}[1]{\ensuremath{\left[ #1\right]}}
\newcommand{\cond}[2]{\ensuremath{\left[ \left. #1\right| #2\right]}}
\def\lev{L\'{e}vy\xspace}
\def\lk{L\'{e}vy--Khintchine\xspace}
\def\lib{LIBOR\xspace}
\def\mg{martingale\xspace}
\def\smmg{semimartingale\xspace}
\def\alm{affine LIBOR model\xspace}
\def\alms{affine LIBOR models\xspace}
\def\dalms{defaultable \alms}
\def\ap{affine process\xspace}
\def\aps{affine processes\xspace}

\def\half{\frac1{2}}

\def\F{\ensuremath{\mathcal{F}}}
\def\bD{\mathbf{D}}
\def\bF{\mathbf{F}}
\def\bG{\mathbf{G}}
\def\bH{\mathbf{H}}
\def\R{\ensuremath{\mathbb{R}}}
\def\Rp{\mathbb{R}_{\geqslant0}}
\def\Rm{\mathbb{R}_{\leqslant 0}}
\def\C{\ensuremath{\mathbb{C}}}
\def\U{\ensuremath{\mathcal{U}}}
\def\I{\mathcal{I}}
\def\N{\mathbbN}

\def\P{\ensuremath{\mathds P}} 
\def\Q{\mathbb{Q}}
\def\E{\ensuremath{\mathds E}} 

\def\hP{\ensuremath{\widehat{\mathrm{I\kern-.2em P}}}}
\def\hE{\ensuremath{\widehat{\mathrm{I\kern-.2em E}}}}

\def\bP{\ensuremath{\overline{\mathrm{I\kern-.2em P}}}}
\def\bE{\ensuremath{\overline{\mathrm{I\kern-.2em E}}}}

\def\bphi{\overline{\phi}}
\def\bpsi{\overline{\psi}}

\def\ott{{0\leq t\leq T}}
\def\idd{{1\le i\le d}}

\def\icc{\mathpzc{i}}
\def\ecc{\mathbf{e}_\mathpzc{i}}

\def\uk{u_{k+1}}
\def\vk{v_{k+1}}

\def\e{\mathrm{e}}
\def\a{\mathrm{a}}
\def\b{\mathrm{b}}
\def\ud{\ensuremath{\mathrm{d}}}
\def\dt{\ud t}
\def\dr{\ud r}
\def\ds{\ud s}
\def\dx{\ud x}
\def\dy{\ud y}
\def\dv{\ud v}
\def\du{\ud u}
\def\dw{\ud w}
\def\dz{\ud z}
\def\dsdx{\ensuremath{(\ud s, \ud x)}}
\def\dtdx{\ensuremath{(\ud t, \ud x)}}

\def\rx{\mathrm{x}}

\def\1{\boldsymbol1}

\def\tc{\ensuremath{\mathpzc{T}}}
\def\afflm{\ensuremath{(X,\mathcal{X},T_N,u,v)}\xspace}

\def\lsnc{\ensuremath{\mathrm{LSNC-}\chi^2}}
\def\nc{\ensuremath{\mathrm{NC-}\chi^2}}

\newcommand{\cD}{{\mathcal{D}}}
\newcommand{\cF}{{\F}}
\newcommand{\cG}{{\mathcal{G}}}
\newcommand{\cH}{{\mathcal{H}}}
\newcommand{\cK}{{\mathcal{K}}}
\newcommand{\cM}{{\mathcal{M}}}
\newcommand{\cT}{{\mathcal{T}}}
\newcommand{\ha}{{\mathbb{H}}}
\newcommand{\indik}{{\mathbf1}}
\newcommand{\ifdefault}[1]{\ensuremath{\mathbf1_{\{\tau \leq #1\}}}}
\newcommand{\ifnodefault}[1]{\ensuremath{\mathbf1_{\{\tau > #1\}}}}

\newcommand\bovermat[2]{%
  \makebox[0pt][l]{$\smash{\overbrace{\phantom{%
    \begin{matrix}#2\end{matrix}}}^{#1}}$}#2}
\newcommand\bundermat[2]{%
  \makebox[0pt][l]{$\smash{\underbrace{\phantom{%
    \begin{matrix}#2\end{matrix}}}_{#1}}$}#2}
\makeatletter

\newcommand{\ind}{\mathbbm1}

\newcommand{\robert}[1]{\todo[linecolor=blue!40!white,backgroundcolor=blue!20!white,size=\footnotesize]{Robert: #1}}

\def\red{\color{red}}
\def\blue{\color{magenta}}
\def\green{\color{OliveGreen}}

\newcommand{\rr}[1]{{\red #1}}
\newcommand{\rg}[1]{{\green #1}}

\title[Continuous tenor affine LIBOR models and XVA]{Continuous tenor extension\\of affine LIBOR models with multiple curves\\and applications to XVA}

\author[A. Papapantoleon]{Antonis Papapantoleon}
\author[R. Wardenga]{Robert Wardenga}


\address{Institute of Mathematics, University of Mannheim, 68131 Mannheim, Germany}
\email{papapan@math.uni-mannheim.de}

\address{Institut f\"ur Mathematische Stochastik, TU Dresden, 01062 Dresden, Germany}
\email{Robert.Wardenga@tu-dresden.de}



\date{}
\thanks{RW acknowledges funding from the Excellence Initiative of the German Research Foundation (DFG) under grant ZUK 64. Financial support from the Europlace Institute of Finance project ``Post-crisis models for interest rate markets'' is gratefully acknowledged.}

\begin{abstract}
We consider the class of affine LIBOR models with multiple curves, which is an analytically tractable class of discrete tenor models that easily accommodates 
positive or negative interest rates and positive spreads. 
By introducing an interpolating function, we extend the affine LIBOR models to a continuous tenor and derive expressions for the instantaneous forward rate and the short rate. 
We show that the continuous tenor model is arbitrage-free, that the analytical tractability is retained under the spot martingale measure, and that under mild 
conditions an interpolating function can be found such that the extended model fits any initial forward curve. 
This allows us to compute value adjustments (i.e. XVAs) consistently, by solving the corresponding `pre-default' BSDE. 
As an application, we compute the price and value adjustments for a basis swap, and study the model risk associated to different interpolating functions.
\end{abstract}

\maketitle\pagestyle{myheadings}\frenchspacing

\section{Introduction}

In the aftermath of the credit crisis and the European sovereign debt crisis, several of the classical paradigms in finance were no longer able to describe the new reality and needed to be designed afresh. 
On the one hand, significant spreads have appeared between rates of different tenors, which led to the development of multiple curve interest rate models. 
On the other hand, counterparty credit risk has emerged as the native form of default risk, along with liquidity risk, funding constraints and the collateralization of trades. 
Therefore, in post-crisis markets the quoted price of a derivative product (or, better, the cost of its hedging portfolio) is computed as the `clean' price of the product together with several value adjustments that reflect counterparty credit risk, liquidity risk, funding constraints, etc. 
In the context of interest rate derivatives, the `clean' price is typically computed as the discounted expected payoff under a martingale measure using a (discrete tenor) LIBOR market model, while the value adjustments are provided via the solution of a BSDE, which requires the existence of a short rate to discount the cash flows. 

The aim of this work is to compute prices and value adjustments consistently, in the sense that we only calibrate a discrete tenor LIBOR model and then infer the 
dynamics of the short rate from it, instead of resorting to an additional, external short rate model. 
In the sequel we will work with the class of affine LIBOR models with multiple curves. 
This class of models easily produces positive interest rates and positive spreads, as well as negative interest rates alongside positive spreads. 
Moreover, the models are analytically tractable in the sense that the driving process remains affine under all forward measures, which allows to derive explicit expressions for the prices of caplets and semi-analytic expressions for swaptions. 
Thus these models can be efficiently calibrated to market data; cf. \citet{Grbac_Papapantoleon_Schoenmakers_Skovmand_2014} for more details.  

Once the affine LIBOR model has been set up, we introduce an interpolating function that allows to extend the model from a discrete to a continuous tenor, and derive explicit expressions for the dynamics of the instantaneous forward rate and of the short rate. 
This part follows and extends \citet{KellerRessel_2009}, while a similar interpolation for \alms has been recently introduced by \citet{Cuchiero_Fontana_Gnoatto_2016}.
Moreover, we show that the resulting continuous tenor model is arbitrage-free and belongs to the class of affine term structure models. 
Let us mention that, on the contrary, the arbitrage-free interpolation of `classical' LIBOR market models is a challenging task; see e.g. \citet{Beveridge_Joshi_2012}.
In addition, we show that the driving process remains an affine process under the spot martingale measure, hence also the short rate is analytically tractable 
under this measure. 
The choice of the interpolating function is not innocuous though, as it may lead to undesirable behavior of the short rate; e.g. it may induce jumps at fixed times. 
Thus we investigate what properties the (discrete tenor) affine LIBOR model and the interpolating function should have in order to avoid such situations. 
In particular, we show that under mild assumptions there exists an interpolating function such that the extended model can fit any initial forward curve.

Then, we can compute value adjustments via solutions of a `pre-default' BSDE using the framework of \cite{Crepey2012a,Crepey2012b}. 
As an illustration, we design and calibrate an affine LIBOR model, and consider a simple post-crisis interest rate derivative, namely a basis swap. 
Using the methodology outlined above, we derive the dynamics of the short rate and of the basis swap using an interpolating function, and compute the value adjustments for different specifications of the contract. 
As the choice of an interpolating function is still arbitrary, we study the model risk associated to different choices.

This paper is organized as follows: Section \ref{sec:2} reviews affine processes and Section \ref{sec:3} presents an overview of multiple curve markets and affine \lib models. 
Section \ref{sec:4} focuses on the continuous tenor extension of \alms and studies the properties of interpolating functions. 
The final Section \ref{sec:5} outlines the computation of value adjustments in affine LIBOR models, and discusses the model risk associated with the choice of interpolating functions. 
The appendix contains a useful result on the time integration of affine processes.

\section{Affine processes on $\R_{\geqslant0}^d$}
\label{sec:2}

This section provides a brief overview of the basic notions and properties of affine processes. 
Proofs and further details can be found in \citet{DuffieFilipovicSchachermayer03}, in \citet{KellerRessel08}, and in \citet{Filipovic05} for the time-inhomogeneous case. 

Let $(\Omega,\F,\mathbb F,\P)$ denote a complete stochastic basis in the sense of \citet[Def. I.1.3]{JacodShiryaev03}, where $\mathbb F=(\F_t)_{t\in[0,T]}$ and $T\in[0,\infty)$ denotes the time horizon. 
In the sequel, we will consider a process $X$ that satisfies the following:
\begin{asa}
Let $(X,(\P_\rx)_{\rx\in D})$ be a conservative, time-homogeneous, stocha\-stically continuous Markov process taking values in $D=\Rp^d$, i.e. $(\P_\rx)_{\rx\in D}$ is a family of probability measures on $(\Omega,\F)$ and \prozess[X] is a Markov process such that for every $\rx \in D$ it holds $X_0=\rx$, $\P_\rx$-almost surely. 
Denote by $\E_\rx$ the expectation w.r.t. the measure $\P_\rx$ and by $\langle \cdot,\cdot \rangle$ the inner product in $\R^d$. 
Setting
\[
\mathcal{I}_T
 = \left\{ u\in\R^d \colon \E_\rx \big[\e^{\scal{u}{X_T}}\big]<\infty, 
   \;\mbox{for all } \rx \in D\right\} ,
\]
we assume that
\begin{enumerate}[label=(\roman*)]
\item $0\in\mathcal{I}_{T}^{\circ}$, where $\mathcal I_T^\circ$ denotes the interior of $\mathcal I_T$;
\item the conditional moment generating function of $X_t$ under $\P_\rx$ has expo\-nentially-affine dependence on $\rx$; that is, there exist functions 
		$\phi\colon[0,T]\times\mathcal{I}_{T}\rightarrow\mathbb{R}$ and $\psi\colon[0,T]\times\mathcal{I}_{T}\rightarrow\R^d$ such that
      	\begin{equation}\label{eq:affine_property}
			\E_\rx \big[ \exp \langle u,X_t \rangle \big] = \exp \big( \phi_t(u) + \scal{\psi_t(u)}{\rx} \big),
      	\end{equation}
      for all $(t,u,x)\in[0,T]\times\mathcal{I}_T\times D$.
\end{enumerate}
\end{asa}

The functions $\phi$ and $\psi$ satisfy the \textit{semi-flow equations}, that is, for all $0\leq t+s\leq T$ and $u\in\mathcal{I}_{T}$
\begin{align}\label{eq:flows}
\begin{split}
\phi_{t+s}(u) &=
  \phi_t(u) + \phi_s(\psi_t(u)), \\ 
\psi_{t+s}(u) &= \psi_s(\psi_t(u)), 
\end{split}
\end{align}
with initial conditions
\[
\phi_0(u)=0 \quad\mbox{ and }\quad \psi_0(u)=u.
\]
Using the semi-flow equations we can derive the \emph{generalized Riccati equations} 
\begin{align}\label{Riccati}
\begin{split}
\frac{\partial}{\partial t}\phi_t(u)
 &= F(\psi_t(u)),  \qquad \phi_0(u)=0, \\ 
\frac{\partial}{\partial t}\psi_t(u)
 &= R(\psi_t(u)),  \qquad \psi_0(u)=u,    
\end{split}
\end{align}
for $(t,u)\in[0,T]\times\mathcal{I}_{T}$, where $F$ and $R=(R_1,\dots,R_d)$ are functions of L\'evy--Khintchine form:
\begin{align}\label{F-R-def}
\begin{split}
F(u) &= \langle b,u\rangle +
     \int_D\big(\e^{\langle\xi,u\rangle}-1\rangle\big)m(\ud \xi),\\
R_i(u) &= \langle \beta_i,u\rangle
       + \Big\langle\frac{\alpha_i}2u,u\Big\rangle
       + \int_D\big(\e^{\langle\xi,u\rangle}-1-\langle
          u,h_i(\xi)\rangle\big)\mu_i(\ud \xi),
\end{split}
\end{align}
while $(b,m,\alpha_i,\beta_i,\mu_i)_{1\leq i\leq d}$ are \textit{admissible parameters}---see Definition 2.6 in \cite{DuffieFilipovicSchachermayer03} for details---and $h_i\colon\Rp^d\rightarrow\R^d$ are suitable truncation functions. 
The infinitesimal generator of a process satisfying Assumption ($\mathbb A$) is provided by
\begin{align}\label{eq:generator}
\mathcal{A}f(x) 
 &= \left\langle b+\sum_{i=1}^{d}\beta_ix_i,\nabla f(x)\right\rangle
  + \frac1{2} \sum_{i=1}^{d} \alpha_{i,kl}x_i 
    \frac{\partial^{2}}{\partial x_k\partial x_{l}}f(x) \nonumber \\
 &\quad + \int_D \big(f(x+\xi)-f(x)\big) m(\ud \xi) \\
 &\quad + \sum_{i=1}^{d} \int_D
    \big( f(x+\xi) - f(x) -\langle h_i(\xi),\nabla f(x)\rangle \big)x_{i} 
    \mu_i(\ud \xi), \nonumber 
\end{align}
for all $f\in C_0^{2}(D)$ and $x\in D$.

Additional results are summarized in the following lemma. 
In the sequel, inequalities have to be understood componentwise, in the sense that $(a_1,a_2)\leq(b_1,b_2)$ if and only if $a_1\leq b_1$ and $a_2\leq b_2$.

\begin{lemma}\label{lem:phi_psi_prop}
The functions $\phi$ and $\psi$ satisfy the following:
\begin{enumerate}
\item $\phi_{t}(0)=\psi_{t}(0)=0$ for all $t\in[0,T]$.
\item $\mathcal{I}_{T}$ is a convex set; moreover for each $t\in[0,T]$, the 
      functions $u\mapsto\phi_{t}(u)$ and $u\mapsto\psi_{t}(u)$, for
      $u\in\mathcal{I}_{T}$, are (componentwise) convex.
\item $\phi_{t}(\cdot)$ and $\psi_{t}(\cdot)$ are order preserving: let
      $(t,u),\,(t,v)\in[0,T]\times\mathcal{I}_{T}$, with $u\leq v$. Then
      \begin{equation}
      \phi_{t}(u)\leq\phi_{t}(v)
      \quad \text{ and } \quad
      \psi_{t}(u)\leq\psi_{t}(v).\label{eq:phi,psi order pres.}
      \end{equation}
\item $\psi_{t}(\cdot)$ is strictly order-preserving: let 
      $(t,u),\,(t,v)\in[0,T]\times\mathcal{I}_{T}^{\circ}$, with $u<v$. Then 
      $\psi_{t}(u)<\psi_{t}(v)$.
\item $\phi$ and $\psi$ are jointly continuous on 
      $[0,T]\times\mathcal{I}_{T}^{\circ}$.
\item The partial derivatives 
      \[
	\frac{\partial}{\partial u_{i}}\phi_{t}(u)
	\quad \text{ and } \quad
	\frac{\partial}{\partial u_{i}}\psi_{t}(u)
	,\quad i=1,\dots,d
      \]
      exist and are continuous for 
      $\left(t,u\right)\in[0,T]\times\mathcal{I}_{T}^{\circ}$.
\end{enumerate}
\end{lemma}

\begin{proof}
See \citet[Lem. 4.2]{KellerResselPapapantoleonTeichmann09} for statements (1)--(4) and \citet[Prop. 3.16 and Lem. 3.17]{KellerRessel08} for the last two.\end{proof}

Affine processes have rich structural properties which have been proved particularly useful when it comes to financial modeling. 
However, there are situations where the condition of time-homogeneity cannot be met; for example, time-inhomogeneity may be introduced through an equivalent change of measure. 
\cite{Filipovic05} introduced time-inhomogeneous affine processes, whose conditional moment generating function takes the form
\[
\E_\rx \big[ \exp \langle u,X_t \rangle \big|\F_s \big] 
  = \exp \big( \phi_{s,t}(u)+ \langle \psi_{s,t}(u),X_{s} \rangle \big),
\]
for all $0\leq s\leq t\leq T$ and $u\in\mathcal{I}_{T}$. Theorem 2.7 in \cite{Filipovic05} yields that the infinitesimal generator is provided by
\begin{align}\label{eq:inhom_generator}
\mathcal{A}_t \e^{\langle u,\rx \rangle }
 &= \left. - \frac{\partial}{\partial s^{-}} \E\left[\left. 
    \e^{\scal{u}{X_t}}\right|X_{s}=\rx \right]\right|_{s=t} \nonumber \\
 &= \left(F(t,u)+\left\langle R(t,u),\rx\right\rangle \right)
     \e^{\left\langle u,\rx\right\rangle },
\end{align}
where the functions $F$ and $R$ retain the same form as in the time-homogeneous case, however the (admissible) parameters are now time-dependent. 
If the process $X$ is \textit{strongly regular affine}---that is, the parameters satisfy some continuity conditions, see Definition 2.9 in \cite{Filipovic05} for more details---then $\phi_{s,t}(u)$ and $\psi_{s,t}(u)$ satisfy generalized Riccati equations with time-dependent functional characteristics $F(s,u)$ and $R(s,u)$, i.e. for all $0\leq s\leq t\leq T$
\begin{align}\label{eq:time_inhomgeneous_riccati}
\begin{split}
-\frac{\partial}{\partial s}\phi_{s,t}(u) &=  
  F\left(s,\psi_{s,t}(u)\right),\quad\phi_{t,t}(u)=0,\\
-\frac{\partial}{\partial s}\psi_{s,t}(u) &= 
  R\left(s,\psi_{s,t}(u)\right),\quad\psi_{t,t}(u)=u.
\end{split}
\end{align}

\section{Affine LIBOR models with multiple curves}
\label{sec:3}

\subsection{A multiple curve setting}
\label{sub:A-multiple-curve-setting}

We start by introducing some basic notation and the main concepts used in multiple curve LIBOR models, following the approach introduced by \citet{Mercurio_2010}; see also \citet{Grbac_Papapantoleon_Schoenmakers_Skovmand_2014} for an overview and more details. 

The emergence of significant spreads between the OIS and LIBOR rates which depend on the investment horizon, also called tenor, means that we cannot work with a single tenor structure any longer. 
Let $\mathcal{T}=\{0=T_0<T_1<\cdots<T_N = T\}$ denote a discrete, equidistant time structure where $T_k$, for $k\in\mathcal{K}=\{1,\dots,N\}$, denote the relevant market dates, e.g. payment dates and maturities of traded instruments. 
The set of tenors is denoted by $\mathcal{X}=\{x_1,\dots,x_n\}$, where we typically have $\mathcal{X}=\{1,3,6,12\}$ months. 
Then, for every $x\in\mathcal{X}$ we consider the corresponding tenor structure $\mathcal{T}^x=\{0=T_0^x<T_1^x<\cdots<T_{N^x}^x=T_N\}$ with constant tenor length $\delta_x=T_k^x-T_{k-1}^x$. 
We denote by $\mathcal{K}^x=\{1,\dots,N^x\}$ the collection of all subscripts related to the tenor structure $\mathcal{T}^x$, and assume that $\mathcal{T}^x\subseteq\mathcal{T}$ for all $x\in\mathcal{X}$.

The Overnight Indexed Swap (OIS) rate is regarded as the best market proxy for the risk-free interest rate.
Moreover, the majority of traded interest rate derivatives are nowadays collateralized and the remuneration of the collateral is based on the overnight rate. 
Therefore the discount factors $B(0,T)$ are assumed to be stripped from OIS rates and defined for every possible maturity $T\in\mathcal{T}$; see also \citet[\S1.3.1]{Grbac_Runggaldier_2015}.
$B(t,T)$ denotes the discount factor, i.e. the time-$t$ price of a zero coupon bond with maturity \emph{T}, which is assumed to coincide with the corresponding OIS-based zero coupon bond.

Let $(\Omega,\F,\mathbb F,\P_N)$ be a complete stochastic basis, where $\P_N$ denotes the terminal forward measure, i.e. the martingale measure associated to the numeraire $B(\cdot,T_N)$. 
We consider the forward measures $(\P_k^x)_{x,k}$ associated to the numeraires $\{B(\cdot,T_k^x)\}_{x,k}$ for every pair $(x,k)$ with $x\in\mathcal{X}$ and $k\in\mathcal{K}^x$. 
Assuming that the processes $B(\cdot,T_k^x)/B(\cdot,T_N)$ are true $\P_N$-martingales for every pair $(x,k)$, the forward measures $\P_k^x$ are absolutely continuous with respect to $\P_N$ and defined in the usual way, i.e. via the Radon-Nikodym density 
\[
\frac{\ud\P_k^x}{\ud\P_N} 
  = \frac{B(0,T_N)}{B(0,T_k^x)} \cdot \frac1{B(T_k^x,T_N)}.
\]
Therefore, the forward measures are associated to each other via
\begin{equation}\label{eq:forward_measure_dens}
\frac{\ud\P_k^x}{\ud\P_{k+1}^x}\Big|_{\F_t}
  = \frac{B(0,T_{k+1}^x)}{B(0,T_k^x)} \cdot \frac{B(t,T_k^x)}{B(t,T_{k+1}^x)},
\end{equation}
hence they are related to the terminal measure via 
\begin{equation}\label{eq:forward_measure_terinal_dens}
\frac{\ud\P_k^x}{\ud\P_N} \Big|_{\F_t}
  = \frac{B(0,T_N)}{B(0,T_k^x)} \cdot \frac{B(t,T_k^x)}{B(t,T_N)}.
\end{equation}
The expectations with respect to the forward measures $\P_k^x$ and the terminal measure $\P_N$ are denoted by $\E_k^x$ and $\E_N$ respectively.

Next, we define the main modeling objects in the multiple curve LIBOR setting: the OIS forward rate, the forward LIBOR rate and the corresponding spread. 

\begin{definition}\label{def:OIS-rate}
The time-$t$ \emph{OIS forward rate} for the time interval $[T_{k-1}^x,T_k^x]$ is defined by 
\begin{equation}\label{eq:def-OIS-forward-rate}
F_k^x(t) := \frac1{\delta_x} \left( \frac{B(t,T_{k-1}^x)}{B(t,T_k^x)}-1 \right).
\end{equation}
\end{definition}

\begin{definition}\label{def:FRA-rate}
The time-$t$ \emph{forward LIBOR rate} for the time interval $[T_{k-1}^x,T_k^x]$ is defined by 
\begin{equation}\label{eq:def-FRA-rate}
L_k^x(t) := \E_k^x \left[ L(T_{k-1}^x,T_k^x)|\F_t \right],
\end{equation}
where $L(T_{k-1}^x,T_k^x)$ denotes the spot LIBOR rate at time $T_{k-1}^x$ for the time interval $[T_{k-1}^x,T_k^x]$.
\end{definition}

\noindent The forward LIBOR rate is the rate implied by a forward rate agreement where the future spot LIBOR rate is exchanged for a fixed rate; cf. \citet[pp.~12-13]{Mercurio_2009a}. 
The spot LIBOR rate $L(T_{k-1}^x,T_k^x)$ is set in advance, hence it is $\F_{T_{k-1}^x}$-measurable, therefore we have that the forward LIBOR rate coincides with the spot LIBOR rate at the corresponding tenor dates, i.e.
\[
L_k^x(T_{k-1}^x) 
= \E_k^x \big[ L(T_{k-1}^x,T_k^x)|\F_{T_{k-1}^x}\big]
= L(T_{k-1}^x,T_k^x).
\]

\begin{definition}\label{def:spread}
The \emph{(additive) spread} between the forward LIBOR rate and the OIS forward rate is defined by
\[
S_k^x(t) := L_k^x(t) - F_k^x(t).
\]
\end{definition}

\begin{remark}
In a single curve setup, the forward LIBOR rate is defined via \eqref{eq:def-OIS-forward-rate} and the spread equals zero for all times. 
However, in a multiple curve model these rates are not equal any more and we typically have that $L_k^x \geq F_k^x$. $F_k^x$ and $L_k^x$ can also be interpreted as forward rates corresponding to a riskless and a risky bond respectively; see e.g. \citet{Crepey_Grbac_Nguyen_2011}.
\end{remark}

\subsection{Affine LIBOR models with multiple curves}

We turn now our attention to the affine LIBOR models developed by \citet{KellerResselPapapantoleonTeichmann09} and extended to the multiple curve setting by \citet{Grbac_Papapantoleon_Schoenmakers_Skovmand_2014}. 
An important ingredient are martingales that are greater than, or equal to, one. 
Consider a process $X$ satisfying Assumption ($\mathbb{A}$) and starting at the canonical value $\boldsymbol1=(1,1,\dots,1)$, and let $u\in\mathcal{I}_T$. 
Then, the process \prozess[M^u] defined by 
\begin{align}\label{eq:def-M^u}
M_t^u
:= \E_{\boldsymbol1} \big[ \e^{\langle u,X_T\rangle}\big|\F_t \big]
 = \exp\big( \phi_{T-t}(u)+\langle\psi_{T-t}(u),X_t\rangle \big)
\end{align}
is a martingale. 
Moreover, if $u\in\I_T\cap\Rp^d$ the mapping $u \mapsto M_t^u$ is increasing and $M^u_t\ge1$ for every $t\in[0,T]$; see \citet[Thm~5.1]{KellerResselPapapantoleonTeichmann09} and \cite{Papapantoleon10b}.

The multiple curve \alms are defined as follows:

\begin{definition}
A multiple curve \alm $(X,\mathcal{X},T_N,u,v)$ consists of the following elements:
\begin{itemize}
\item An affine process $X$ under $\P_N$ satisfying Assumption ($\mathbb{A}$) and starting at the canonical value $\1$.
\item A set of tenors $\mathcal{X}$.
\item A terminal maturity $T_N$.
\item A sequence of vectors $u=(u_1,\dots,u_N)$ with $u_l=:u_k^x\in\I_T\cap\Rp^d$, for all $l=kT_1^x/T_1$ and $x\in\mathcal{X}$, such that
      \begin{align}\label{eq:def-ineq-u}
       u_1 \ge u_2 \ge \cdots \ge u_N=0.
      \end{align}
\item A collection of sequences of vectors $v=\set{(v_1^x,\dots,v^x_{N^x})}_{x\in\mathcal{X}}$ with $v_k^x\in\I_T\cap\Rp^d$, such that
      \begin{align}\label{eq:def-ineq-v}
       v_k^x \ge u_k^x \quad \text{for all } k\in\mathcal{K}^x,x\in\mathcal{X}.
      \end{align}
\end{itemize}
The dynamics of the OIS forward rates and the forward \lib rates in the model evolve according to 
\begin{equation}\label{eq:model-OIS-FRA-rate}
1 + \delta_x F_k^x(t) = \frac{M_t^{u_{k-1}^x}}{M_t^{u_k^x}}
 \quad\mbox{ and }\quad
1 + \delta_x L_k^x(t) = \frac{M_t^{v_{k-1}^x}}{M_t^{u_k^x}},
\end{equation}
for all $t\in[0,T_k^x]$, $k\in\mathcal{K}^x$ and $x\in\mathcal{X}$.
\end{definition}

The definition of multiple curve \alms implies that the dynamics of OIS forward rates and forward LIBOR rates, more precisely of $1+\delta_xF_k^x$ and $1+\delta_xL_k^x$, exhibit an exponential-affine dependence in the driving process $X$; see \eqref{eq:def-M^u} and \eqref{eq:model-OIS-FRA-rate}.
\citet{Glau_Grbac_Papapantoleon_2016} recently showed that models that exhibit this exponential-affine dependence are the only ones that produce structure preserving LIBOR models; cf. Proposition 3.11 therein.
The denominators in \eqref{eq:model-OIS-FRA-rate} are the same in both cases, since both rates have to be $\P_k^x$-martingales by definition.
On the other hand, different sequences $(u_l)_{l\in\mathcal{K}}$ and $(v_k^x)_{k\in\mathcal{K}^x}$ are used in the numerators in \eqref{eq:model-OIS-FRA-rate} producing different dynamics for OIS and LIBOR rates.
These sequences are used to fit the multiple curve \alm to a given initial term structure of OIS and LIBOR rates.
In particular, the subsequent propositions show that by fitting the model to the initial term structure we obtain automatically sequences $(u_l)_{l\in\mathcal{K}}$ and $(v_k^x)_{k\in\mathcal{K}^x}$ that satisfy \eqref{eq:def-ineq-u} and \eqref{eq:def-ineq-v} respectively; see also \citet[Rem. 4.4 and 4.5]{Grbac_Papapantoleon_Schoenmakers_Skovmand_2014} for further comments on these sequences.

The following quantity measures the ability of a multiple curve affine LIBOR model to fit a given initial term structure
\begin{equation}\label{eq:gamma-X}
\gamma_{X}
 := \underset{u\in\I_T\cap\R_{>0}^d} \sup \E_{\boldsymbol1} 
    \big[ \exp\langle u,X_{T}\rangle \big].
\end{equation}
In several models commonly used in mathematical finance, such as the Cox--Ingersoll--Ross model and Ornstein--Uhlenbeck processes driven by subordinators, this quantity equals infinity.
The following propositions show that the \alms are well-defined and can fit any initial term structure under mild conditions. 

\begin{proposition}\label{prop:Fitting-OIS} 
Consider the time structure $\mathcal{T}$, let $B(0,T_l)$, $l\in\mathcal{K}$, be the initial term structure of OIS discount factors and assume that 
\begin{equation}\label{eq:initial-term-struct-ordering}
B(0,T_1)\geq\cdots\geq B(0,T_N)>0.
\end{equation}
Then the following hold:
\begin{enumerate}
\item If $\gamma_{X}>B(0,T_1)/B(0,T_N)$, there exists a sequence $(u_l)_{l\in\mathcal{K}}$ in $\I_T\cap\Rp^d$ satisfying \eqref{eq:def-ineq-u} such that
      \[
		M_0^{u_l} = \frac{B(0,T_l)}{B(0,T_N)}\quad\text{for all }l\in\mathcal{K}.
      \]
      In particular, if $\gamma_{X}=\infty$, then the multiple curve affine LIBOR model can fit any initial term structure of OIS rates.
\item If X is one-dimensional, the sequence $(u_l)_{l\in\mathcal{K}}$ is unique.
\item If all initial OIS rates are positive, the sequence $(u_l)_{l\in\mathcal{K}}$ is strictly decreasing.
\end{enumerate}
\end{proposition}

\begin{proof}
See Proposition 6.1 in \citet{KellerResselPapapantoleonTeichmann09}.
\end{proof}

\begin{proposition}\label{prop:Fitting-FRA} 
Consider the setting of the previous proposition, fix $x\in\mathcal{X}$ and the corresponding tenor structure $\mathcal{T}^x$. Let $L_k^x(0)$, $k\in\mathcal{K}^x$, be the initial term structure of non-negative forward LIBOR rates and assume that for every $k\in\mathcal{K}^x$
\begin{equation}\label{eq:FRA-OIS-spread}
L_k^x(0) 
  \geq \frac1{\delta_x}\left(\frac{B(0,T_{k-1}^x)}{B(0,T_k^x)}-1\right)
  = F_k^x\left(0\right).
\end{equation}
Then, the following hold:
\begin{enumerate}
\item If $\gamma_{X}>(1+\delta_xL_k^x(0))B(0,T_k^x)/B(0,T_N)$ for all $k\in\mathcal{K}^x$, then there exists a sequence $(v_k^x)_{k\in\mathcal{K}^x}$ in $\I_T\cap\Rp^d$ satisfying \eqref{eq:def-ineq-v} such that
      \[
	M_0^{v_k^x} = \left(1+\delta_xL_{k+1}^x(0)\right) M_0^{u_{k+1}^x},
	\quad \text{for all } k\in\mathcal{K}^x\setminus\{N^x\}.
      \]
      In particular, if $\gamma_{X}=\infty$, then the multiple curve affine LIBOR model can fit any initial term structure of forward LIBOR rates.
\item If X is one-dimensional, the sequence $(v_k^x)_{k\in\mathcal{K}^x}$ is unique.
\item If all initial LIBOR-OIS spreads are positive (i.e. \eqref{eq:FRA-OIS-spread} becomes strict), then $v_k^x>u_k^x$, for all $k\in\mathcal{K}^x\setminus\{N^x\}$.
\end{enumerate}
\end{proposition}

\begin{proof}
See Proposition 4.2 in \cite{Grbac_Papapantoleon_Schoenmakers_Skovmand_2014}.  
\end{proof}

\begin{remark}\label{rem:u_k_not_unique}
The proofs of these propositions are constructive and provide an easy algorithm for fitting an affine LIBOR model to a given initial term structure of OIS and \lib rates. 
However, for $d>1$ the sequences $u$ and $v^x$ are not unique, hence questions about optimality arise; see the discussion in subsections \ref{sec:on-interpolating-function} and \ref{discussion}. 
In the proof of Proposition \ref{prop:Fitting-OIS}, the sequence $(u_l)_{l\in\mathcal{K}}$ is chosen along a straight line in $\mathcal{I}_T\cap\R^d_{\geqslant0}$ from some $u\in\mathcal{I}_T$ to $0$, such that $u$ satisfies
\begin{align}\label{eq:cond-fit}
 M_{0}^{u}
  = \mathds{E}_{\boldsymbol{1}} \left[ \exp \left \langle u,X_{T} \right 
    \rangle \right]
  > \gamma_{X}-\varepsilon
  >\frac{B\left(0,T_{1}\right)}{B\left(0,T_{N}\right)}.
\end{align}
However, any other continuous path from another $u'$ to $0$, that satisfies \eqref{eq:cond-fit} and is componentwise decreasing, would have worked as well. 
\end{remark}

The next proposition shows that multiple curve affine LIBOR models are analytically tractable, in the sense that the affine structure is preserved under any forward measure. 

\begin{proposition}\label{prop:XistimeinhomunderPxk} 
The underlying process $X$ is a time-inhomogeneous affine process under the measure $\P_k^x$, for every $x\in\mathcal{X}$ and $k\in\mathcal{K}^x$. 
The moment generating function is provided by
\[
\E_{x,\rx}^k \big[ \exp \langle w,X_{t} \rangle \big]
  = \exp \Big(\phi_t^{x,k}(w)+\big\langle \psi_t^{x,k}(w),\rx\big\rangle \Big),
\]
for every $w$ such that $w+\psi_{T_N-t}(u_k^x)\in\mathcal{I}_{T}$, where 
\begin{eqnarray*}
\phi_{t}^{x,k}(w) & = 
	\phi_{t}\left(\psi_{T_N-t}(u_k^x)+w\right)-\phi_{t}\left(\psi_{T_N-t}(u_{k}^x)\right),\\
\psi_{t}^{x,k}(w) & = 
	\psi_{t}\left(\psi_{T_N-t}(u_k^x)+w\right)-\psi_{t}\left(\psi_{T_N-t}(u_{k}^x)\right).
\end{eqnarray*}
\end{proposition}

\begin{proof}
See Proposition 4.6 in \cite{Grbac_Papapantoleon_Schoenmakers_Skovmand_2014}. 
\end{proof}

The multiple curve \alms defined above and satisfying the prerequisites of Propositions \ref{prop:Fitting-OIS} and \ref{prop:Fitting-FRA} are arbitrage-free discrete tenor models, in the sense that
\begin{center}
$F_k^x$ and $L_k^x$ are $\P_k^x$-martingales 
\end{center}
for every $k\in\mathcal{K}^x$, $x\in\mathcal{X}$, while the interest rates and the spread are positive, i.e.
\begin{center}
$F_k^x(t)\ge0$ and $S_k^x(t)=L_k^x(t)-F_k^x(t)\ge0$ 
\end{center}
for every $t\in[0,T_{k-1}^x], k\in\mathcal{K}^x$, $x\in\mathcal{X}$; cf. Proposition 4.3 in \citet{Grbac_Papapantoleon_Schoenmakers_Skovmand_2014}.

\begin{remark}
The class of \alms with multiple curves can be extended to accomodate negative interest rates alongside positive spreads; see \citet[\S4.1]{Grbac_Papapantoleon_Schoenmakers_Skovmand_2014} for the details.
\end{remark}

\begin{remark}
We could use time-dependent parameters, i.e. time-inhomogeneous affine processes, in the construction of \alms, in particular since the dynamics of $X$ are time-dependent under forward measures; see Proposition \ref{prop:XistimeinhomunderPxk}.
We use affine processes instead, in order to ease the presentation of the model and its properties, and to be consistent with the relevant literature (cf. \citealt{KellerResselPapapantoleonTeichmann09} and \citealt{Grbac_Papapantoleon_Schoenmakers_Skovmand_2014}).
\end{remark}

\section{Continuous tenor extension of affine LIBOR models}
\label{sec:4}

\subsection{Discrete to continuous tenor}

This section is devoted to the extension of the affine LIBOR models from a discrete to a continuous tenor structure, and the derivation of the dynamics of the corresponding instantaneous forward rate and short rate.
The main tool is an interpolating function $\mathpzc{U}$, which is a function defined on $[0,T_N]$ that matches $u_l$ at each tenor date $T_l$. 
This subsection follows and extends \citet{KellerRessel_2009}.

\begin{definition}\label{def:interpol-fun}
An \textit{interpolating function} for the multiple curve affine LIBOR model $(X,\mathcal{X},T_N,u,v)$ is a continuous, componentwise decreasing function 
$\mathpzc{U}:[0,T_N] \to \Rp^d$ with $\mathpzc U(t)\in\mathcal I_T$ for all $t\in[0,T_N]$ and bounded right-hand derivatives, such that $\mathpzc{U}(T_l)=u_l$ for all $T_l\in\mathcal{T}$.
\end{definition}

\begin{remark}\label{rem:interpol-0}
Since $\mathpzc{U}$ is a mapping from $[0,T_N]$, it makes sense to define a $\mathpzc U_0$ element. 
This can be chosen such that $M_0^{\mathpzc{U}(0)}=\frac{1}{B(0,T_N)}$, which is consistent with Proposition \ref{prop:Fitting-OIS}.
\end{remark}

The interpolating function allows to derive an explicit expression for the dynamics of zero coupon bond prices in the multiple curve affine LIBOR model. 
In particular, they belong to the class of \textit{affine term structure models}; see e.g. \citet[\S24.3]{Bjoerk09}.

\begin{lemma}\label{lem:ATSM}
Let $\mathpzc{U}$ be an interpolating function for the multiple curve affine LIBOR model $(X,\mathcal{X},T_N,u,v)$ and define the (OIS zero coupon) bond price 
$B(t,\mathpzc{T})$ by 
\begin{equation}\label{eq:def-continuous-tenor}
\frac{B(t,\mathpzc{T})}{B(t,T_N)} = M_t^{\mathpzc{U}(\mathpzc{T})},
\end{equation}
for $0\le t\le \mathpzc{T}\le T_N$, where $M^\mathpzc{U(T)}$ is the martingale defined by \eqref{eq:def-M^u}. 
Then, bond prices satisfy 
\begin{equation}\label{eq:Bond(t,tau)_equals_exp}
B(t,\mathpzc{T}) = \exp\big( \alpha(t,\mathpzc{T})
		 + \langle \beta(t,\mathpzc{T}),X_t\rangle \big), 
\end{equation}
where
\begin{align}
\begin{split}\label{eq:bond-AB}
\alpha(t,\mathpzc{T}) &= 
  \phi_{T_N-t}(\mathpzc{U}(\mathpzc{T}))-\phi_{T_N-t}(\mathpzc{U}(t)),\\
\beta(t,\mathpzc{T}) &= 
  \psi_{T_N-t}(\mathpzc{U}(\mathpzc{T}))-\psi_{T_N-t}(\mathpzc{U}(t)).
\end{split}
\end{align}
\end{lemma}

\begin{proof}
Using the definition of the OIS forward rate, \eqref{eq:model-OIS-FRA-rate} and the positivity of bond prices, we get in the discrete tenor case, 
\[
B(T_k,T_i)
  = \prod_{l=k}^{i-1} \frac{B(T_k,T_{l+1})}{B(T_k,T_{l})}
  = \prod_{l=k}^{i-1} \frac{M_{T_k}^{u_{l+1}}}{M_{T_{k}}^{u_{l}}}
  = \frac{M_{T_k}^{u_{i}}}{M_{T_k}^{u_k}}
\]
for every $T_k,T_i\in\mathcal{T}$ such that $T_k\leq T_i\leq T_N$. 
Similarly in the continuous tenor case, using \eqref{eq:def-continuous-tenor} we get for $0\leq t\leq \mathpzc{T}\leq T_N$ that
\begin{align}\label{eq:bonds_for_all_maturities2}
B(t,\mathpzc{T}) \nonumber
 &= \frac{B\left(t,T_{\lfloor t \rfloor+1}\right)}{B\left(t,t\right)} 
    \left(\prod_{l=\lfloor t \rfloor+1}^{\lfloor \mathpzc{T} \rfloor-1}
    \frac{B(t,T_{l+1})}{B(t,T_l)}\right) 
    \frac{B(t,\mathpzc{T})}{B\left(t,T_{\lfloor \mathpzc{T} \rfloor}\right)} \\
 &= \frac{M_t^{u_{\lfloor t \rfloor+1}}}{M_t^{\mathpzc{U}(t)}}
    \left(\prod_{l=\lfloor t \rfloor+1}^{\lfloor \mathpzc{T} \rfloor-1} 
    \frac{M_t^{u_{l+1}}}{M_{t}^{u_{l}}}\right) 
    \frac{M_t^{\mathpzc{U}(\mathpzc{T})}}{M_t^{u_{\lfloor \mathpzc{T} \rfloor-1}}}
  = \frac{M_t^{\mathpzc{U}(\mathpzc{T})}}{M_t^{\mathpzc{U}(t)}},
\end{align}
where $\lfloor t \rfloor$ is such that $T_{\lfloor t \rfloor}$ is the largest element in the time structure $\mathcal{T}$ less than or equal to $t$. 
Hence, since $M_t^{\mathpzc{U(T)}}$ depends exponentially-affine on $X_{t}$, we arrive at \eqref{eq:Bond(t,tau)_equals_exp}--\eqref{eq:bond-AB}.
\end{proof}

Next, we will show that the extension of an affine LIBOR model from a discrete to a continuous tenor is an arbitrage-free term structure model. 
Following \citet[Def. 2.3]{MusielaRutkowski97b}, we say that a family of bond prices $(B\brac{t,\mathpzc{T}})_{0\leq t\leq \mathpzc{T}\leq T_N}$ satisfies a \emph{no-arbitrage condition} if there exists a measure $\mathds Q$ such that $B\brac{\cdot,\mathpzc{T}}/B\brac{\cdot,T_N}$ is a $\mathds Q$-local martingale and $B\brac{t,\mathpzc{T}}\leq 1$, for any $0\leq t\leq \mathpzc{T}\leq T_N$.

\begin{theorem}\label{thm:continuous-tenor}
Let $\mathpzc{U}$ be an interpolating function for the multiple curve affine LIBOR model $(X,\mathcal{X},T_N,u,v)$. 
Then $(B(t,\mathpzc{T}))_{0\leq t\leq\mathpzc{T}\leq T_N}$ is a \emph{continuous tenor extension of the affine LIBOR model}, i.e. an arbitrage-free model for all maturities $\mathpzc{T}\in[0,T_N]$, such that for all maturities $T\in\mathcal{T}$ the bond prices coincide with those of the (discrete tenor) affine LIBOR model. 
\end{theorem}

\begin{proof}
The definition of the interpolating function yields immediately that bond prices in the continuous tenor extension coincide with those of the discrete tenor affine LIBOR model for all maturities.

According to \citet[\S2.3]{MusielaRutkowski97b}, in order to show that the model is arbitrage-free it suffices to verify the following conditions on the family $(B(t,\mathpzc{T}))_{0\leq t\leq \mathpzc{T}\leq T_N}$ of bond prices:
\begin{enumerate}[label=(\roman*)]
\item $B(\cdot,\mathpzc{T})$ is a strictly positive special semimartingale and the left-hand limit process $B\left(\cdot-,\mathpzc{T}\right)$ is also strictly 			positive, for every $\mathpzc{T}\in\left[0,T_N\right].$
\item The bond price quotients $B(\cdot,\mathpzc{T})/B(\cdot,T_N)$ are $\P_N$-martingales.
\item $B(t,S)\leq B(t,U)$ for all $0\leq t\leq S\leq U\leq T_N$. 
\end{enumerate}

The second condition follows immediately from \eqref{eq:def-continuous-tenor} and the construction of $M^u$ as a $\P_N$-martingale. 
In order to check the first and the third conditions, we shall use the representation for the bond prices from Lemma \ref{lem:ATSM}. 
Indeed, the last condition follows directly from representation \eqref{eq:Bond(t,tau)_equals_exp}--\eqref{eq:bond-AB}, using the monotonicity of the function $\mathpzc{U}$ and the order preserving property of $\phi$ and $\psi$; cf. Lemma \ref{lem:phi_psi_prop}. Moreover, the continuity of $\phi$ and $\psi$ together with \eqref{eq:Bond(t,tau)_equals_exp} imply that $B(t-,\mathpzc{T}) = \exp(\alpha(t,\mathpzc{T}) + \langle\beta(t,\mathpzc{T}),X_{t-}\rangle)$, which ensures the positivity of $B(\cdot-,\mathpzc{T})$.

Finally, $B(\cdot,\mathpzc{T})$ is a smooth function of $X$ hence it is also a semimartingale, which is special if its associated jump process $\Delta 
B(t,\mathpzc{T})=B(t,\mathpzc{T})-B\left(t-,\mathpzc{T}\right)$ is absolutely bounded; cf. \citet[Lemma I.4.24]{JacodShiryaev03}. 
The processes $X$ and $X_{-}$ are non-negative a.s. and the same is true for $\Delta X$, since the compensator of the jump measure of $X$ is entirely supported on the positive half-space; cf. \citet[Definition~2.6]{DuffieFilipovicSchachermayer03}. 
Using again Lemma \ref{lem:phi_psi_prop}(4) and the monotonicity of $\mathpzc{U}$, we get that $\alpha(t,\mathpzc{T})$ and $\beta(t,\mathpzc{T})$ take values in the negative half space. 
Thus we can estimate the jump process $\Delta B(t,\mathpzc{T})$ as follows:
\begin{align*}
\left| \Delta B(t,\mathpzc{T}) \right| 
 &= \left|\e^{\alpha(t,\mathpzc{T})
	+ \left\langle \beta(t,\mathpzc{T}),X_{t}\right\rangle} 
   - \e^{\alpha(t,\mathpzc{T}) 
	+ \left\langle \beta(t,\mathpzc{T}),X_{t-}\right\rangle }\right|\\
 &= \e^{ \alpha(t,\mathpzc{T}) + \left\langle \beta(t,\mathpzc{T}),X_{t-}\right\rangle } 
    \left|\e^{\left\langle \beta(t,\mathpzc{T}),\Delta X_{t}\right\rangle}-1\right| 
    \le 1. \qedhere
\end{align*}
\end{proof}

Having bond prices for all maturities $\mathpzc{T}\in[0,T_N]$ at hand, we can now calculate the dynamics of the instantaneous forward rate $f(t,\mathpzc{T})$ 
with maturity $\mathpzc{T}$ prevailing at time $t$ and of the short rate $r_{t}$ prevailing at time $t$. 
These quantities are commonly defined as
\[
f(t,\mathpzc{T})
 = -\frac{\partial_{+}\log B(t,\mathpzc{T})}{\partial \mathpzc{T}}
 \quad\mbox{ and }\quad
r_{t}=f(t,t).
\]
Together with the requirement that $B(\mathpzc{T},\mathpzc{T})=1$, we get that the former is equivalent to 
\begin{equation}\label{eq:bonds_for_all_maturities}
B(t,\mathpzc{T}) = \exp\Bigg(-\int_{t}^{\mathpzc{T}}f(t,s)\ds\Bigg).
\end{equation}

\begin{lemma}\label{prop:Instantaneous-forward-and-short-rate}
Let $\mathpzc U$ be an interpolating function and consider the continuous tenor extension of the affine LIBOR model \afflm. 
Then, the instantaneous forward rate and the short rate are provided by 
\begin{align}\label{eq:short-rate}
f(t,\mathpzc{T}) 
  = p(t,\mathpzc{T}) + \left\langle q(t,\mathpzc{T}),X_{t}\right\rangle
\quad \text{ and } \quad 
r_{t} = p_{t} + \left\langle q_{t},X_{t}\right\rangle,
\end{align}
where 
\begin{align}\label{eq:p-q}
\begin{split}
p(t,\mathpzc{T}) 
  &= - \scal{\left. \nabla_u\phi_{T_N-t}\brac{u} 
      \right|_{u=\mathpzc{U}(\mathpzc{T})}} 
      {\left. \frac{\ud \mathpzc{U}(t)}{\ud t+}\right|_{t=\tc}}, 
\quad p_t = p(t,t),      
\\
q(t,\mathpzc{T}) 
  &=-\left. \nabla_u \psi_{T_N-t}\brac{u} 
    \right|_{u=\mathpzc{U}(\mathpzc{T})}\circ 
    \left.\frac{\ud \mathpzc{U}(t)}{\ud t+}\right|_{t=\tc}, 
\quad q_t = q(t,t),
\end{split}
\end{align}
for all $0\le t\le \tc \le T_N$.
Here $a\circ b$ denotes the componentwise multiplication of two vectors $a$ and $b$ having the same dimension.
\end{lemma}

\begin{proof}
We know from Lemma \ref{lem:ATSM} that bond prices are log-affine functions of $X$, in particular they are provided by \eqref{eq:Bond(t,tau)_equals_exp}--\eqref{eq:bond-AB}. 
The result now follows by taking the right hand derivative of $\alpha(t,\mathpzc{T})+\left\langle\beta(t,\mathpzc{T}),X_t\right\rangle$ w.r.t. $\tc$, which exists by Definition \ref{def:interpol-fun} and Lemma \ref{lem:phi_psi_prop}(6). 
Note also that $p\brac{t,\tc}$ and $q\brac{t,\tc}$ are positive, for all $t\in[0,\tc]$.
\end{proof}

\noindent Moreover, the continuously compounded bank account $B^\star$ is 
defined as usual:
\begin{equation}\label{eq:Bank-account}
  B^\star = \exp\Bigg(\int_0^\cdot r_{s}\ds\Bigg),
\end{equation}
while the associated spot measure $\P_\star$, under which bond prices are 
provided by
\[
B(t,\mathpzc{T}) 
 = \E_\star \left[\frac{B_{t}^\star}{B_{\mathpzc{T}}^\star}\Big|\F_t\right],
\]
is calculated next.

\begin{lemma}\label{thm:spot_measure_density}
Let $\mathpzc U$ be an interpolating function and consider the continuous tenor extension of the affine LIBOR model \afflm. 
Then, the spot measure $\P_\star$ is determined by the density process
\[
\left. \frac{\ud\P_\star}{\ud\P_N} \right|_{\F_t} 
  = \frac{1}{M_0^{\mathpzc{U}(0)}}
    \exp\left( P_t + \left\langle Q_t,X_t\right\rangle 
	      + \int_0^t r_s\ds\right),
\]
where $P_t:=\phi_{T_N-t}(\mathpzc{U}(t))$ and $Q_t:=\psi_{T_N-t}(\mathpzc{U}(t))$.
\end{lemma}

\begin{proof}
The spot measure and the terminal forward measure are related via 
$$ \left. \frac{\ud\P_\star}{\ud\P_N} \right|_{\F_t} = \frac{B_t^\star B\left(0,T_N\right)}{B(t,T_N)}; $$
cf. \citet[\S13.2.2]{MusielaRutkowski05}.
Then, the representation above follows easily from \eqref{eq:Bank-account}, \eqref{eq:Bond(t,tau)_equals_exp}--\eqref{eq:bond-AB}, Lemma  \ref{lem:phi_psi_prop}(1) and Remark \ref{rem:interpol-0}, using the fact that $\mathpzc{U}(T_N)=0$, hence $B(t,T_N)=1/M_t^{\mathpzc{U}(t)}$. 
\end{proof}

The next result resembles Proposition \ref{prop:XistimeinhomunderPxk} and shows that the driving process $X$ remains an affine process under the spot measure 
$\P_\star$. 
In other words, the multiple curve \alm remains analytically tractable under the spot measure as well.

\begin{theorem}
Let $\mathpzc{U}$ be an interpolating function and consider the continuous tenor extension of the affine LIBOR model \afflm. 
Then the underlying process $X$ is a time-inhomogeneous affine process under the spot measure $\P_\star$. 
In particular, $X$ is strongly regular affine and the functional characteristics under $\P_\star$ are provided by
\begin{align*}
F^{\star}\left(t,w\right)=F\left(w+Q_t\right)-F\left(Q_t\right)
\quad \text{and} \quad
R^{\star}\left(t,w\right)=R\left(w+Q_t\right)-R\left(Q_t\right),
\end{align*}
for every $w$ such that $w+Q_t\in\mathcal{I}_{T}$.
\end{theorem}

\begin{proof}
We will first show that the moment generating function of $X$ has an exponential-affine form under $\P_\star$. 
Starting from the moment generating function of $X$ under $\P_\star$, and using the conditional density process in Lemma \ref{thm:spot_measure_density} and the dynamics of the short rate process in \eqref{eq:short-rate}--\eqref{eq:p-q}, we arrive at 
\begin{align}\label{eq:AB}
\E_{\star,\rx} \Big[ \e^{\scal{w}{X_t}} \big|\F_{s} \Big] 
  &= \E_{N,\rx} \left[\left. \e^{\scal{w}{X_t}}    
      \frac{B_{t}^{\star}\, B(s,T_N)}{B_{s}^{\star}\, B(t,T_N)}
      \right|\F_{s}\right] \nonumber \\ \nonumber
  &= \exp\left( P_t-P_s-\scal{Q_s}{X_s}
      + \int_{s}^{t}p_{u}\du - \int_0^{s}\scal{q_u}{X_u}\du\right)\\
  &\quad\times \E_{N,\rx} \left[\left. \exp\left( \scal{w+Q_t}{X_t}
      + \int_0^t \scal{q_u}{X_u}\du \right)\right|\F_s\right] \nonumber\\
  &=: A \times B.
\end{align}
Theorem 4.10 in \citet{KellerRessel08} provides an elegant way to calculate the functional characteristics of a time integrated affine process. 
This result is proved for $Y_\cdot=\int_0^\cdot X_u \du$ and is extended to $\widetilde{Y}_\cdot=\big(\int_0^\cdot \theta_u^iX_u^i \du\big)_{1\le i\le d}$ for a deterministic, bounded and positive $\theta$ in Theorem \ref{thm:time-integral} of the Appendix. 
Then, we have that the functional characteristics of the joint process $(X_t,\widetilde{Y}_t)$ are provided by
\[
\widetilde{F}(t,w_x,w_y) = F(w_x)
  \quad\text{and}\quad
\widetilde{R}(t,w_x,w_y) = \left( \begin{array}{c}
  R\left(w_x\right)+\theta_t \circ w_{y}\\0
  \end{array}\right).
\]
The definition of the interpolating function together with Lemma \ref{lem:phi_psi_prop}(6) yield that $q$ in \eqref{eq:p-q} is bounded.
Hence, applying Theorem \ref{thm:time-integral} yields that $B$ in \eqref{eq:AB} takes the form
\[
B = \exp\left( \widetilde{\phi}_{t-s}(w+Q_t,\mathbf1) 
  + \left\langle \widetilde{\psi}_{t-s}(w+Q_t,\mathbf1), (X_s,\widetilde{Y}_s) 
\right\rangle \right),
\]
where $\widetilde{\phi}$ and $\widetilde{\psi}$ are the solutions of the generalized Riccati equations defined by $\widetilde{F}$ and $\widetilde{R}$; cf. \eqref{Riccati}. 
The form of $\widetilde{R}$ implies that the components of $\widetilde{\psi}$ corresponding to $\widetilde{Y}$ satisfy $\widetilde{\psi}_{t-s}(w_x,w_{y})_{y}=w_{y}$. 
Hence, we get from \eqref{eq:AB} that
\begin{multline*}
\E_{\star,\rx} \Big[ \e^{\scal{w}{X_t}} \big|\F_{s} \Big] \\
  = \exp\Bigg( P_t - P_s - \scal{Q_s}{X_s} + \int_s^tp_u\du 
    + \widetilde{\phi}_{t-s}(w+Q_t,\mathbf1) 
  + \left\langle \widetilde{\psi}_{t-s}(w+Q_t,\mathbf1)_x,X_s\right\rangle 
    \Bigg).
\end{multline*}

Now, conditioning on $X_s=\rx$ and taking the right-hand derivatives with respect to $s$ at $t=s$, we arrive at the generator of $X$ under $\P_\star$:
\begin{align}\label{eq:Pstar-generator}
\mathcal{A}_t \e^{\left\langle w,x\right\rangle} 
  = \big( F(w+Q_t)-F(Q_t) + \langle R(w+Q_t)-R(Q_t),x\rangle 
    \big)\e^{\langle w,x\rangle} ;
\end{align}
compare with \eqref{eq:inhom_generator}. 
The semigroup of the affine process $X$ under $\P_\star$ is weakly regular in the sense of \citet[Def. 2.3]{Filipovic05}, since the process $X$ is stochastically continuous under $\P_\star$ and the generator exists and is continuous at $w=0$ for all $(t,x)\in[0,T]\times D$. 
Moreover, $X$ is strongly regular affine under $\P_\star$ since the weakly admissible parameters $(\alpha^\star(t),b^\star(t),\beta^\star(t), m^\star(t),\mu^\star(t))$ implied by \eqref{eq:Pstar-generator} are continuous transformations of $\brac{\alpha,b,\beta,m,\mu}$.
\end{proof}

Using the last proposition, we get that the conditional moment generating function of $X$ under $\P_\star$ is given by
\begin{align}\label{eq:Pstar-mgf}
\E_{\star,\rx} \big[ \exp\left\langle w,X_t \right\rangle \big|\F_s \big]
 = \exp\left( \phi_{s,t}^{\star}(w) + \left\langle 
	\psi_{s,t}^{\star}(w),X_{s}\right\rangle \right),
\end{align}
where $\phi_{s,t}^{\star}$ and $\psi_{s,t}^{\star}$ are the solutions of the generalized Riccati equations with functional characteristics $F^{\star}$ and $R^{\star}$; cf. \eqref{eq:time_inhomgeneous_riccati}. 
Since both the instantaneous forward rate and the short rate are time-dependent affine transformations of the driving process $X$, they will inherit many (distributional) properties from $X$. 
In fact, once we have computed the characteristics of the driving process $X$ under the spot measure $\P_\star$, it is easy to see that also the short rate 
$r$ has time-inhomogeneous characteristics, that are affine w.r.t. $X$. 
Indeed, from \eqref{eq:short-rate} and \eqref{eq:Pstar-mgf} we get 
\begin{align}
\E_{\star,\rx} \big[ \exp\left(wr_{t}\right)|\F_{s}\big] 
&= \exp\left( wp_{t} + \phi_{s,t}^{\star}\left(wq_{t}\right)
	+\left\langle \psi_{s,t}^{\star}\left(wq_{t}\right),X_{s} 
	 \right\rangle \right) \nonumber \\ 
&=: \exp\left( \phi_{s,t}^{r}(w) + \left\langle 
	\psi_{s,t}^{r}(w),X_{s}\right\rangle \right).
	\label{eq:short_notation_MGF_r_t}
\end{align}

\subsection{On the choice of the interpolating function}
\label{sec:on-interpolating-function}

The requirements on the interpolating function $\mathpzc U$ are rather weak, such that even a linear interpolation between the $u_k$'s corresponding to the 
maturities $T_k$, $k\in\mathcal K$, can be used. 
However looking at equations \eqref{eq:short-rate}--\eqref{eq:p-q} for the dynamics of the short rate process $r$, we can immediately observe that jumps will occur at fixed times, the maturities $T_k$, if $\mathpzc U$ is not continuously differentiable. 
A more sophisticated, but still arbitrary, choice for an interpolating function are cubic splines, i.e. piecewise polynomials of degree three, which are 
continuously differentiable and thus do not lead to deterministic discontinuities; see Figure \ref{fig:spline-vs-linear} for an illustration.  

The next corollary is an immediate consequence of Lemma \ref{prop:Instantaneous-forward-and-short-rate} and the fact that $X$ is stochastically continuous.

\begin{figure}
\begin{centering}
\parbox{0.49\textwidth}{\includegraphics[width=0.53\textwidth]{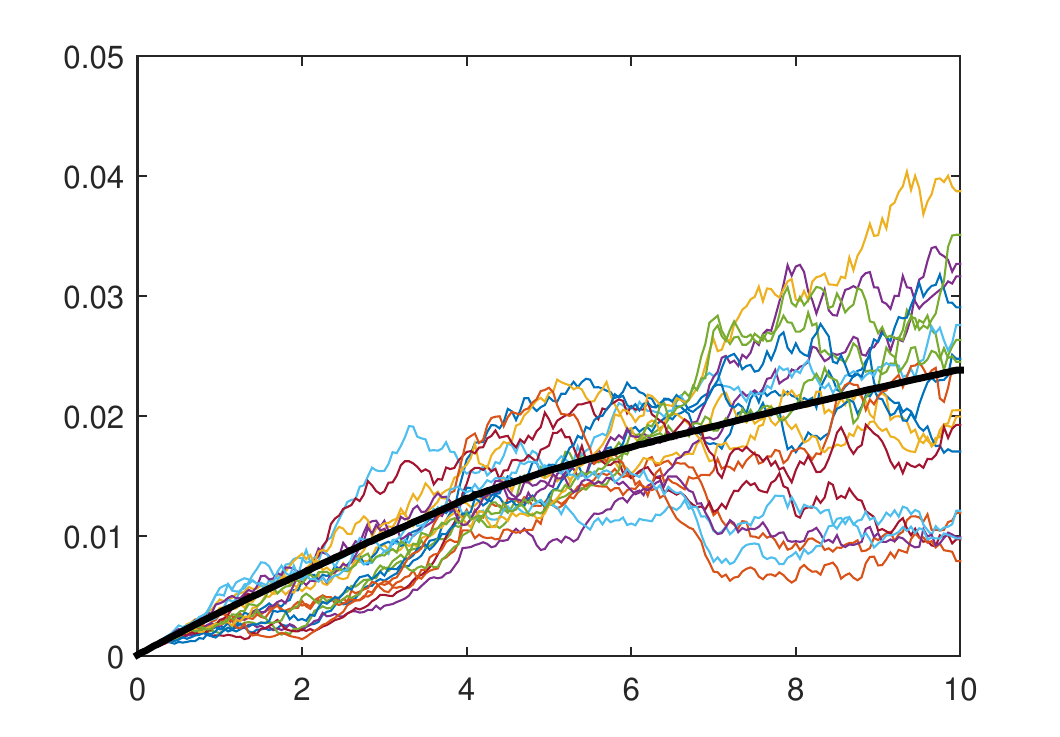}}
\parbox{0.49\textwidth}{\includegraphics[width=0.53\textwidth]{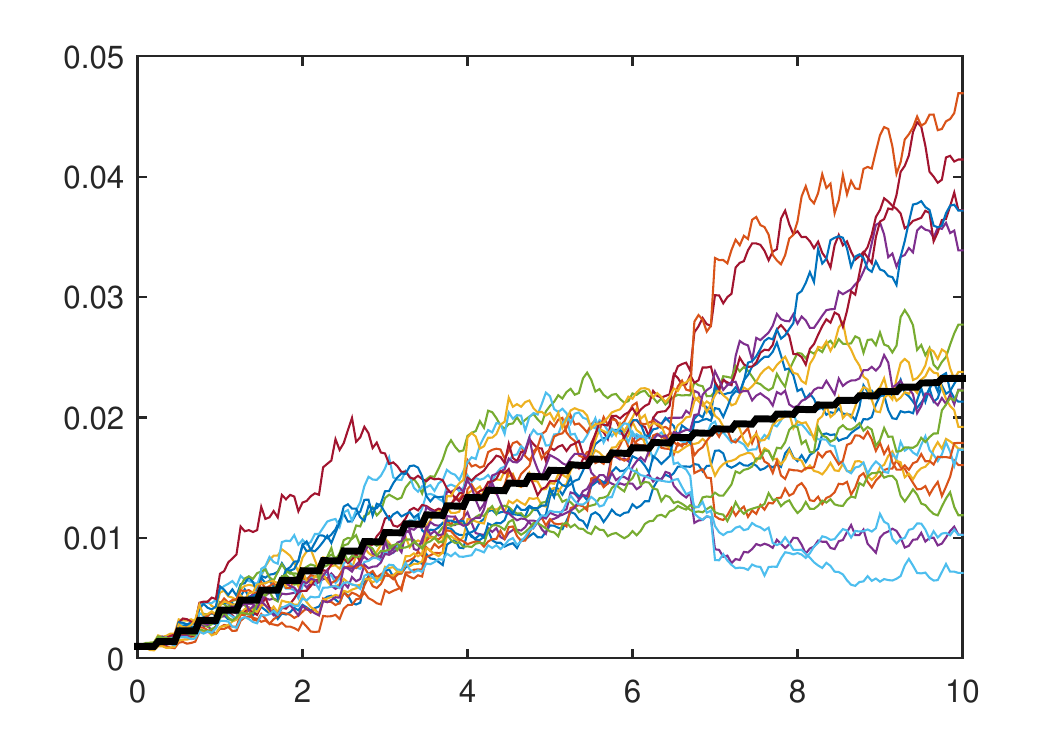}}
\end{centering}
\protect\caption{\label{fig:spline-vs-linear}20 sample paths of the short rate when $X$ is the 1D CIR process, together with the mean as a function of time 
(computed over $10^{5}$ paths). \emph{Left panel:} using a cubic spline nterpolation of $\left(u_k\right)$ to obtain $\mathpzc{U}$. \emph{Right panel: }using a linear interpolation.}
\end{figure}

\begin{corollary}
Let $\mathpzc U$ be a continuously differentiable interpolating function, i.e. $\mathpzc U \in C^1\left(\mathbb{R},\Rp^d\right)$, and consider the continuous 
tenor extension of the \alm \afflm. 
Then the short rate process $r$ is stochastically continuous.
\end{corollary}

However, even when the interpolating function is continuously differentiable there can be sources of undesirable behaviour of the short rate inhereted from 
the sequence $(u_k)$ itself (which is not unique unless $d=1$; cf. Remark \ref{rem:u_k_not_unique}). 
Consider, for example, the following `diagonal' structure for $(u_k)$, which is similar to the one employed by \citet{Grbac_Papapantoleon_Schoenmakers_Skovmand_2014} to model independence between rates of different maturities:
\[
\begin{array}[t]{cccccc}
u_N & = & (0 & \cdots & 0 & 0)\\
u_{N-1} & = & (0 & \reflectbox{$\ddots$}  & 0 & \bar{u}_N)\\
u_{N-2} & = & (0 & \reflectbox{$\ddots$} & \bar{u}_{N-1} & \bar{u}_N)\\
 &  & \vdots & \reflectbox{$\ddots$} & \vdots & \vdots\\
u_1 & = & (\bar{u}_1 & \cdots & \bar{u}_{N-1} & \bar{u}_N)\\
u_0 & = & (\bar{u}_0 & \cdots & \bar{u}_{N-1} & \bar{u}_N)
\end{array}
\]
\noindent with $\bar{u}_i\in\Rp$, for $1\le i\le N$. 
The only paths that can be used to interpolate in this case are the ones connecting the elements $u_k$ and $u_{k+1}$ of the sequence $(u_k)$ via straight lines, otherwise the interpolating function will not be component-wise decreasing. 
Hence, any interpolating function maps onto a non-smooth manifold. 
Then, requiring that the interpolating function is continuously differentiability (in time) will lead to a short rate that drops to zero at every $T_k$, $k\in\mathcal{K}$, since the derivative of the interpolating function will equal zero at each $T_k$; see again \eqref{eq:short-rate} and \eqref{eq:p-q}, and the illustration in Figure \ref{fig:drop-to-zero}.

\begin{figure}
\begin{centering}
\includegraphics[width=0.55\textwidth]{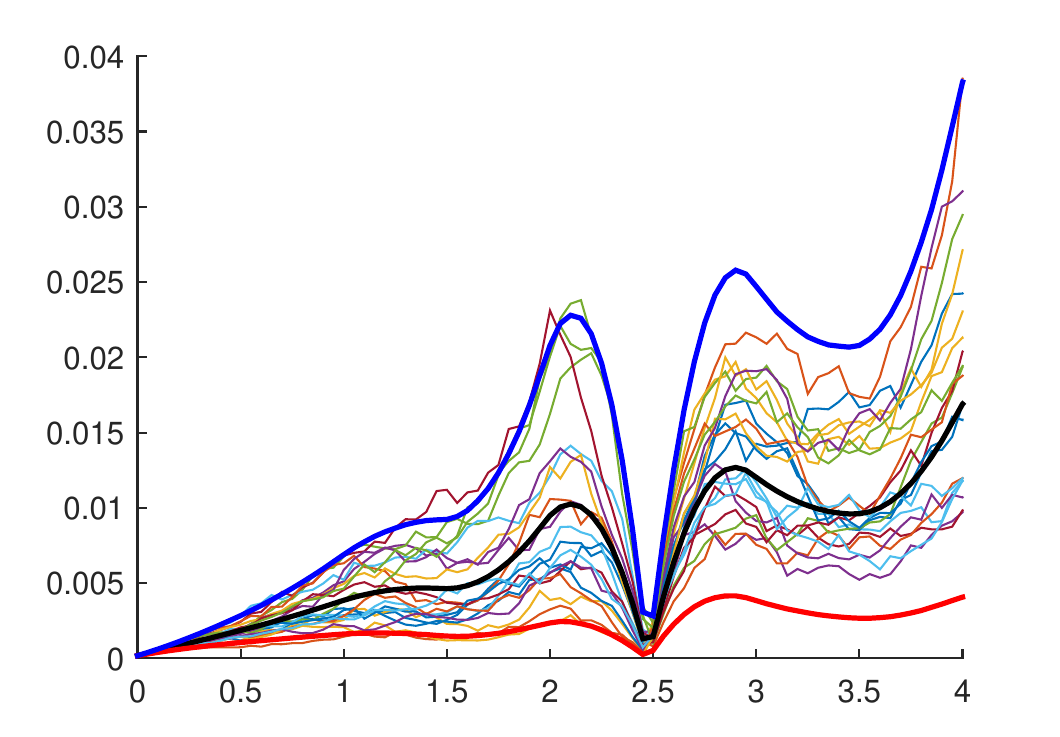}
\end{centering}
\protect\caption{\label{fig:drop-to-zero}20 sample paths of the short rate together with the 2.5/97.5 percentiles induced by a diagonal structure and a continuously differentiable interpolating function.}
\end{figure}

We would like in the sequel to provide conditions such that the short rate resulting from a continuous tenor extension of an affine \lib model exhibits `reasonable' behavior, in the sense that it neither jumps at fixed times, nor drops to zero at each maturity date. 
Moreover, we would like to identify a method for choosing an interpolating function that removes the arbitrariness from this choice. 
A condition for the former is that the sequence $(u_k)_{k\in\mathcal{K}}$ lies on a smooth manifold. 
Regarding the latter, we could require that a continuum of bond prices are fitted as well. 
These two together lead to a uniquely defined, continuously differentiable interpolating function. 

\begin{example}\label{example:interpolating-function}
Let \afflm be an \alm and assume that the sequence $(u_k)_{k\in\mathcal{K}}$ admits an interpolating function $\hat{\mathpzc{U}}$ that maps onto a smooth 
manifold $\mathscr{M}=\set{\hat{\mathpzc U}(t);t\in\left[0,T_N\right]}$. 
Moreover, let $\tilde{f}(0,\cdot)\colon [0,T_N]\rightarrow\Rp$ be an initial forward curve---belonging e.g. to the Nelson--Siegel or Svensson family---that 
is consistent with the initial bond prices, i.e.
\begin{align}\label{eq:initial-bonds}
B(0,T_k)=\exp\left( -\int_0^{T_k} \tilde f(0,s)\ds \right).
\end{align}
Then, we can find an interpolating function $\mathpzc U$ such that the continuous-tenor extended affine LIBOR model fits the given initial forward 
curve $\tilde f$. 
In order to achieve this, the interpolating function $\mathpzc{U}$ should satisfy the following:
\begin{eqnarray}\label{eq:root}
\phi_{T_N}\big(\mathpzc{U(T)}\big) 
  + \big\langle \psi_{T_N}\big(\mathpzc{U(T)}\big),X_0 \big\rangle
  = \int_{\mathpzc T}^{T_N} \tilde f(0,s) \ds
\quad \mbox{ and } \quad 
\mathpzc{U(T)} \in \mathscr{M}
\label{eq:forward-curve-fit}
\end{eqnarray}
for all $\mathpzc T \in [0,T_N]$. 
This equation follows directly from the requirement that \eqref{eq:initial-bonds} holds for all $\mathpzc T \in [0,T_N]$, together with \eqref{eq:Bond(t,tau)_equals_exp}--\eqref{eq:bond-AB} and Remark \ref{rem:interpol-0}. 
Moreover, the dynamics of the instantaneous forward rate are provided by \eqref{eq:short-rate}--\eqref{eq:p-q} for all $\tc\in[0,T_N]$, and satisfy the initial condition $f\brac{0,\tc}=\tilde{f}\brac{0,\tc}$.
\end{example}

The curve fitting problem in \eqref{eq:root} can be solved analogously to the problem of fitting the sequence $(u_k)$ to an initial term structure of bond 
prices; compare with Proposition 6.1 in \citet{KellerResselPapapantoleonTeichmann09} and the corresponding proof. 
The resulting interpolating function is then differentiable with respect to time, as the following result shows.

\begin{proposition}
Let \afflm be a multiple curve affine LIBOR model, assume that $\tilde f(0,\cdot):[0,T_N]\to\Rp$ is continuous and that $(u_k)_{k\in\mathcal{K}}$ allows for an interpolating function that maps onto a $C^1$-manifold $\mathscr{M}\subseteq\mathcal{I}_{T} \cap \R^d_{\geqslant0}$. 
Then, there exists a unique continuously differentiable interpolating function $\mathpzc U \colon [0,T_N]\rightarrow\mathscr{M}$ satisfying \eqref{eq:forward-curve-fit}.
\end{proposition}

\begin{proof}
This statement is an easy consequence of the implicit function theorem, where the differentiability of $\phi$ and $\psi$ in space and time as well as their 
order preserving property (cf. Lemma \ref{lem:phi_psi_prop}) are used. 
Indeed, the function 
\[
G\brac{t,u} 
  = \phi_{T_N}\brac{u} + \scal{\psi_{T_N}\brac{u}}{X_0} 
  - \int^{T_N}_t\tilde{f}\brac{0, s}\ds
\] 
is continuously differentiable in $t\in[0,T_N]$ and $u\in\mathcal{I}_{T}$. 
Let $\brac{U_\alpha,g_\alpha}_{\alpha\in\mathpzc{A}}$ be an atlas for the $C^1$-manifold $\mathscr{M}$, where $\mathpzc{A}$ is a finite index set, $\brac{U_\alpha}_{\alpha\in\mathpzc{A}}$ is an open covering of $\mathscr{M}$ and $g_\alpha:U_\alpha\rightarrow I_\alpha\subseteq\R$ is a $C^1$-homeomorphism. 
Define $G^\star\colon [0,T_N]\times I_\alpha\rightarrow\R$ by $\brac{t,x}\mapsto G\brac{t,g^{-1}_\alpha\brac{x}}$. 
By the (strict-)order preserving property of $\phi$ and $\psi$ we know that the partial derivative $\frac{\partial}{\partial x}G^\star\brac{t,x}$ is not zero, hence by a compactness argument there exists a unique continuously differentiable function $x\colon[0,T_N]\rightarrow \R$ such that $G^\star\brac{t,x\brac{t}}=0$ for all $t\in[0,T_N]$. 
The interpolating function is then given by $\mathpzc{U}\brac{t}\coloneqq g^{-1}_\alpha(x(t))$.
\end{proof}

\begin{remark}
Figure \ref{fig:spline-vs-linear} reveals another interesting behavior of the short rate implied by an affine LIBOR model. 
In particular, there exists a lower bound for the short rate that is greater than zero. 
Indeed, since the state space of our driving affine process is $\Rp^d$, we have that $r_{t}\geq p_{t}$, which is greater than zero as $\phi$ is strictly order preserving and $u$ is decreasing. 
A similar phenomenon was already observed in the discrete tenor model for the LIBOR rate, compare with \citet[Rem.~6.4]{KellerResselPapapantoleonTeichmann09}. 
\end{remark}

\section{Computation of XVA in affine LIBOR models}
\label{sec:5}

The quoted price of a derivative product in pre-crisis markets was equal to its discounted expected payoff (under a martingale measure), since counterparties 
were considered default-free, there was abundance of liquidity in the markets, and other frictions were also negligible.
In post-crisis markets however these assumptions have been challenged; in particular, counterparty credit risk has emerged as the natural form of default risk, there is shortage of liquidity in financial markets, while other frictions have also gained importance. 
These facts have thus to be factored into the quoted price. 
One way to do that, is to compute first the so-called `clean' price of the derivative, which equals its discounted expected payoff (under a martingale measure), and then add to it several \textit{value adjustments}, collectively abbreviated as XVA, that reflect counterparty credit risk, liquidity costs, etc. 
We refer to \cite{Brigo_Morini_Pallavicini_2013}, \cite{Crepey2012a,Crepey2012b}, \citet{Crepey_Bielecki_2014} and \cite{Bichuch_Capponi_Sturm_2016} among others for more details on XVA.

\subsection{Clean valuation}
\label{sec:Clean-Valuation}

This section reviews basis swaps and provides formulas for their clean price in \alms with multiple curves. 
The clean valuation of caps, swaptions and basis swaptions in these models is extensively studied in \citet{Grbac_Papapantoleon_Schoenmakers_Skovmand_2014}.

The typical example of an interest rate swap is where a floating rate is exchanged for a fixed rate; see, e.g., \citet[\S9.4]{MusielaRutkowski05}. 
The appearance of significant spreads between rates of different tenors has given rise to a new kind of interest rate swap, called basis swap, where two streams of floating payments linked to underlying rates of different tenors are exchanged. 
As an example, in a 3M-6M basis swap linked to the LIBOR, the 3-month LIBOR is paid quarterly and the 6-month LIBOR is received semiannually. 

Let $\mathcal{T}_{p_{1}q_{1}}^{x_{1}} = \left\{T_{p_{1}}^{x_{1}},\dots,T_{q_{1}}^{x_{1}}\right\}$ and $\mathcal{T}_{p_{2}q_{2}}^{x_{2}} = \left\{T_{p_{2}}^{x_{2}},\dots,T_{q_{2}}^{x_{2}}\right\}$ denote two tenor structures, where $T_{p_{1}}^{x_{1}}=T_{p_{2}}^{x_{2}}$, $T_{q_{1}}^{x_{1}}=T_{q_{2}}^{x_{2}}$ and $\mathcal{T}_{p_{2}q_{2}}^{x_{2}}\subset\mathcal{T}_{p_{1}q_{1}}^{x_{1}}$. 
Consider a basis swap that is initiated at $T_{p_{1}}^{x_{1}}=T_{p_2}^{x_2}$, with the first payments due at $T_{p_{1}+1}^{x_{1}}$ and $T_{p_{2}+1}^{x_{2}}$ respectively. 
In order to reflect the possible discrepancy between the floating rates at initiation, the interest rate $L\left(T_{i-1}^{x_{1}},T_{i}^{x_1}\right)$ corresponding to the shorter tenor length $x_{1}$ is replaced by $L\left(T_{i-1}^{x_{1}},T_{i}^{x_1}\right)+S$ for a fixed $S$, which is called the basis swap spread. 
The time-$r$ value of a basis swap with notional amount normalized to $1$, for $0\le r\le T_{p_1}^{x_1}$, is given by
\begin{multline*}
\mathbb{BS}_r\left(S,\mathcal{T}_{p_{1}q_{1}}^{x_{1}},\mathcal{T}_{p_{2}q_{2}}^{
x_ { 2 } } \right)
= \sum_{i=p_{2}+1}^{q_{2}}\delta_{x_{2}}B\left(r,T_{i}^{x_{2}}
\right)\mathds{E}_i^{x_{2}}\left[\left.L\left(T_{i-1}^{x_{2}},T_{i}^{x_{2}}
\right)\right|\mathcal{F}_{r}\right]\\
-\sum_{i=p_1+1}^{q_1}\delta_{x_{1}}B\left(r,T_{i}^{x_{1}}\right)\mathds{E}_i^
{x_{1}}\left[\left.L\left(T_{i-1}^{x_{1}},T_{i}^{x_{1}}
\right)+S\right|\mathcal{F}_{r}\right]\\
=\sum_{i=p_{2}+1}^{q_{2}} \delta_{x_{2}}B(r,T_{i}^{x_{2}})L_{i}^{x_{2
} }(r)-\sum_{i=p_{1}+1}^{q_{1}} \delta_{x_{1}}B(r,T_{i}^{x_{1}}
)\left(L_{i}^{x_{1}}(r)+S\right).
\end{multline*}
The fair basis swap spread $S_r(\mathcal{T}_{p_{1}q_{1}}^{x_{1}},\mathcal{T}_{p_{2}q_{2}}^{x_{2}})$ is then computed such that the value of the swap at inception is zero, i.e. $\mathbb{BS}_r(S,\mathcal{T}_{p_1q_1}^{x_1},\mathcal{T}_{p_2q_2}^{x_2})=0$ for $0\leq r\leq T_{p_{1}}^{x_{1}}$. 
Hence, the fair spread is given by 
\begin{align*}
S_{r}(\mathcal{T}_{p_{1}q_{1}}^{x_{1}},\mathcal{T}_{p_{2}q_{2}}^{x_{2}})  
= \frac{\sum_{i=p_{2+1}}^{q_{2}}\delta_{x_{2}}B(r,T_{i}^{x_{2}})L_{i}^{x_{2}}
(r)-\sum_{i=p_{1}+1}^{q_{1}}\delta_{x_{1}}B(r,T_{i}^{x_{1} 
})L_{i}^{x_{1}}(r)}{\sum_{i=p_{1}+1}^{q_{1}}\delta_{x_{1}} 
B(r,T_{i}^{x_{1}})}.
\end{align*}

\noindent Moreover, the time-$t$ value of the basis swap, for $t\in\left[T_{p_{1}}^{x_{1}},T_{q_{2}}^{x_{2}}\right]$, using \eqref{eq:model-OIS-FRA-rate} and \eqref{eq:bonds_for_all_maturities2}, takes the form:
\begin{align}\label{eq:price_bswap}
\mathbb{BS}_t \left(S_r,\mathcal{T}_{p_1q_1}^{x_1},\mathcal{T}_{p_2q_2}^{x_2}\right) 
&= \frac{1}{M_{t}^{\mathpzc U(t)}} \sum_{i=\lceil t \rceil_2}^{q_2} 
\left(M_{t}^{v_{i-1}^{x_2}} - M_{t}^{u_i^{x_2}}\right) \nonumber \\
& \quad  - \frac{1}{M_{t}^{\mathpzc U(t)}} \sum_{i=\lceil t \rceil_1}^{q_{1}} 
\left(M_{t}^{v_{i-1}^{x_1}} - M_{t}^{u_i^{x_1}} \left( 1- \delta_{x_1} S_r
\right)\right),
\end{align}
where $S_r=S_r\left(\mathcal{T}_{p_1q_1}^{x_1},\mathcal{T}_{p_2q_2}^{x_2}\right)$, for $r\in[0,T_{p_1}^{x_1}]$ being the date of inception, while $\lceil t \rceil_i = \min\left\{ k\in\mathcal{K}^{x_i}: t<T_k^{x_i} \right\}$.

\begin{remark}
Basis swaps are post-crisis financial products, which can only be priced in models accounting for the multiple curve nature of interest rates. 
In a single curve model, the price of a basis swap is zero; cf. \citet[p.~181]{Crepey_Grbac_Nguyen_2011}
\end{remark}

\subsection{XVA equations}

The pricing formulas in the previous subsection reflect valuation in an environment without counterparty credit risk, funding constraints and other market frictions. 
In order to include the latter into the pricing framework, several value adjustments have been introduced: credit and debt valuation adjustment (CVA and DVA), liquidity valuation adjustment (LVA), as well as replacement cost (RC), among others. 
The various valuation adjustments are typically abbreviated by XVA, while we will refer to their sum as the total valuation adjustment (TVA), i.e.
\[
\mbox{TVA}=\mbox{CVA}+\mbox{DVA}+\mbox{LVA}+\mbox{RC};
\]
see also \citet{Crepey_Gerboud_Grbac_Ngor_2012}. 
Our approach to the computation of TVA follows closely the work of \cite{Crepey2012a,Crepey2012b}.

We consider two counterparties, called a bank and an investor in the sequel, that are both defaultable, and denote by $\tau_b$ the default time of the bank, by $\tau_i$ the default time of the investor, while we set $\tau = \tau_b \wedge \tau_i \wedge T$. 
The default intensities of $\tau_b,\tau_i$ and $\tau$ are denoted $\gamma_b, \gamma_i$ and $\gamma$, respectively. 
We also consider the `full model' filtration $\mathbb{G}$, which is given by $\mathbb F$ enlarged by the natural filtrations of the default times $\tau_b$ and $\tau_i$, and assume that the immersion hypothesis holds, that is, every $\mathbb F$-martingale stopped at $\tau$ is a $\mathbb G$-martingale. 

The TVA can be viewed as the price of a dividend paying option on the debt of the bank to the investor, paying off at the first-to-default time $\tau$. 
Here, we have implicitly adopted the point of view of the bank. 
The TVA from the point of view of the investor is similar, but not identical, due to e.g. different funding conditions. 
The effective conclusion of \cite{Crepey2012b} is that the TVA in the setting described above can be computed in a `pre-default' framework, where the default 
risk of the counterparties appears only through the default intensities; see, in particular, Section 3 therein. 
More specifically, the TVA $\Theta$ is the solution of the following BSDE under a martingale measure $\P$:
\begin{align}\label{eq:TVA-bsde}
\Theta_{t} = \E_{t} \left( \int_{t}^{T} g_{s}\left(r_s,P_s,\Theta_{s} 
  \right)\ds\right), \quad t\in[0,T],
\end{align}
where $r$ denotes the short rate process, $P$ the clean price process and $g$ the TVA coefficient. 
The overall price of the contract for the bank, in other words, the cost of the hedge incorporating the various risks, is then given by the diffference between the clean price and the TVA:
\[
\Pi_t = P_t - \Theta_t, \quad t\in[0,T].
\]
The TVA coefficient $g$ has the following form:
\begin{align*}
& \!\!\!\!\!\!\!\!\!\!\!\!\! 
g_t(r_t, P_t, \Theta_t) + r_t\Theta_t \ = \\  
& -\gamma_{t}^{i}\left(1-\rho^{i}\right)\left(Q_t-\Gamma_{t}\right)^{-}
	\tag{CVA}\\
& +\gamma_{t}^{b}\left(1-\rho^{b}\right)\left(Q_t-\Gamma_{t}\right)^{+}
	\tag{DVA}\\
& + b_{t}\Gamma_{t}^{+} - \bar{b}_{t}\Gamma_{t}^{-}
  + \lambda_{t}\left(P_t-\Theta_t-\Gamma_{t}\right)^{+}
  - \tilde\lambda_t \left(P_t-\Theta_t-\Gamma_{t}\right)^{-}
  	\tag{LVA}\\
& +\gamma_{t}\left(P_t-\Theta_t-Q_t\right), \tag{RC}
\end{align*}
where $\tilde\lambda_t = \bar{\lambda}_{t}-\gamma_{t}^{b}(1-\mathfrak{r})$, while each line on the right hand side corresponds to one of the four components 
of the TVA. 
The parameters in the above equation have the following financial interpretation:
\begin{itemize}
\item $\gamma_{t}^{i},\,\gamma_{t}^{b}$ and $\gamma_{t}$ are the default 
      intensities of the investor, the bank and the first to default intensity, 
      respectively.
\item $\rho^{i},\,\rho^{b}$ are the recovery rates of the investor and the bank 
      to each other, and $\mathfrak{r}$ is the recovery rate of the bank to its 
      unsecured funder (which is a third party that jumps in when the banks' 
      internal sources of funding have been depleted; this funder is assumed to 
      be risk free).
\item $Q_t$ is the value of the contract according to some valuation scheme 
      specified in the credit support annex (CSA), which is a common part in an 
      over-the-counter contract.
\item $\Gamma_{t}=\Gamma_{t}^{+}-\Gamma_{t}^{-}$ is the value of the collateral 
      posted by the bank to the investor.
\item $b_{t},\,\bar{b}_{t}$ and $\lambda_{t},\,\bar{\lambda}_{t}$ are the 
      spreads over the risk free rate $r_{t}$ corresponding to the remuneration 
      of collateral and external lending and borrowing (from the unsecured 
      funder).
\end{itemize}
The value of the contract $Q$ and of the collateral $\Gamma$, as well as the funding coefficients $b$ and $\bar{b}$ are specified in the CSA of the contract, 
which is used to mitigate counterparty risk. 
Different CSA specifications will result in different behavior of the TVA; see \citet[Sec.~3]{Crepey_Gerboud_Grbac_Ngor_2012} for more details and also the next 
subsection.

\begin{remark}
The immersion hypothesis implies weak or indirect dependence between the contract and the default times of the involved parties. 
Therefore not every contract can be priced within the pre-default TVA framework. 
As interest rate contracts exhibit weak dependence on the default times, this approach is appropriate for our setting; see also \citet[Rem. 2.3]{Crepey2012b}.
\end{remark}

\subsection{XVA computation in \alms}

We are interested now in computing the value adjustments for interest rate derivatives, and focus on basis swaps as a prime example of a post-crisis product. 
The OIS forward rate and the forward \lib rate for each tenor are modeled according to the \alms with  multiple curves, and the model is calibrated to caplet data; see \citet[\S8]{Grbac_Papapantoleon_Schoenmakers_Skovmand_2014} for details on the calibration of \alms. 
An interpolating function is subsequently chosen and the dynamics of the short rate process are derived. 
Afterwords, the computation of the value adjustments is a straightforward application of the TVA BSDE in \eqref{eq:TVA-bsde}.

This methodology allows us to compute option prices and value adjustments consistently since we only have to calibrate the discrete-tenor \alm, while the dynamics of the short rate process, which is essential in the computation of the TVA, follows from the interpolation. 
In particular, we do not need to introduce and calibrate (or estimate) an `exogenous' model for the short rate, as is done in other approaches. 
The interpolating function plays thus a crucial role in our methodology, since this is the only `free' ingredient once the \alm has been calibrated. 
At the same time, it introduces an element of model risk, through the different possible choices of interpolating functions. 
In the sequel, we are going thus to examine the impact of different interpolating functions on the value adjustments. 

\begin{figure}
\centering
\includegraphics[width=0.8\textwidth]{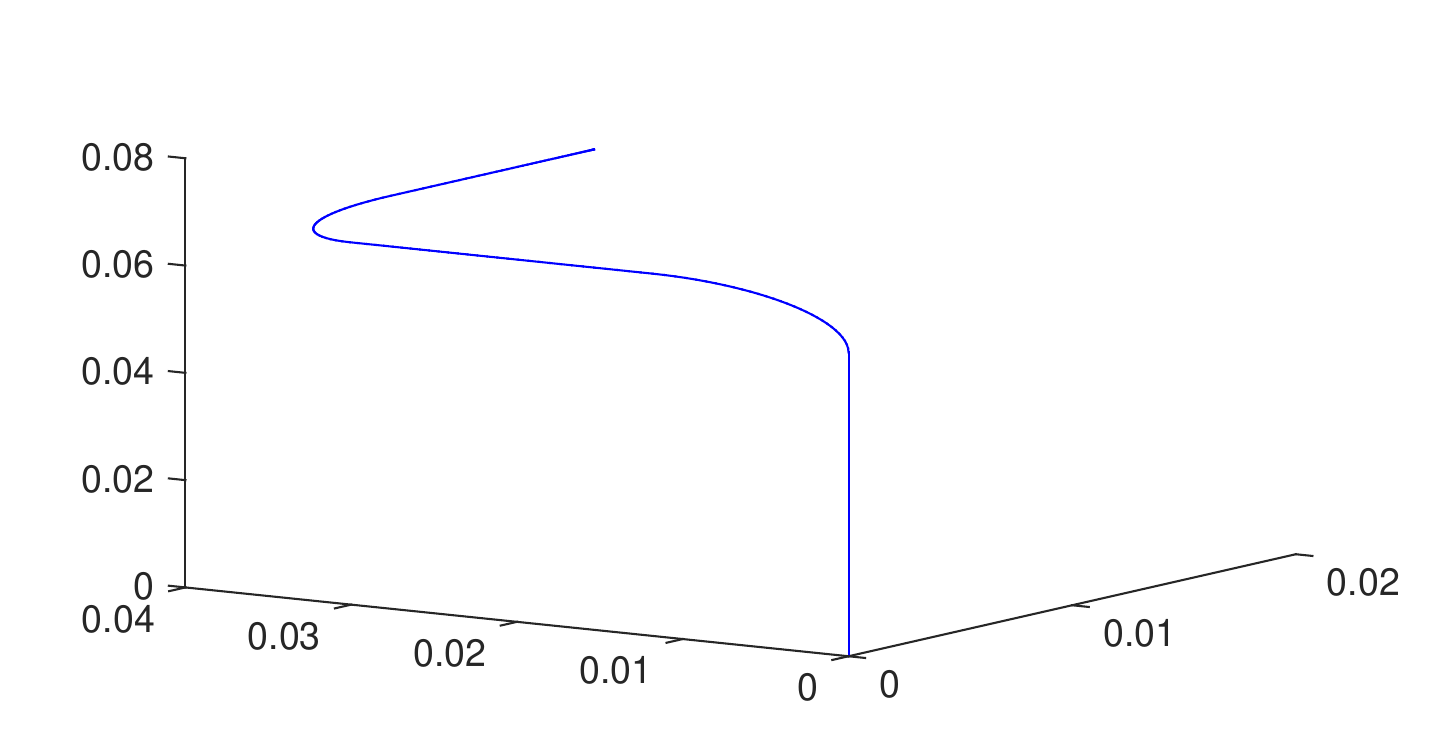}
\label{fig:smooth-manifold}
\caption{The smooth manifold used for fitting the sequences $\brac{u_k}$ and 
$\brac{v_k}$.}
\end{figure}

The data we use for our numerical experiments correspond to the EUR market on 27 May 2013 and were collected from Bloomberg; see also \citet[\S8.4]{Grbac_Papapantoleon_Schoenmakers_Skovmand_2014} for more details. 
The \alm with multiple curves was calibrated to caplet data on a 10 year horizon, where the tenor lengths were 3 and 6 months. 
The driving affine process consists of three independent CIR processes. 
The sequences $(u_k)$ and $(v_k)$ were constructed such that they lie on a smooth manifold on $\mathcal I_T \cap \R^3_{\geqslant0}$; see Figure \ref{fig:smooth-manifold}. 
In particular, $u_k$ lies on straight lines for $k\in\set{0,\dots,k_1}\cup\set{k_2,\dots,k_3}\cup\set{k_4,\dots,N}$ and on elliptical segments for $k\in\set{k_1+1,\dots,k_2-1}\cup\set{k_3+1,\dots,k_4-1}$. 
The sequence $u$ thus looks as follows:
\[
\begin{array}[t]{ccccc}
u_{N-1} & = & ( 0 & 0 & \bar{u}_{N-1} ) \\
		& \vdots  &   &     &           \\
u_{k_4} & = & (0  & 0 & \bar{u}_{k_4} \ \ ) \\
u_{k_4-1} & = & (0  & \tilde u_{k_4-1} & \bar{u}_{k_4-1} ) \\
u_{k_4-2} & = & (0  & \tilde u_{k_4-2} & \bar{u}_{k_4-2} ) \\
		& \vdots  &   &     &           \\
\end{array}
\]
where $\bar u_j, \tilde u_j\in\Rp$ and satisfy $\bar u_j\ge\bar u_{j+1}$ and $\tilde u_j\ge\tilde u_{j+1}$ for all relevant $j\in\mathcal K$; see \eqref{eq:def-ineq-u} again. 
The structure of the sequence $v^x$ for each tenor $x$ is analogous. 
In other words, short term forward \lib rates are driven by all three components of the driving process $X$, medium term rates by two components, while long term rates are only driven by the last component of $X$. 
Once the manifolds have been constructed, the sequences $(u_k)$ and $(v_k)$ were obtained by fitting the model to OIS and EURIBOR data from the same date. 
(Note that in this example we have $N=40$ and we chose $k_1=9, k_2=16, k_3=21$ and $k_4=28$.)

In order to illustrate the effect of different interpolating functions on the value adjustments, we consider three different specifications for the interpolating function $\mathpzc U$:
\begin{itemize}
\item[(IF$_1$):] interpolation by fitting an entire forward curve (see Example \ref{example:interpolating-function});
\item[(IF$_2$):] linear interpolation between the $u_k$'s; 
\item[(IF$_3$):] spline interpolation on sectors where all but one component of the vector $u_k$ are constant in $k$, and linear interpolation in between       				 these sectors (i.e. when the $u_k$'s lie on curved segments of the manifold).	
\end{itemize}

Let us now turn our attention to the computation of value adjustments. 
We consider a 3M-6M basis swap on the LIBOR, with inception at $t=0$ and maturity in 10 years. 
We follow \citet{Crepey_Gerboud_Grbac_Ngor_2012} and consider five different CSA specifications, provided by
\begin{equation*}
\begin{array}{lcclccl}
\hspace{-4.75em} (\text{CSA}_1):
&&&\left(\mathfrak{r},\rho^b,\rho^{i}\right)=\left(0.4,0.4,0.4\right), & Q=P,    && \Gamma=0, \\
\hspace{-4.75em} (\text{CSA}_2):
&&&\left(\mathfrak{r},\rho^b,\rho^{i}\right)=\left(1,0.4,0.4\right),   & Q=P,    && \Gamma=0, \\
\hspace{-4.75em} (\text{CSA}_3):
&&&\left(\mathfrak{r},\rho^b,\rho^{i}\right)=\left(1,1,0.4\right),     & Q=P,    && \Gamma=0, \\
\hspace{-4.75em} (\text{CSA}_4):
&&&\left(\mathfrak{r},\rho^b,\rho^{i}\right)=\left(1,1,0.4\right),     & Q=\Pi,  && \Gamma=0, \\
\hspace{-4.75em} (\text{CSA}_5):
&&&\left(\mathfrak{r},\rho^b,\rho^{i}\right)=\left(1,0.4,0.4\right),   & Q=P,    && \Gamma=Q=P,
\end{array}
\end{equation*}
while the default intensities and spreads equal
\[
\gamma^b=5\%, \ \gamma^i=7\%, \ \gamma=10\%, \ b=\bar{b}=\lambda=1.5\%
\quad \text{and} \quad 
\bar{\lambda}=4.5\%.
\]
The first three CSA specifications correspond to a `clean' recovery scheme without collateralization, since the value of the contract $Q$ equals the clean price and there is no collateral posted. 
The fourth specification corresponds to a `pre-default' recovery scheme without collateralization, while the last one corresponds to a fully collateralized contract. 
Moreover, the first specification yields a linear BSDE in the TVA $\Theta$, which allows to use (forward) Monte Carlo simulations for the computation of the TVA.	

The price $P_t$ of the basis swap is provided by \eqref{eq:price_bswap} for each $t\in[T_{p_1}^{x_1},T_{q_2}^{x_2}]$, and we can observe that $P_t$ is a deterministic transformation of $X_t$. 
Moreover, the short rate $r_t$ is a deterministic, affine, transformation of $X_t$; cf. \eqref{eq:short-rate}. 
Therefore, the TVA coefficient $g_t(r_t,P_t,\Theta_t)$ is also a deterministic transformation of $X_t$, and we can define a deterministic function $\hat g$ such that
\[
\hat g(t,X_t,\Theta_t) :=  g_t(r_t,P_t,\Theta_t).
\]
In other words, the TVA BSDE is Markovian in this case, and the TVA is also provided by the solution of a semi-linear PDE. 
In order to compute the TVA for the basis swap numerically, we worked under the spot martingale measure $\P_\star$, using a space grid consisting of $10^5$ paths and a time grid with $n=200$ steps of step size $h$. 
We applied a backwards regression on the space-time grid, i.e.
\[
\Theta_{t_l}^{n}
 = \mathds{E}_\star \left[\left. \Theta_{t_{l+1}}^{n} 
 + h\hat{g}\left(t_{l+1},X_{t_{l+1}}^{n},\Theta_{t_{l+1}}^{n}\right) 
  \right| X_{t_l}^{n} \right]
 \quad \text{ and }\quad \Theta_{t_{n}}^{n}=0,
\] 
and approximated the conditional expectation using an $m$-nearest neighbors estimator with $m=3$. 
This choice turned out to be optimal when compared to (forward) Monte Carlo simulations in the case of a linear TVA coefficient.

\begin{figure}[ht]
\centering\hspace{-1cm}
\begin{minipage}[b]{0.31\linewidth}
\centering\title{\qquad (IF$_1$)}
\includegraphics[width=1.15\textwidth]{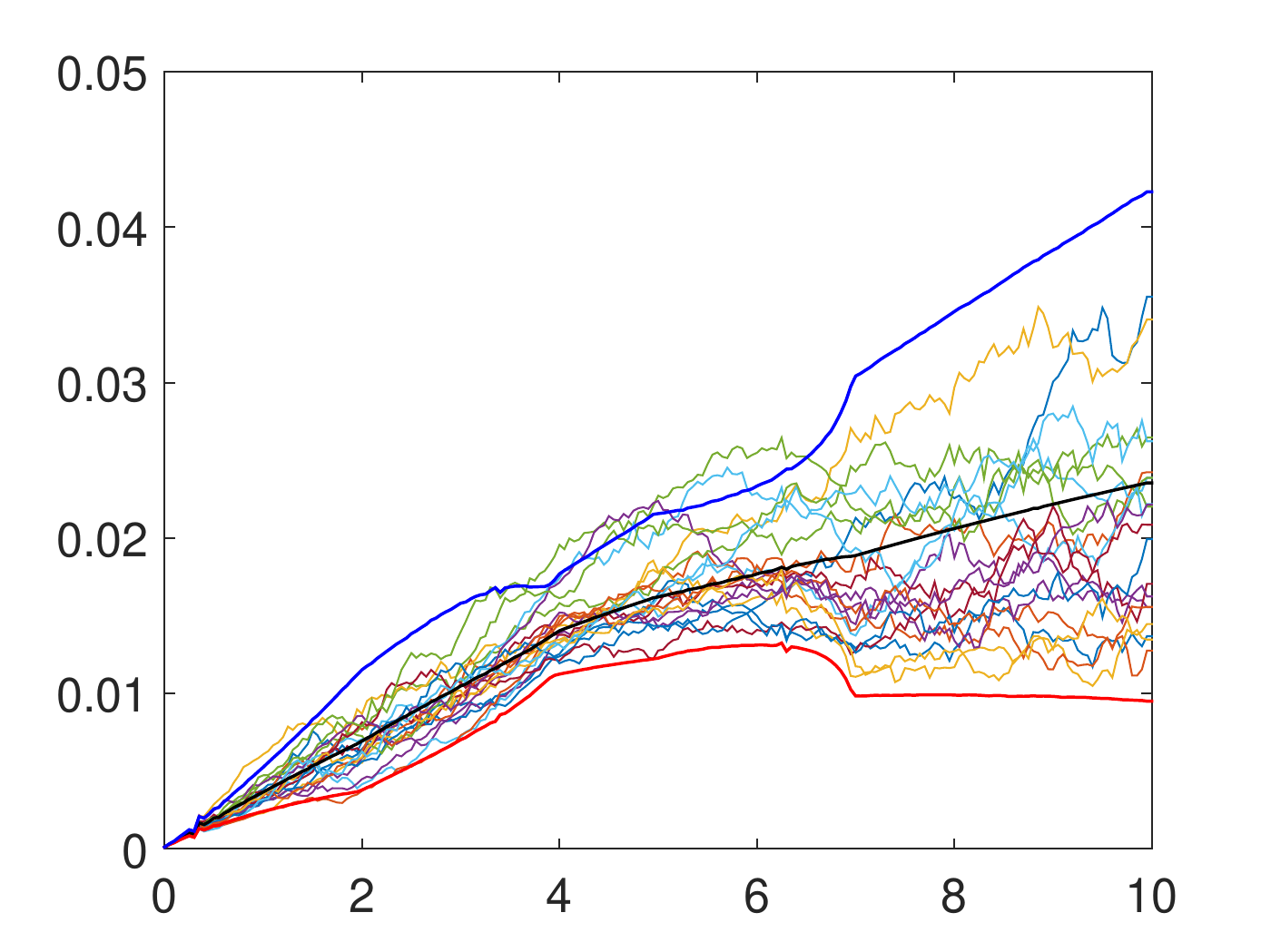}
\includegraphics[width=1.15\textwidth]{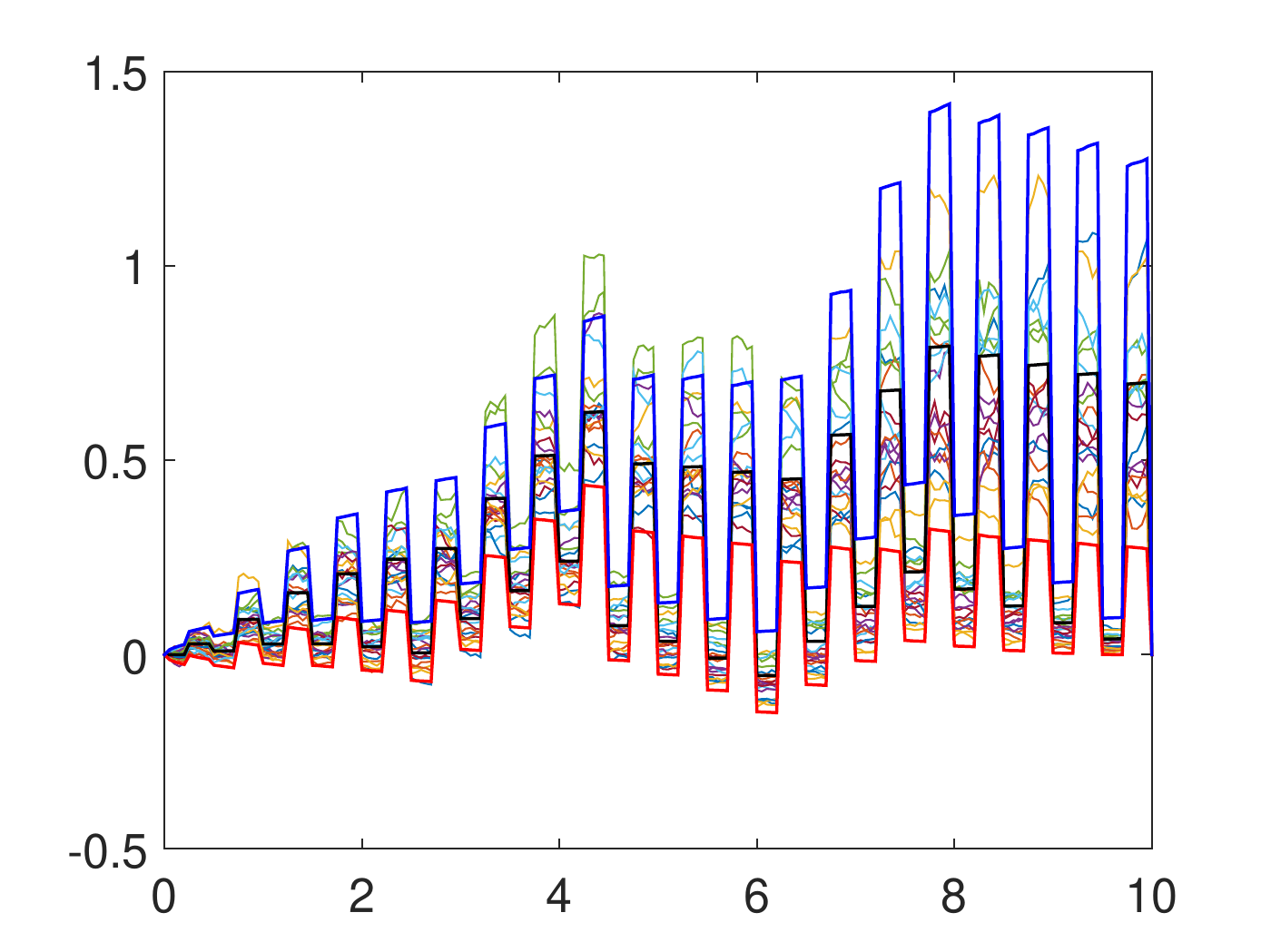}
\end{minipage}
\,
\begin{minipage}[b]{0.31\linewidth}
\centering\title{\qquad (IF$_2$)}
\includegraphics[width=1.15\textwidth]{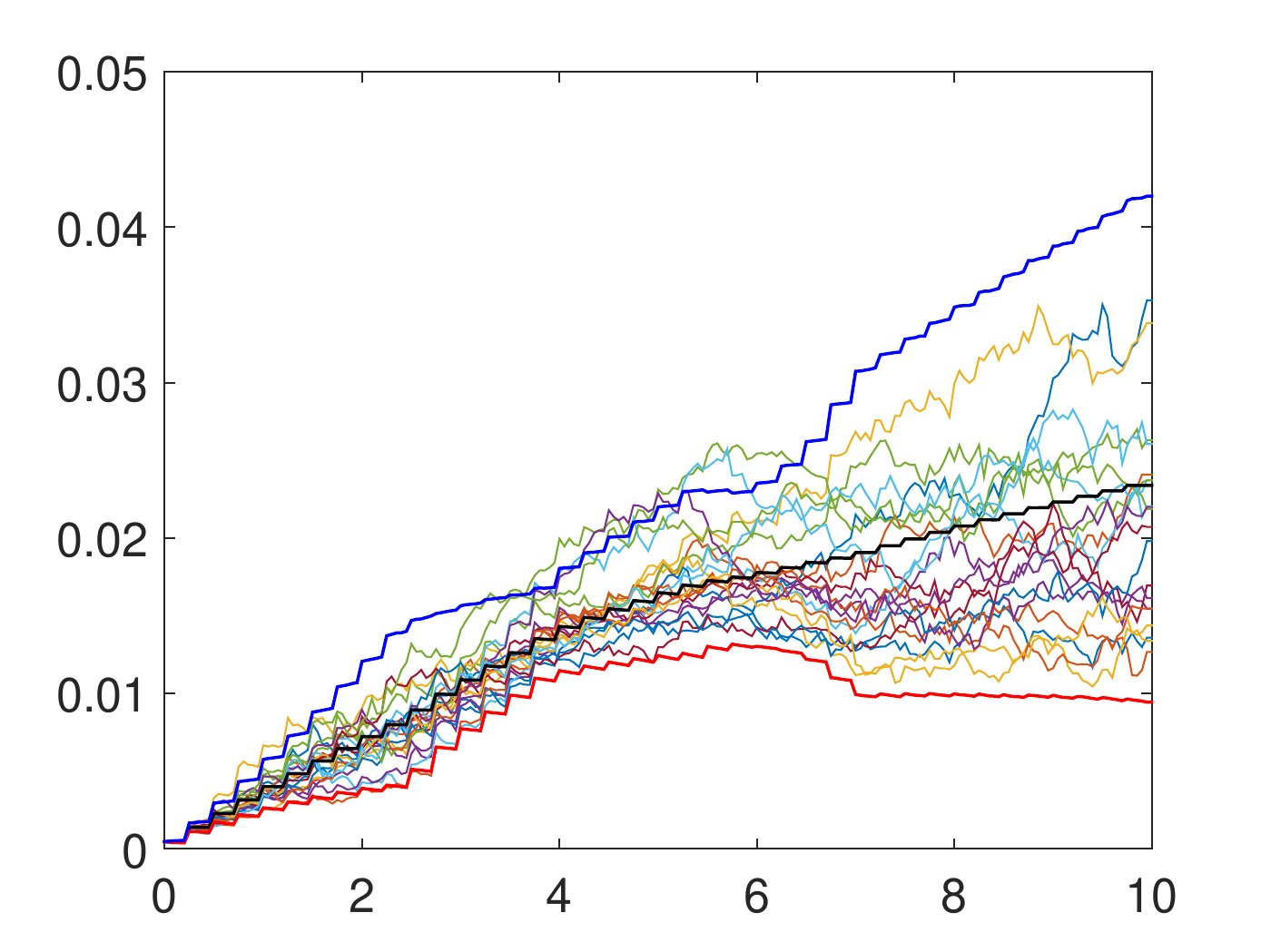}
\includegraphics[width=1.15\textwidth]{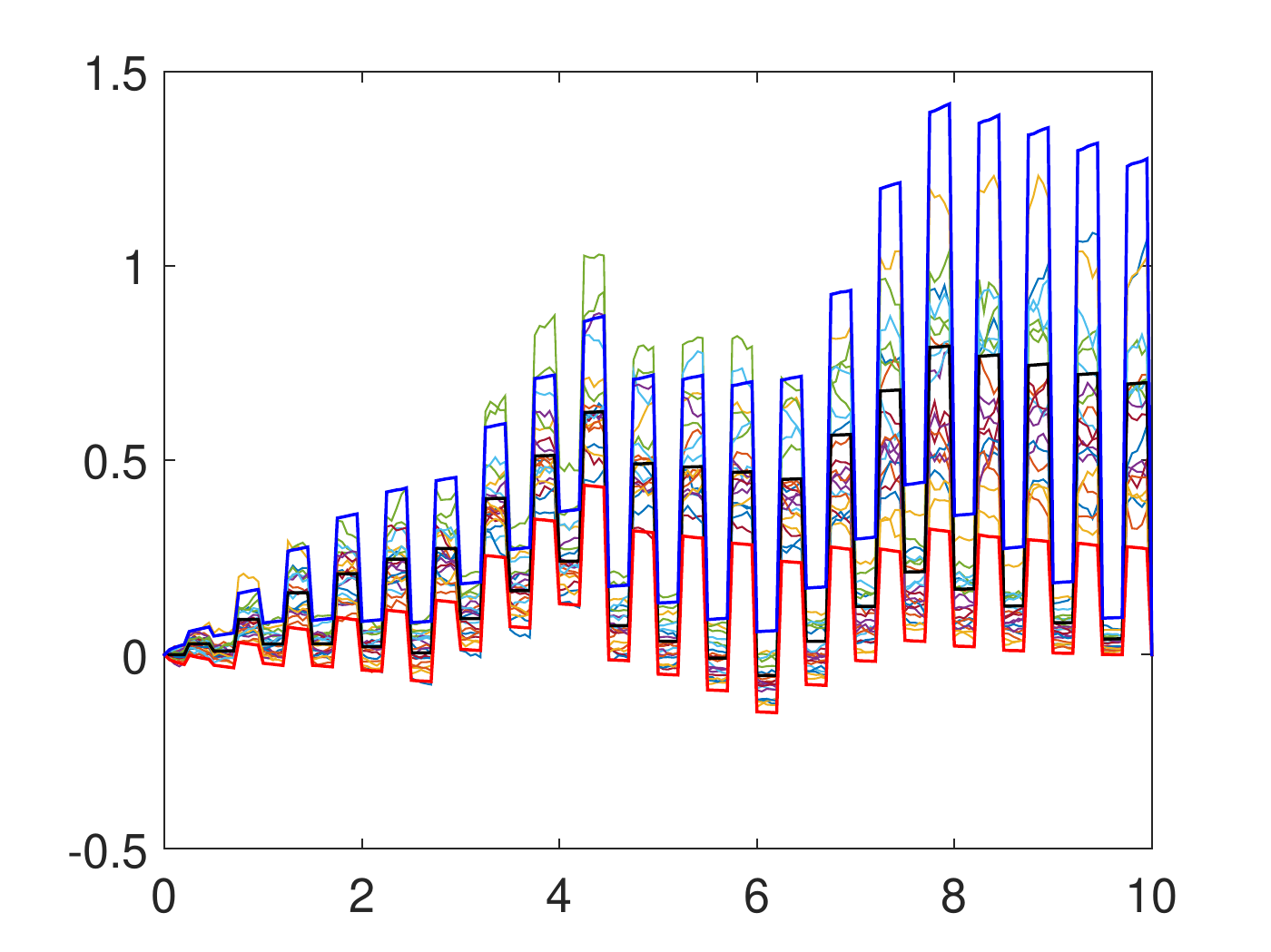}
\end{minipage}
\,
\begin{minipage}[b]{0.31\linewidth}
\centering\title{\qquad (IF$_3$)}
\includegraphics[width=1.15\textwidth]{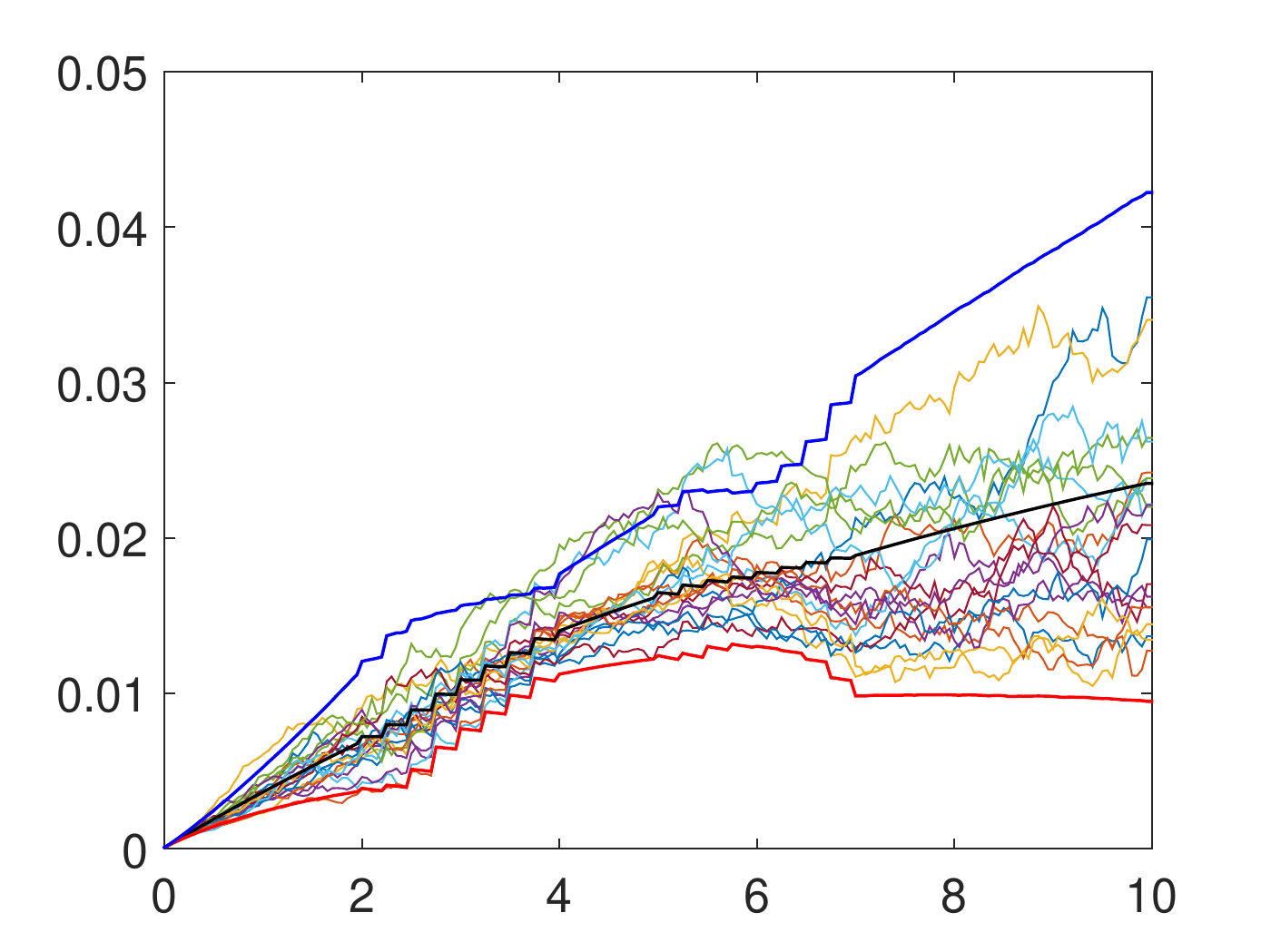}
\includegraphics[width=1.15\textwidth]{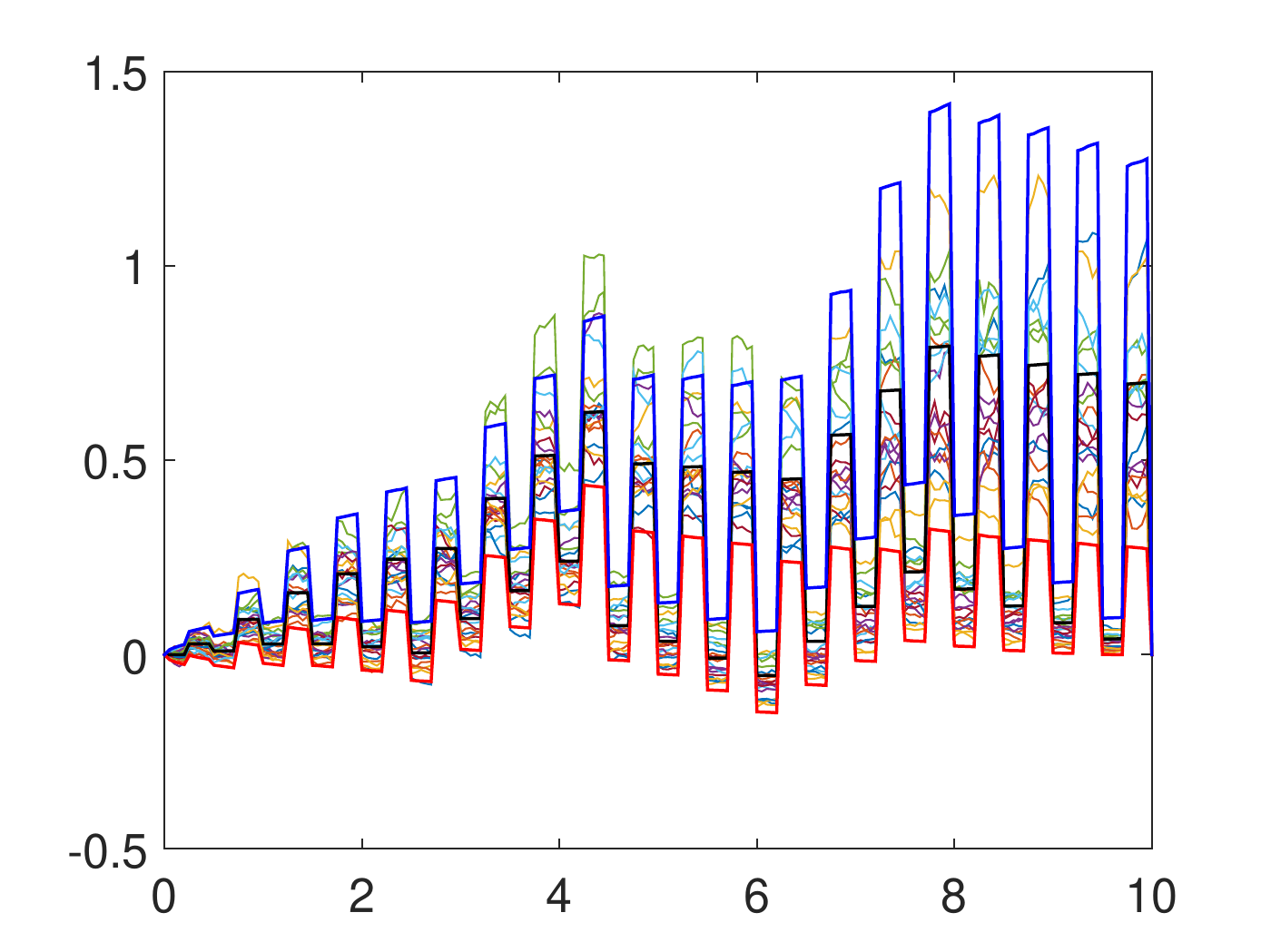}
\end{minipage}
\caption{20 sample paths of the short rate (top panels) and the price process 
of a basis swap (bottom panels) for each interpolating function along with the 
mean (black line) and 97.5\% and 2.5\% percentiles computed over $10^5$ 
realizations.}
\label{fig:r-P-paths}
\end{figure}

\begin{figure}[ht]
\centering\hspace{-1cm}
\begin{minipage}[b]{0.31\linewidth}
\centering\title{\qquad (IF$_1$) vs. (IF$_2$)}
\includegraphics[width=1.15\textwidth]{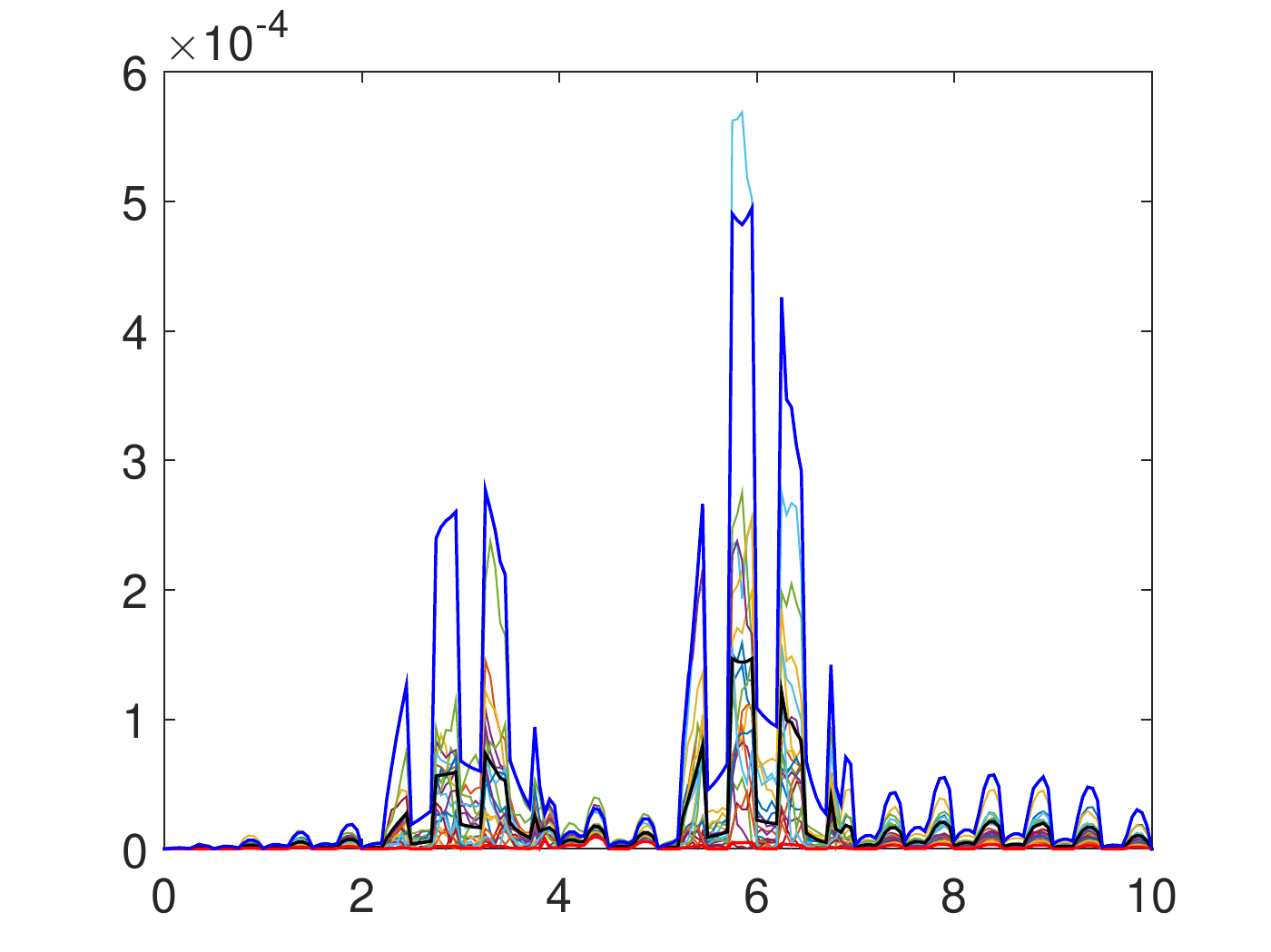}
\end{minipage}
\,
\begin{minipage}[b]{0.31\linewidth}
\centering\title{\qquad (IF$_1$) vs. (IF$_3$)}
\includegraphics[width=1.15\textwidth]{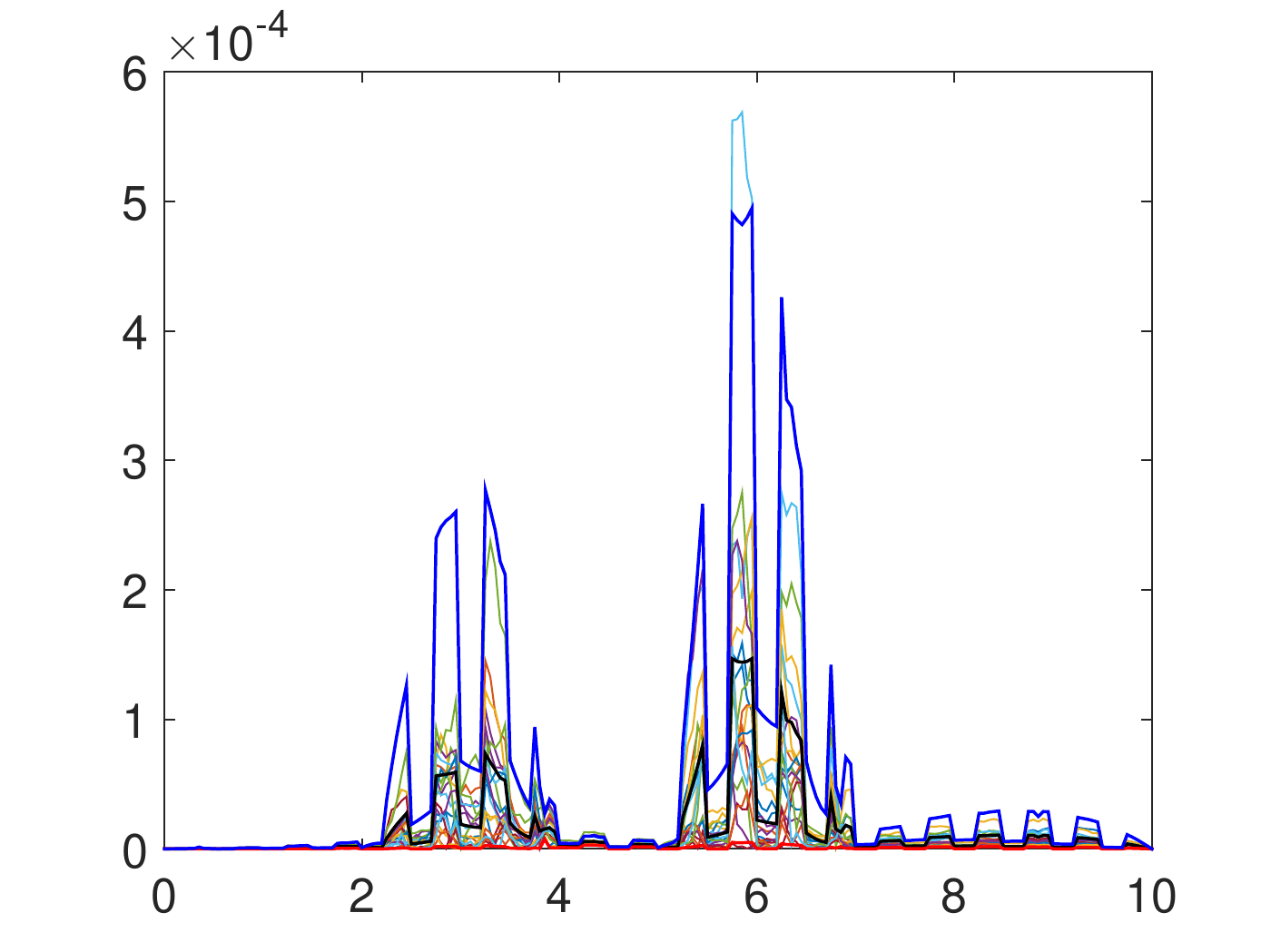}
\end{minipage}
\,
\begin{minipage}[b]{0.31\linewidth}
\centering\title{\qquad (IF$_2$) vs. (IF$_3$)}
\includegraphics[width=1.15\textwidth]{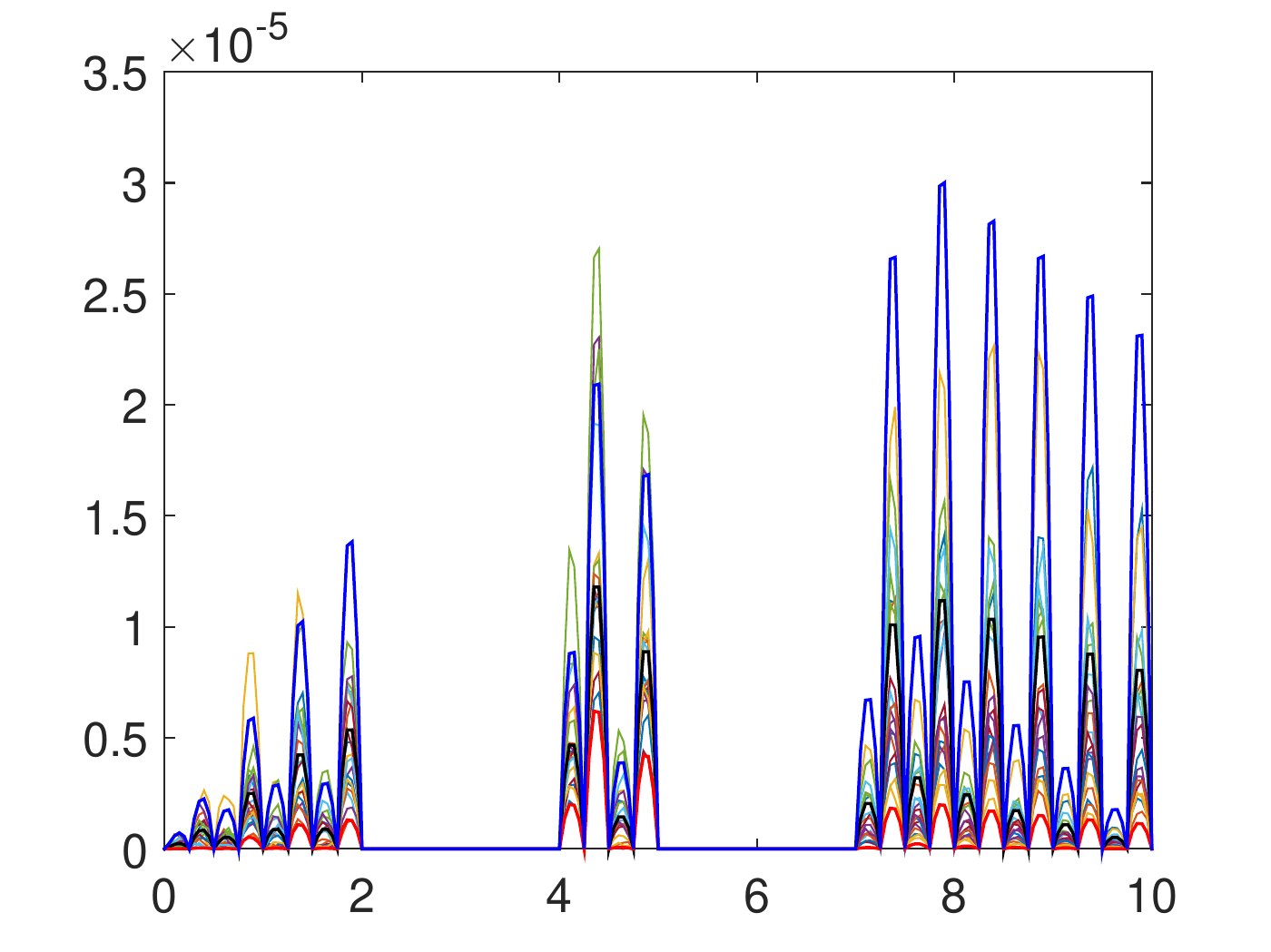}
\end{minipage}
\caption{Absolute difference between the basis swap price processes for 
different interpolating functions.}
\label{fig:r-P-differences}
\end{figure}

\begin{figure}[ht]
\centering\hspace{-1cm}
\begin{minipage}[b]{0.245\linewidth}
\includegraphics[width=1.15\textwidth]{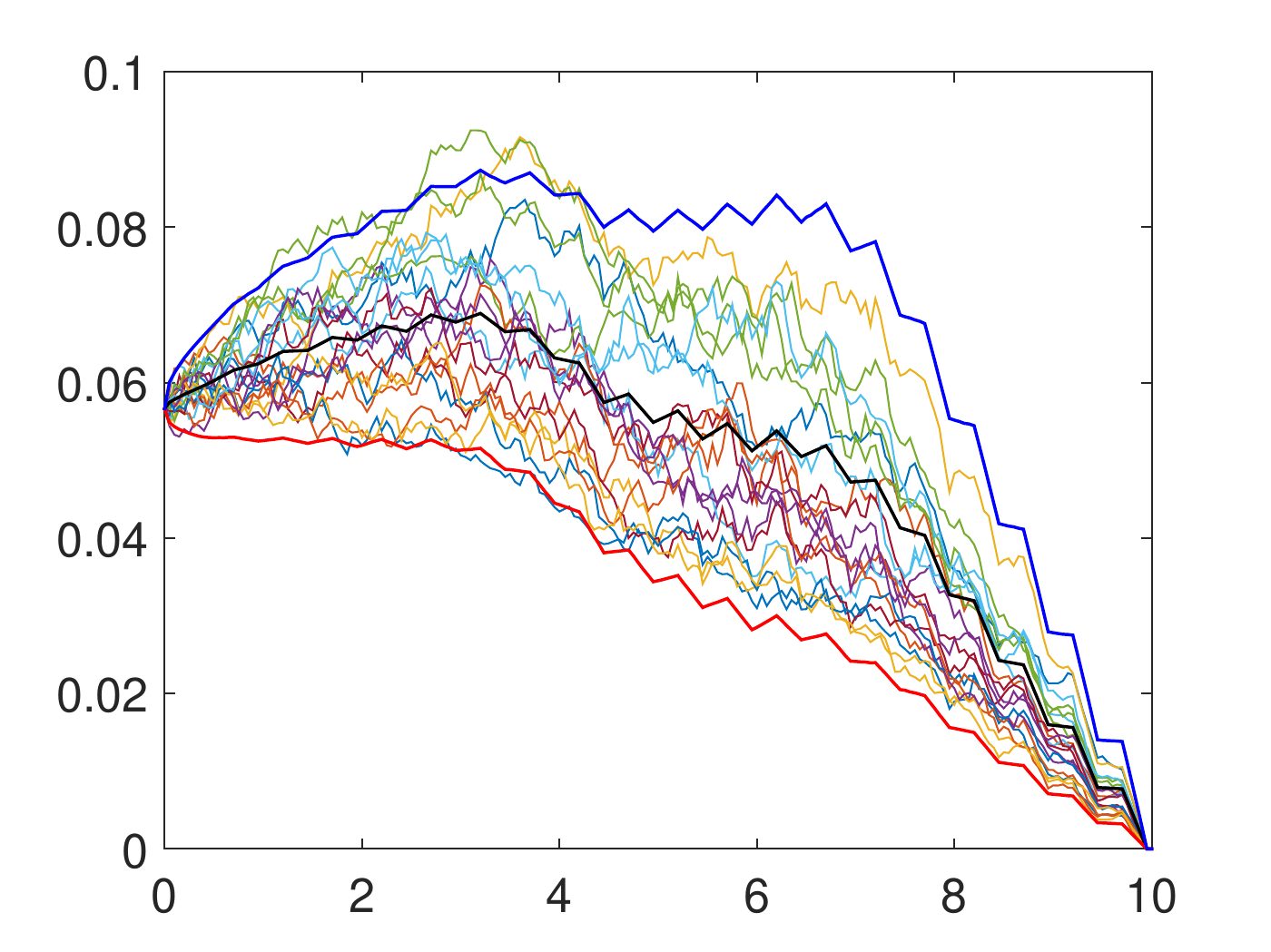}
\includegraphics[width=1.15\textwidth]{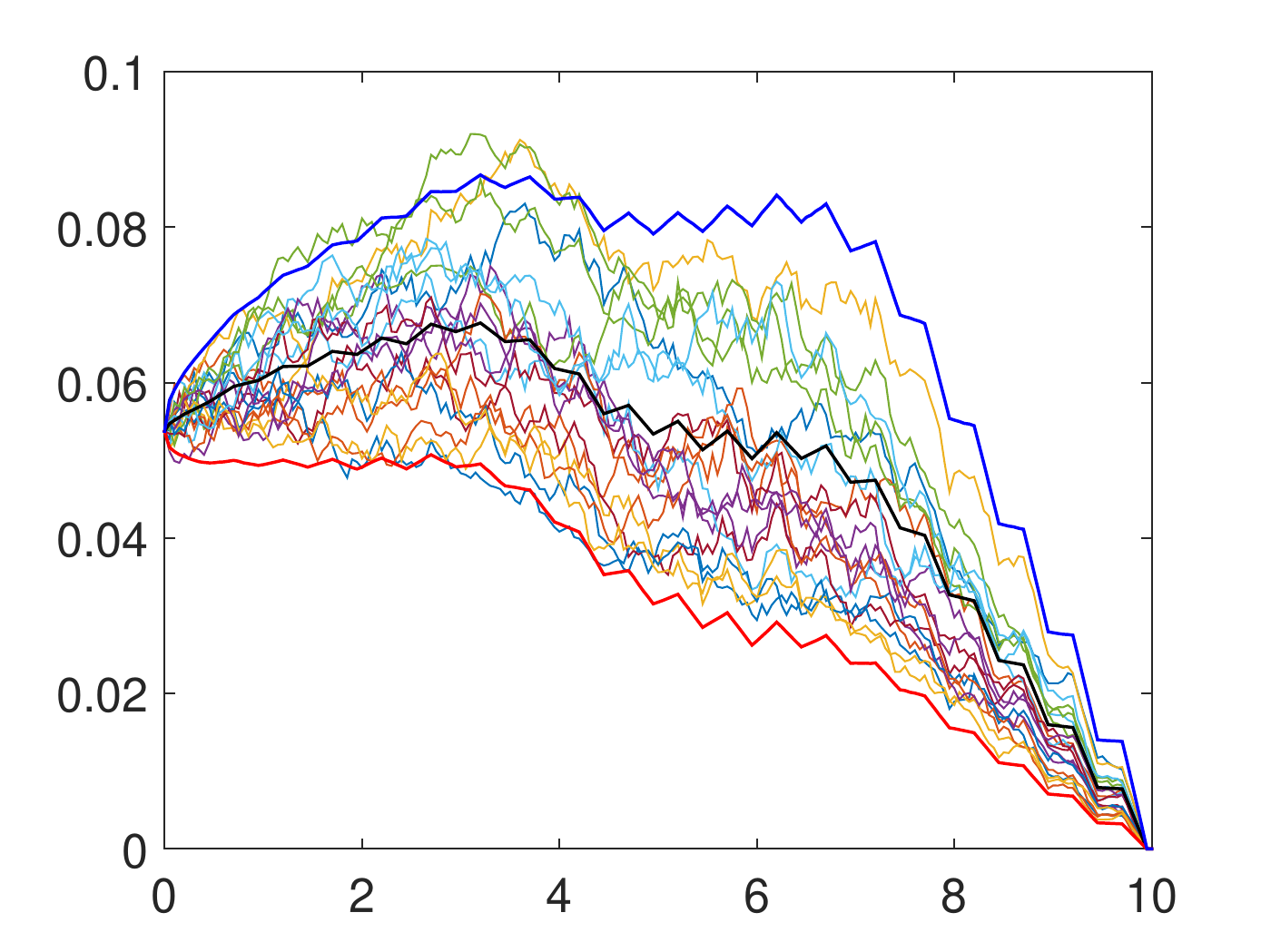}
\includegraphics[width=1.15\textwidth]{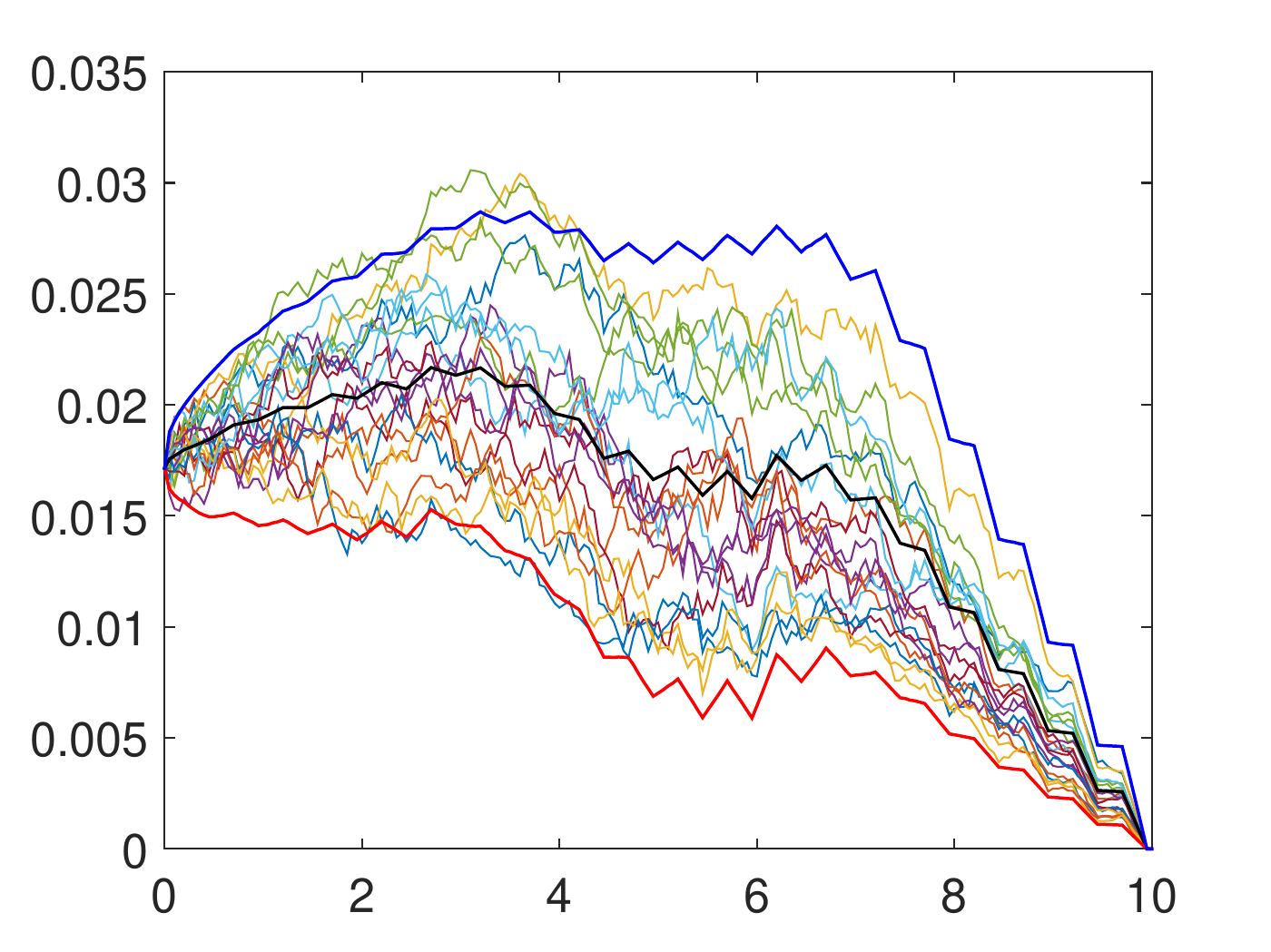}
\includegraphics[width=1.15\textwidth]{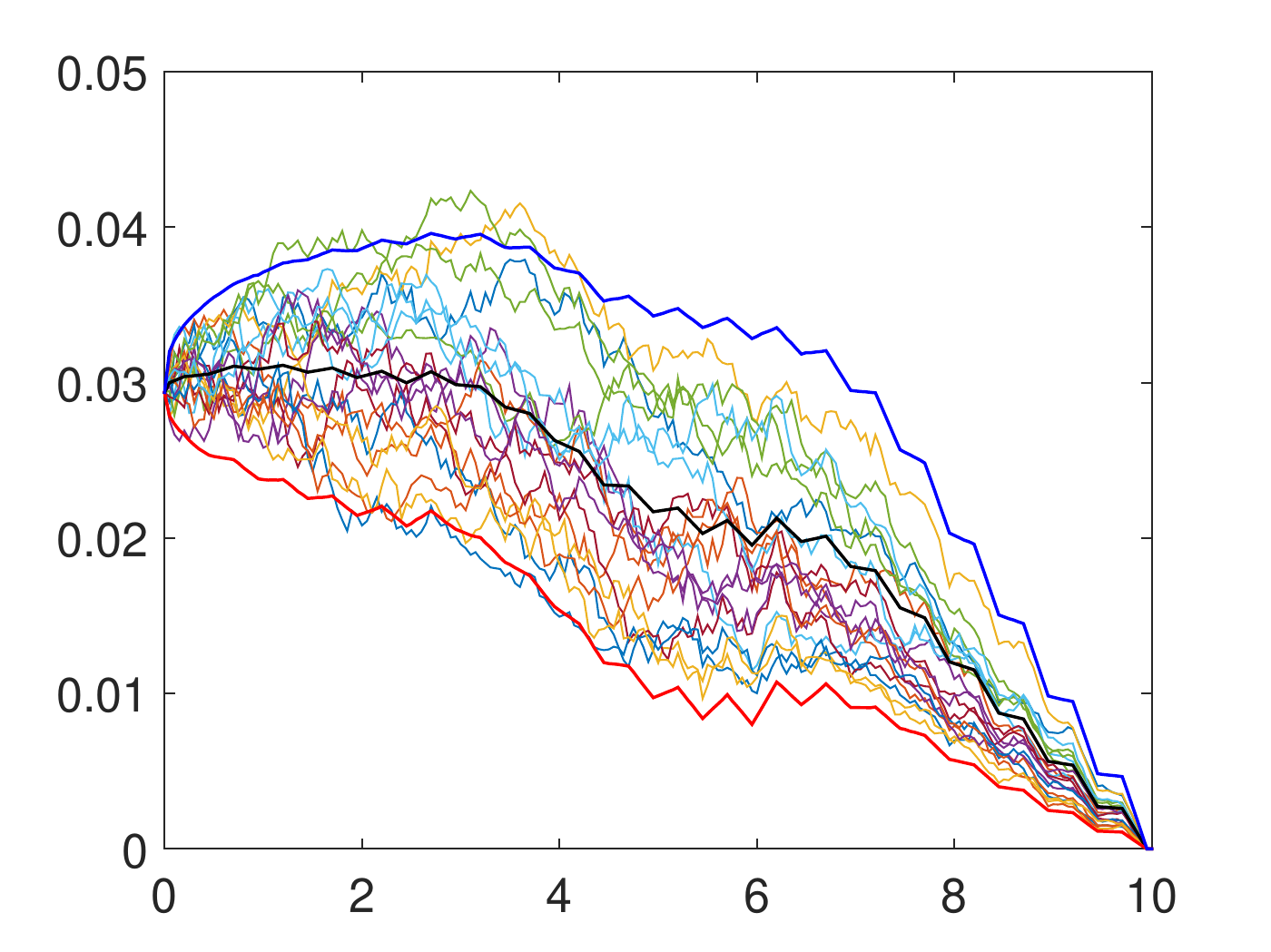}
\includegraphics[width=1.15\textwidth]{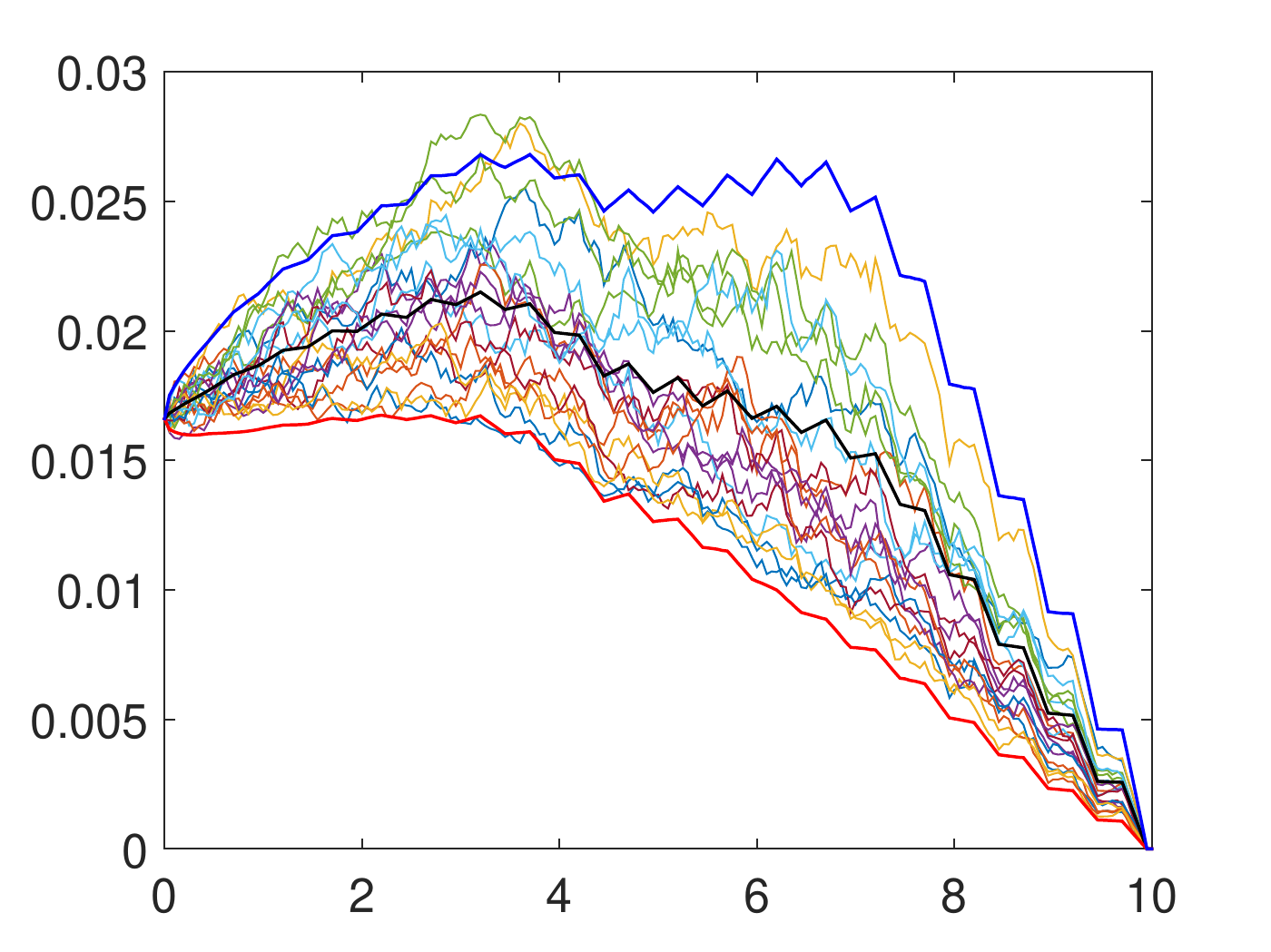}
\end{minipage}
\,
\begin{minipage}[b]{0.245\linewidth}
\centering\title{\qquad (IF$_1$) vs. (IF$_2$)}
\includegraphics[width=1.15\textwidth]{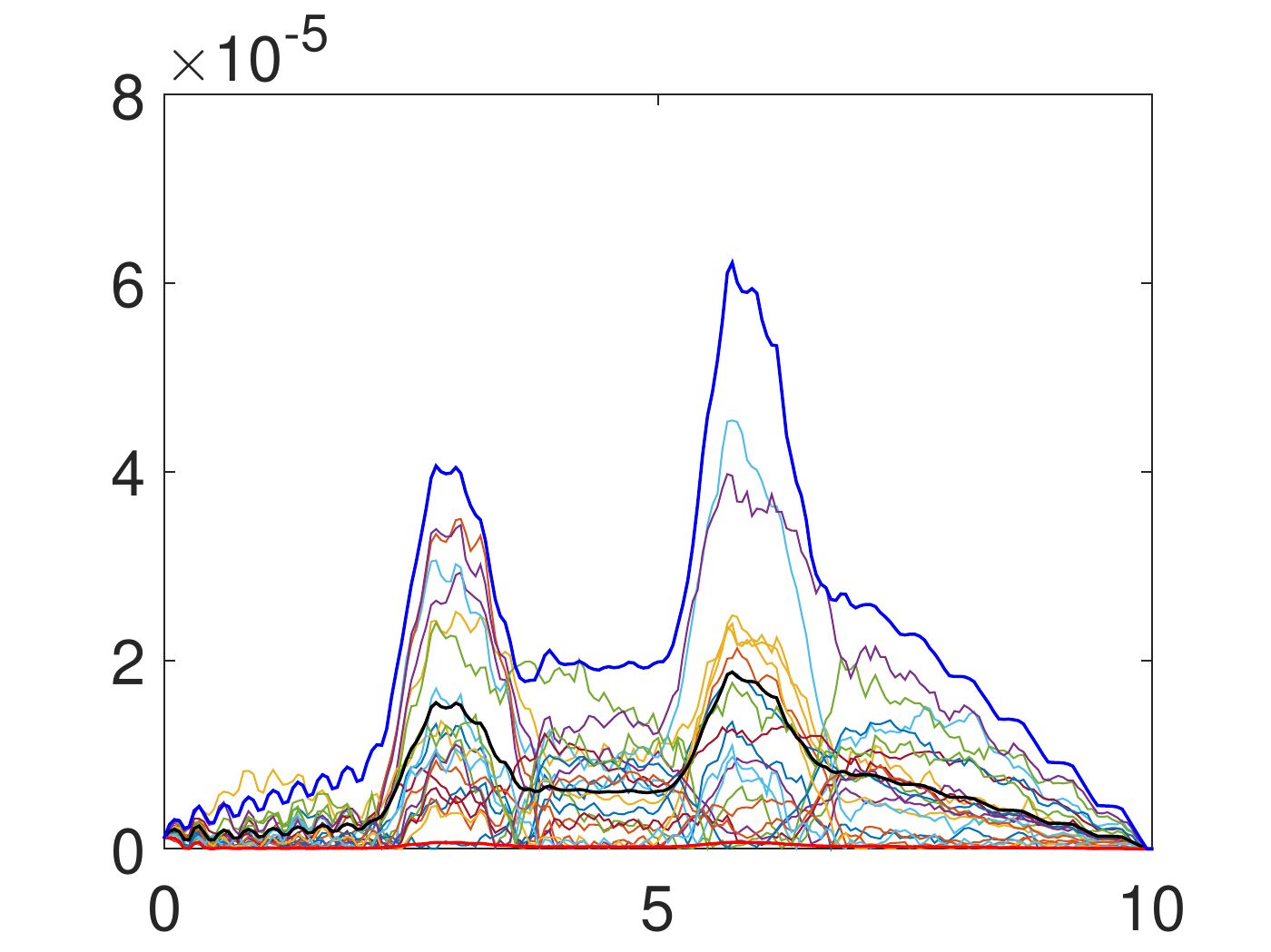}
\includegraphics[width=1.15\textwidth]{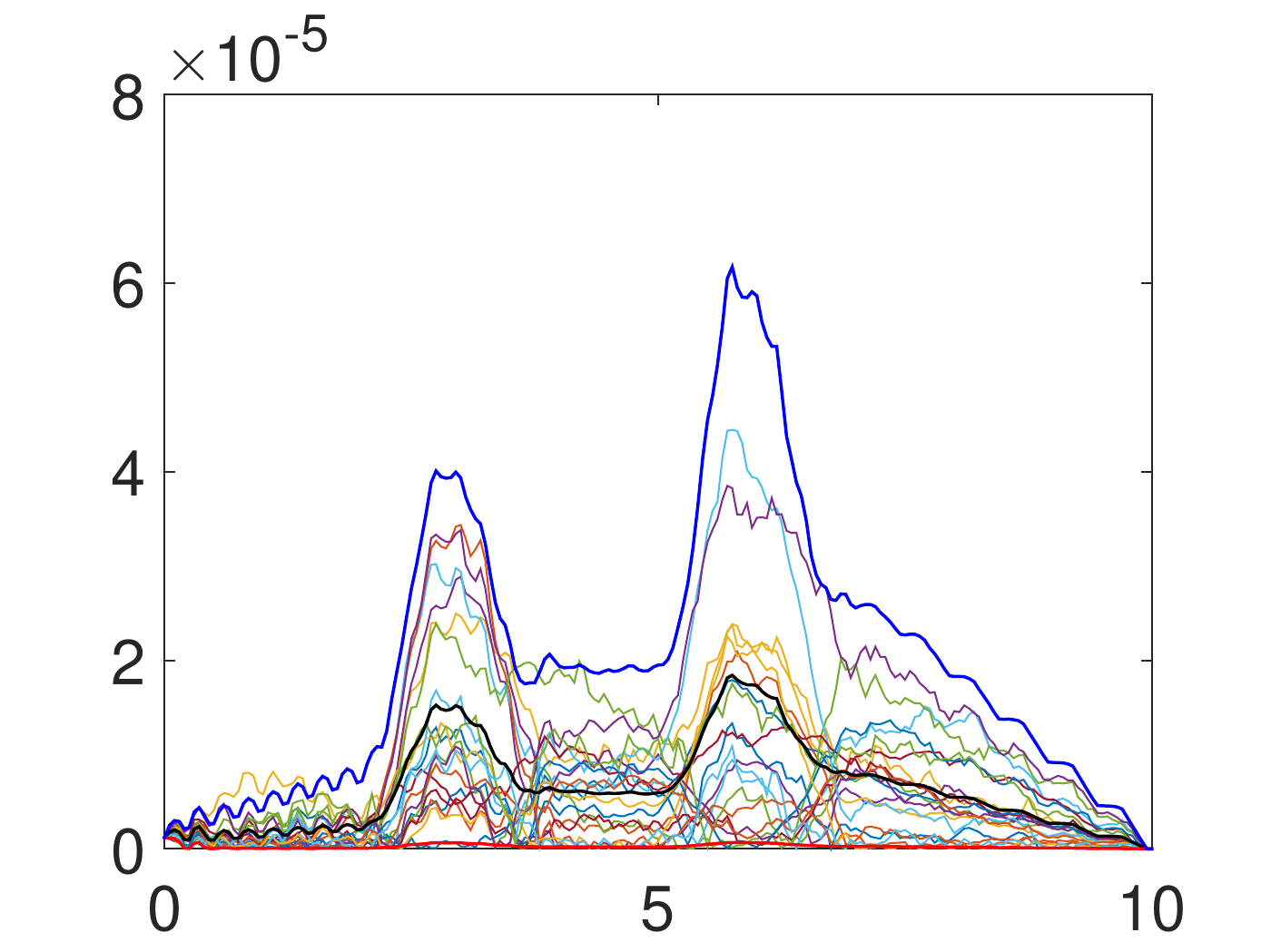}
\includegraphics[width=1.15\textwidth]{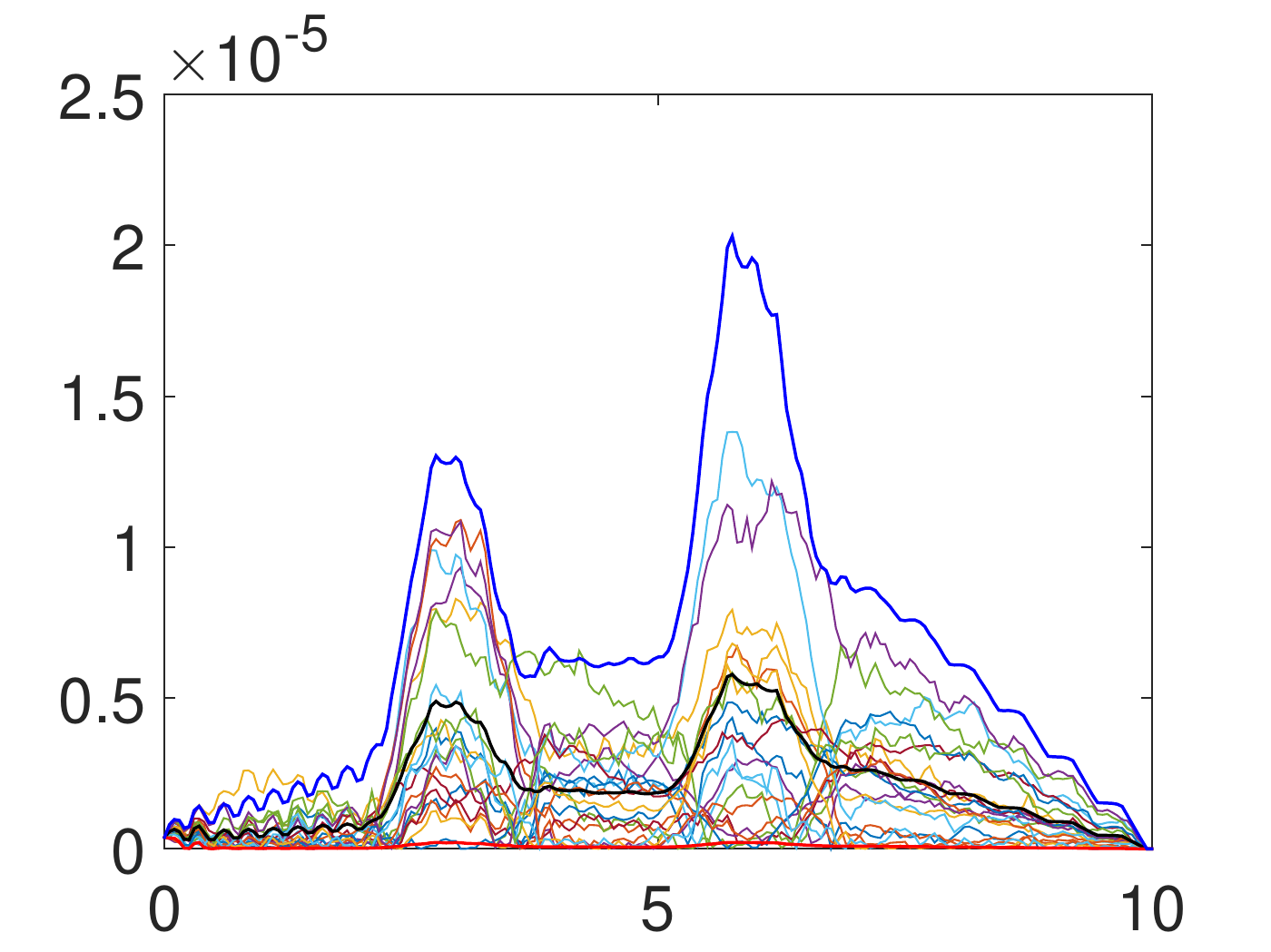}
\includegraphics[width=1.15\textwidth]{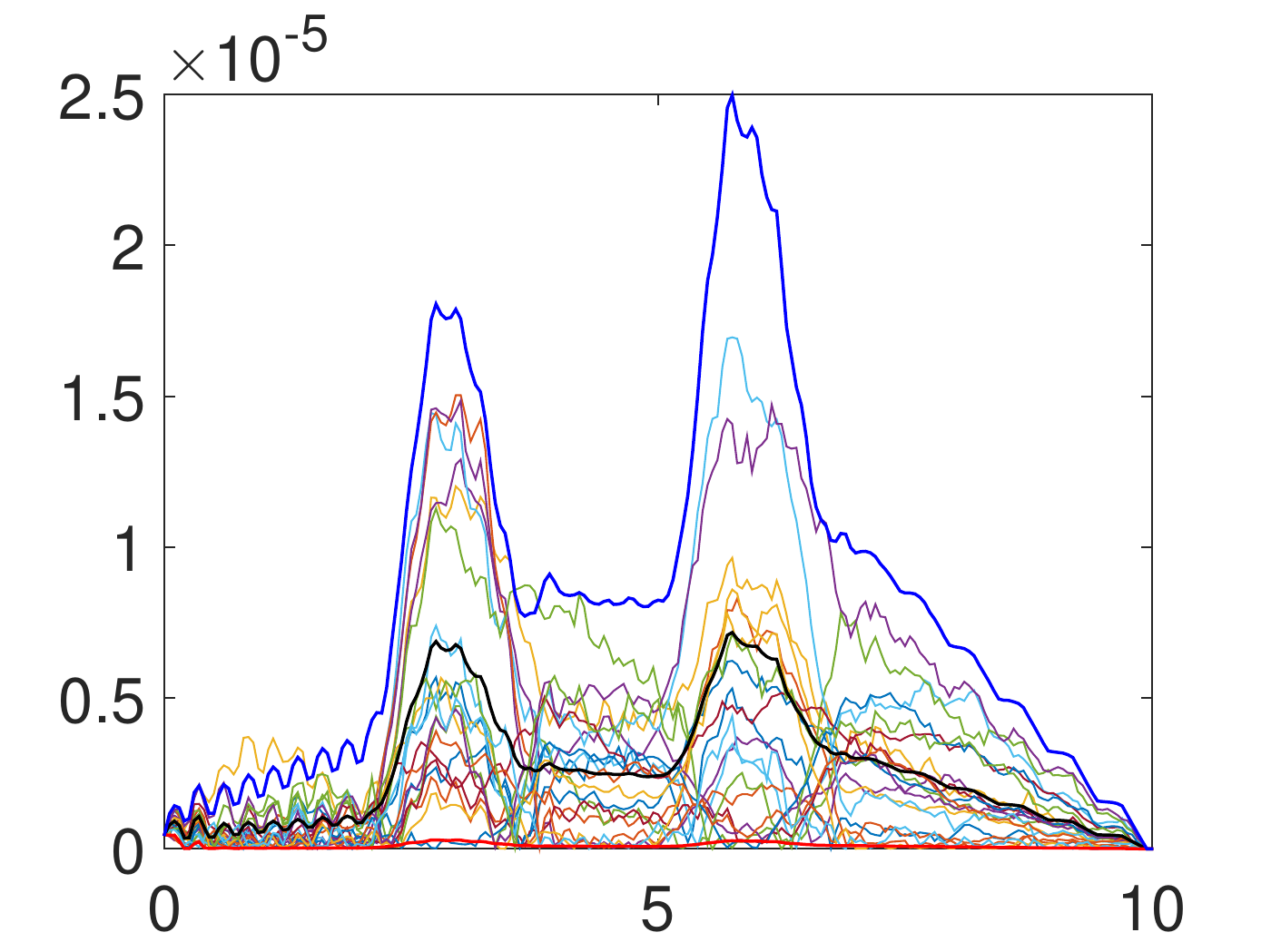}
\includegraphics[width=1.15\textwidth]{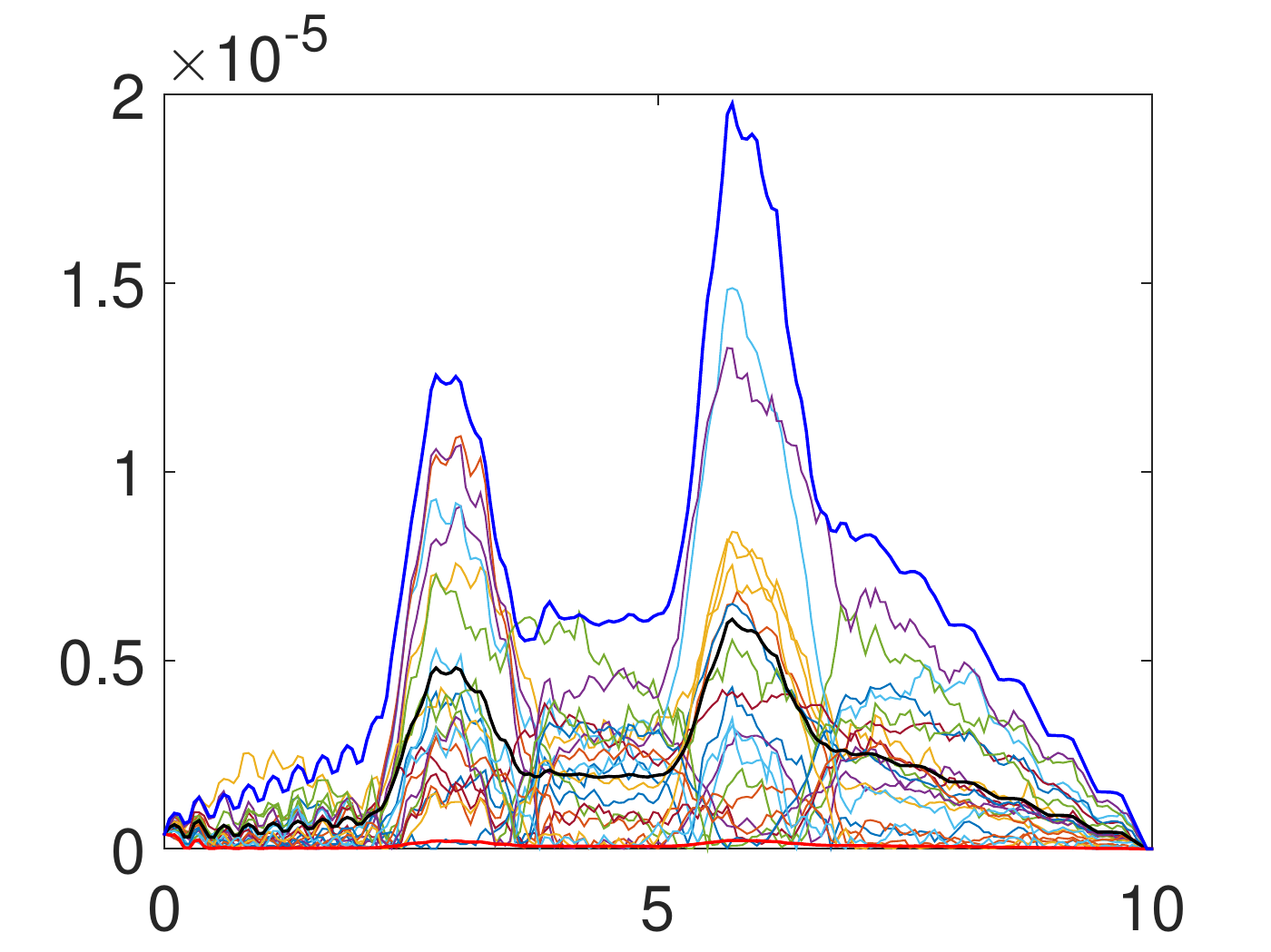}
\end{minipage}
\,
\begin{minipage}[b]{0.245\linewidth}
\centering\title{\qquad (IF$_1$) vs. (IF$_3$)}
\includegraphics[width=1.15\textwidth]{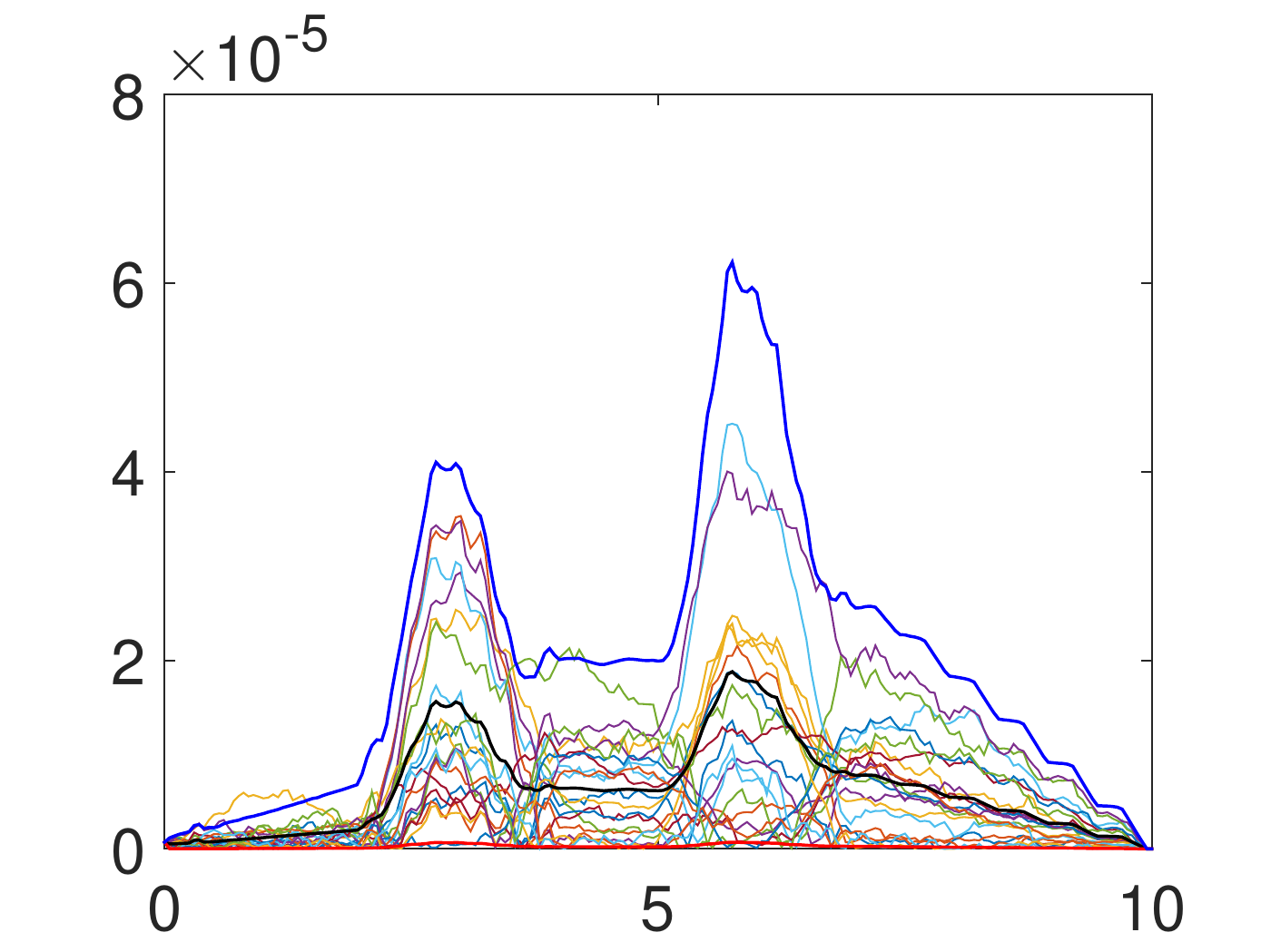}
\includegraphics[width=1.15\textwidth]{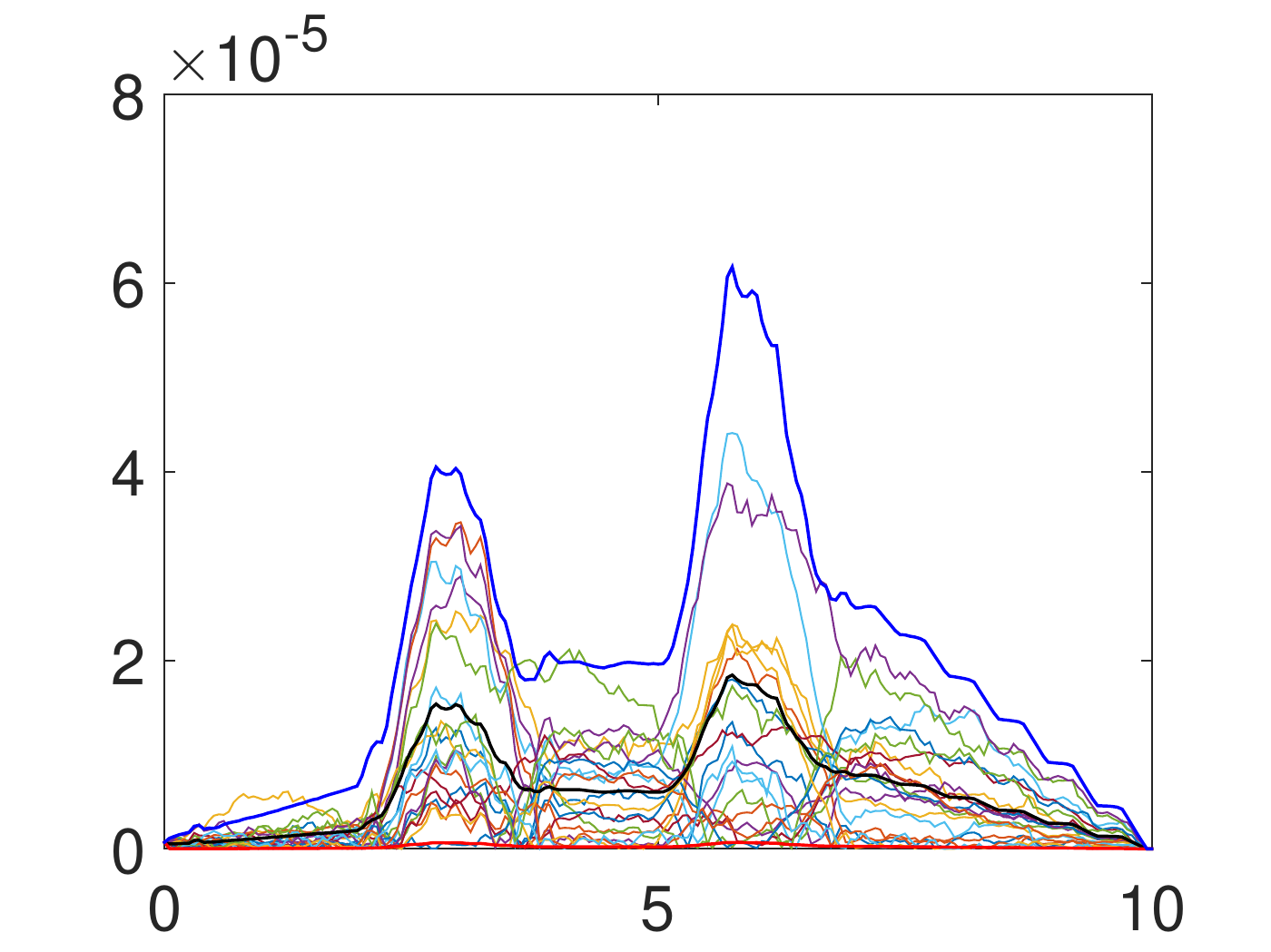}
\includegraphics[width=1.15\textwidth]{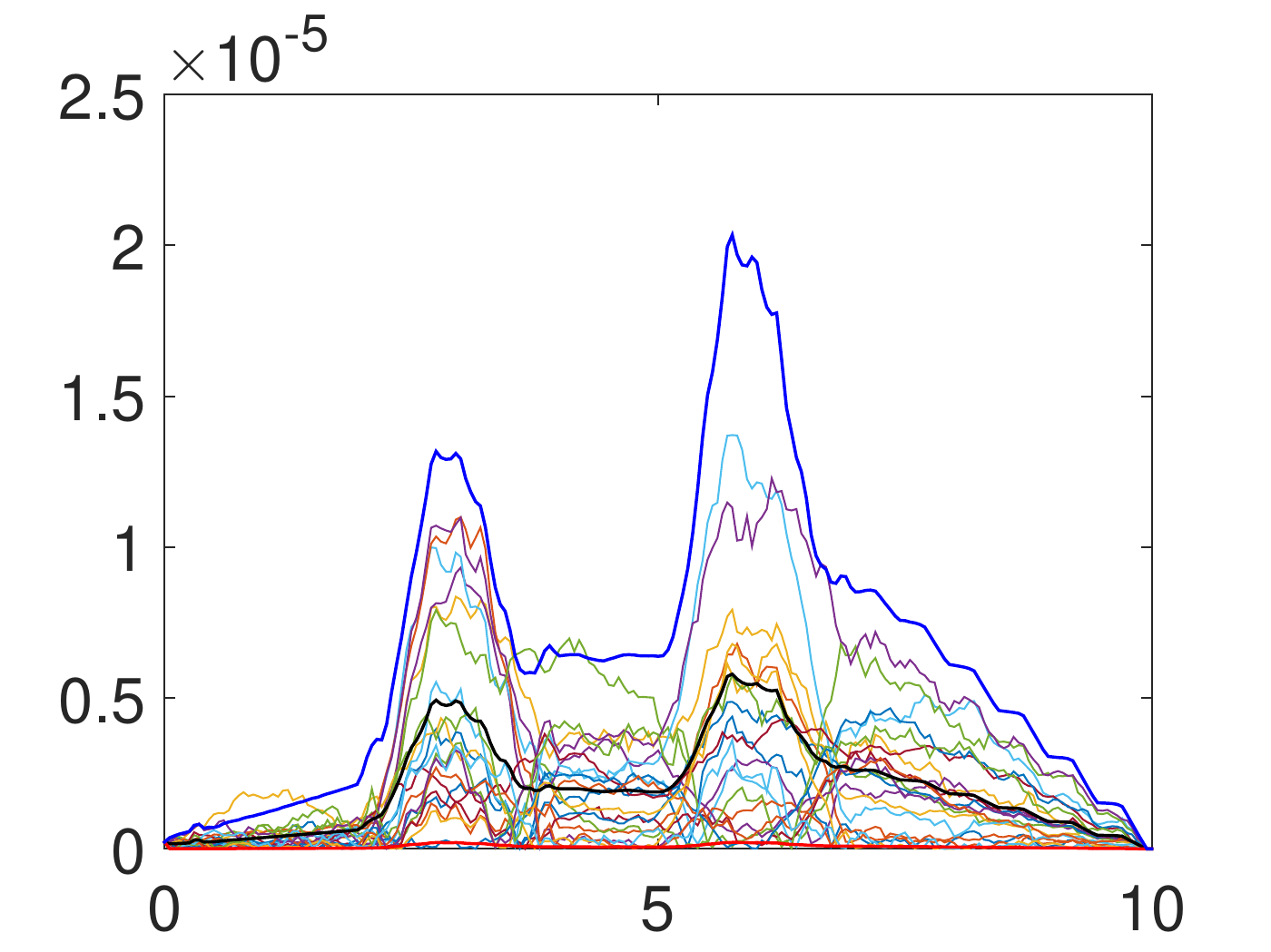}
\includegraphics[width=1.15\textwidth]{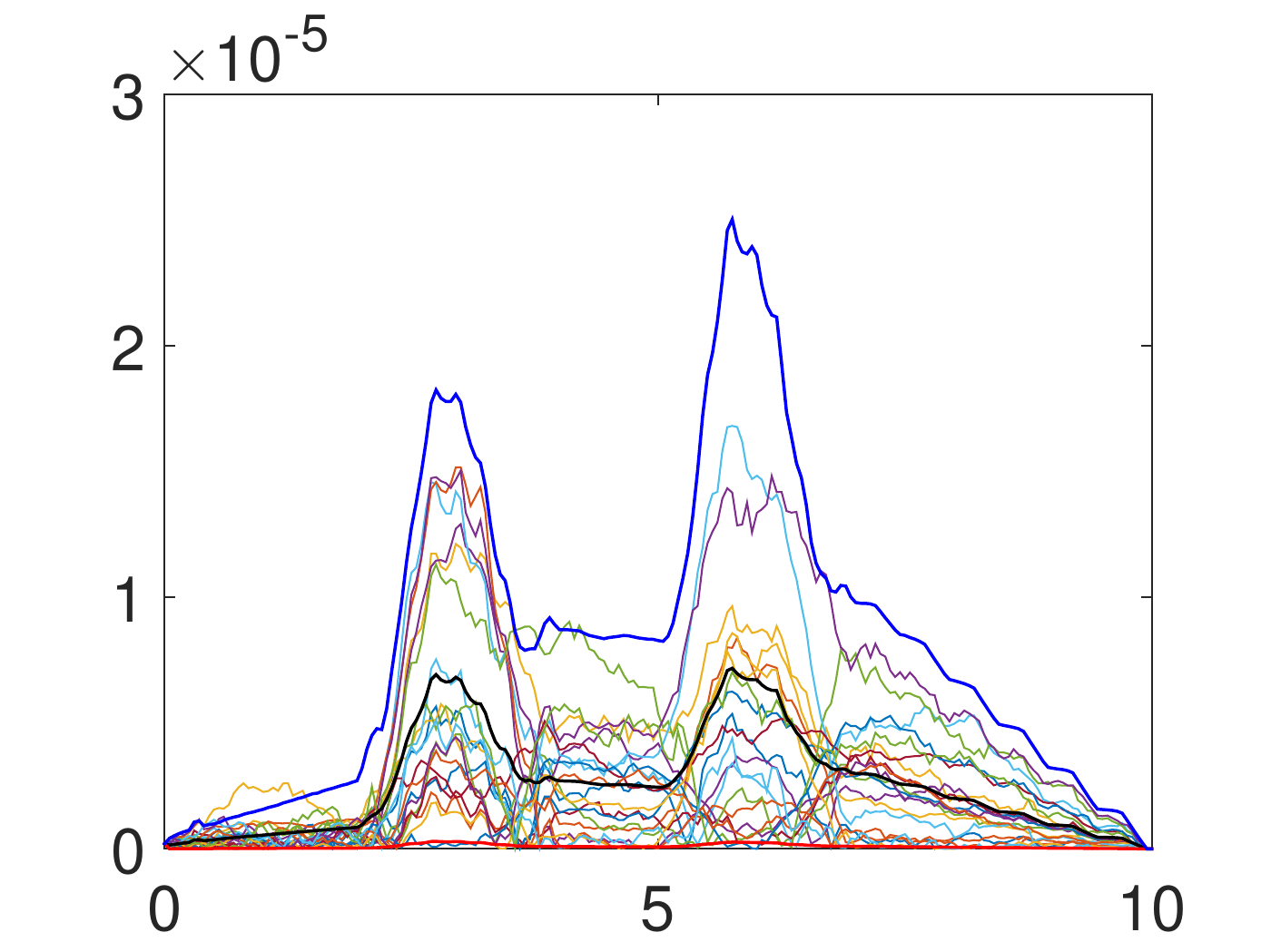}
\includegraphics[width=1.15\textwidth]{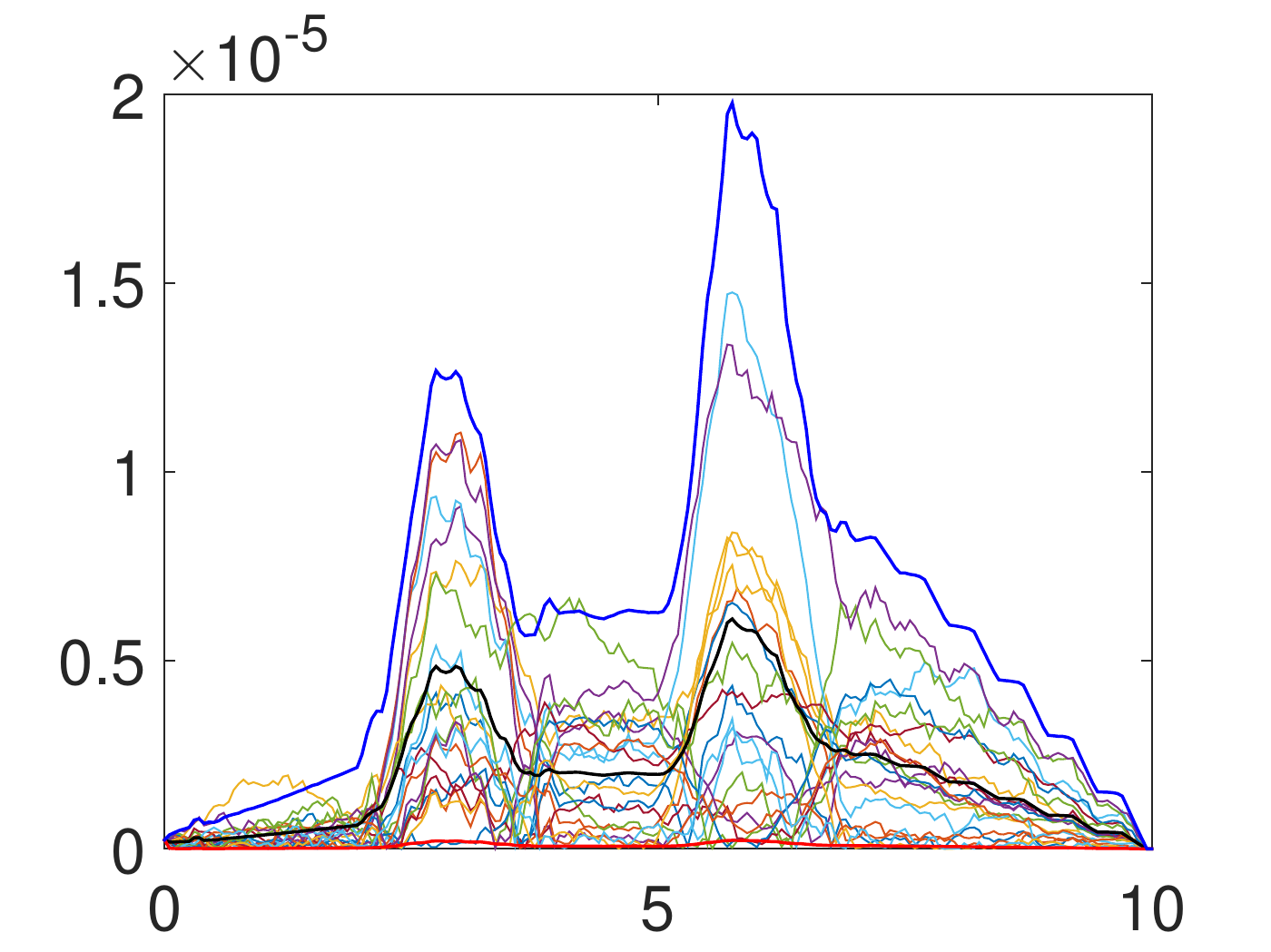}
\end{minipage}
\,
\begin{minipage}[b]{0.245\linewidth}
\centering\title{\qquad (IF$_2$) vs. (IF$_3$)}
\includegraphics[width=1.15\textwidth]{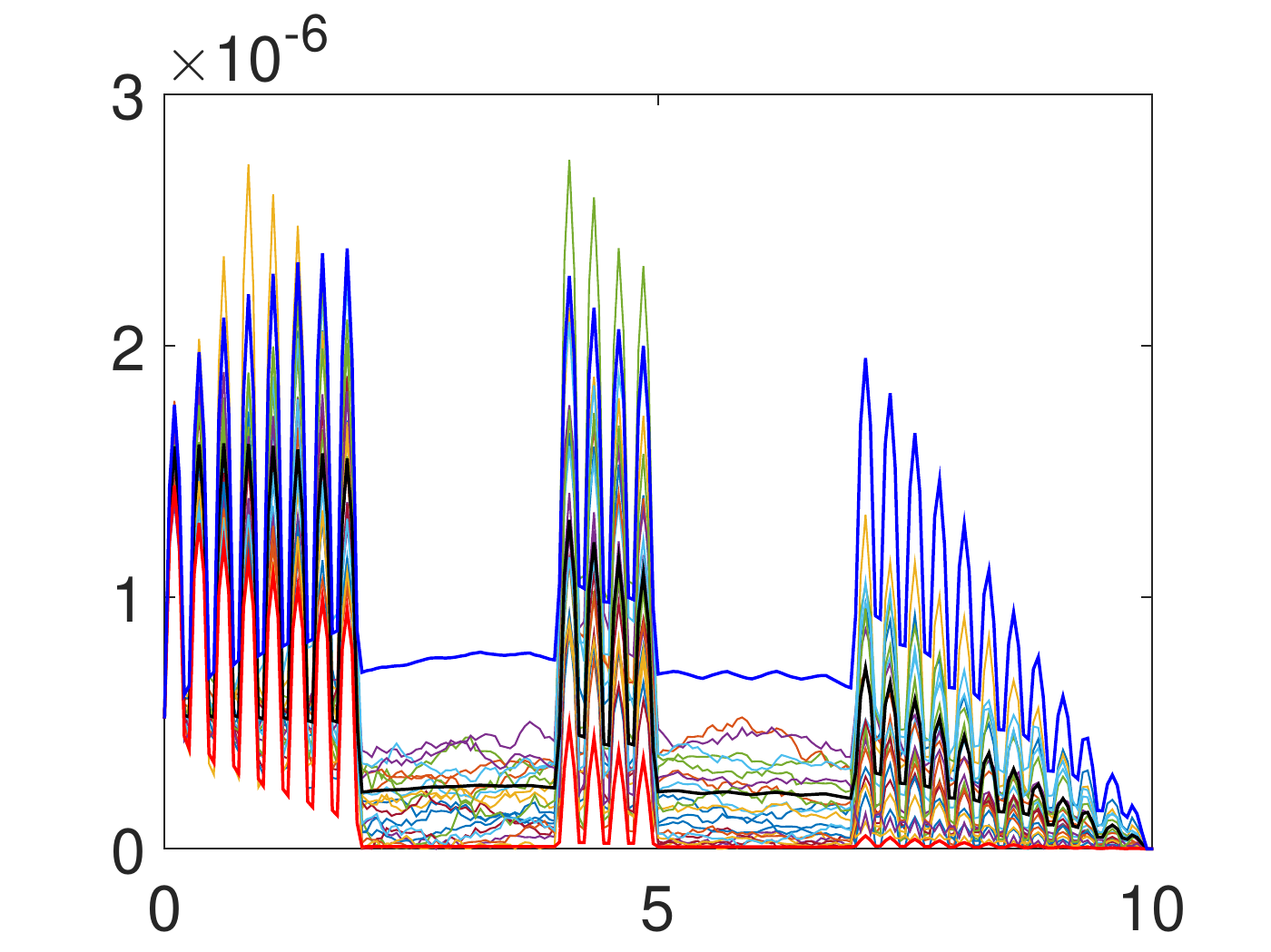}
\includegraphics[width=1.15\textwidth]{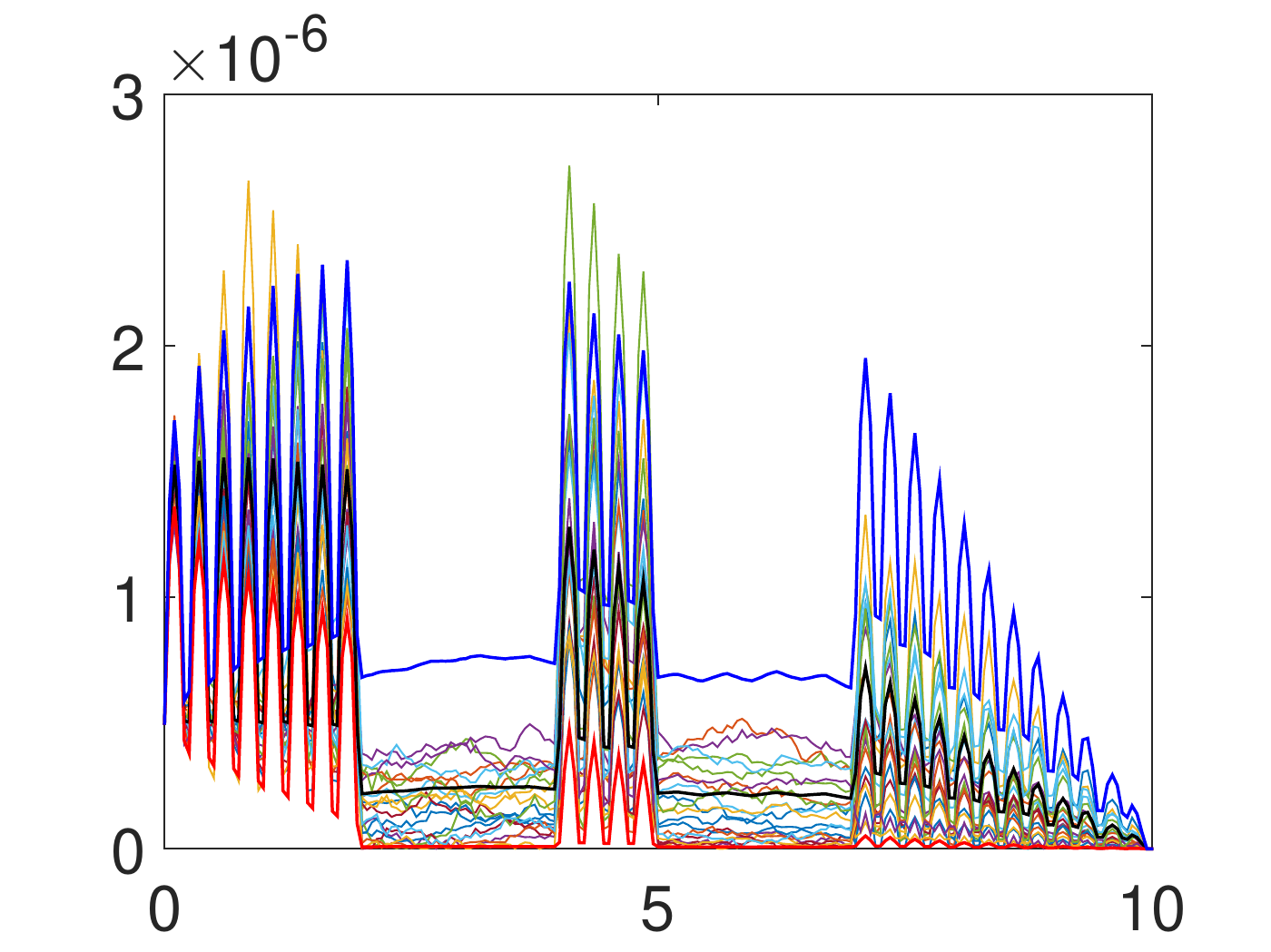}
\includegraphics[width=1.15\textwidth]{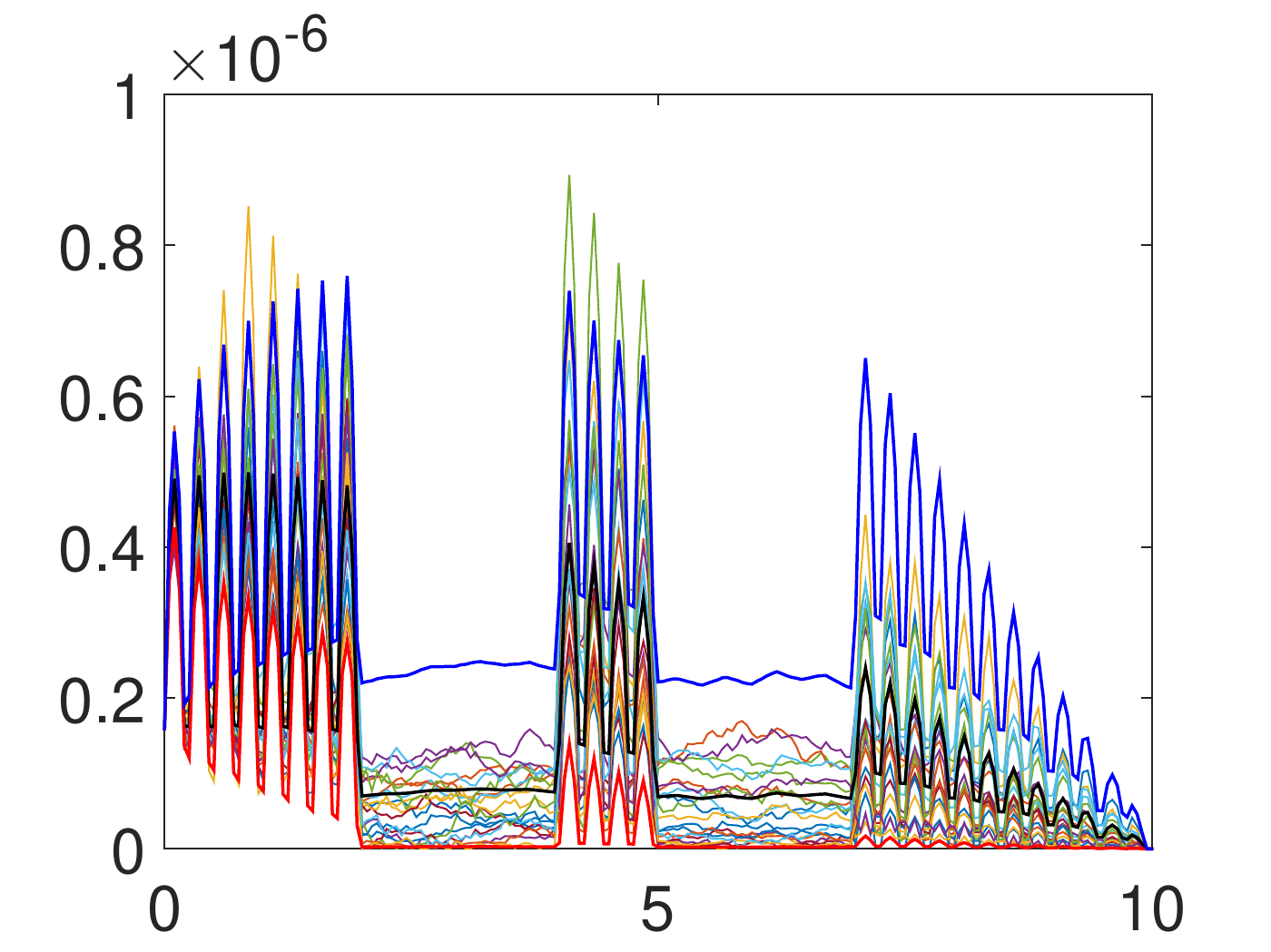}
\includegraphics[width=1.15\textwidth]{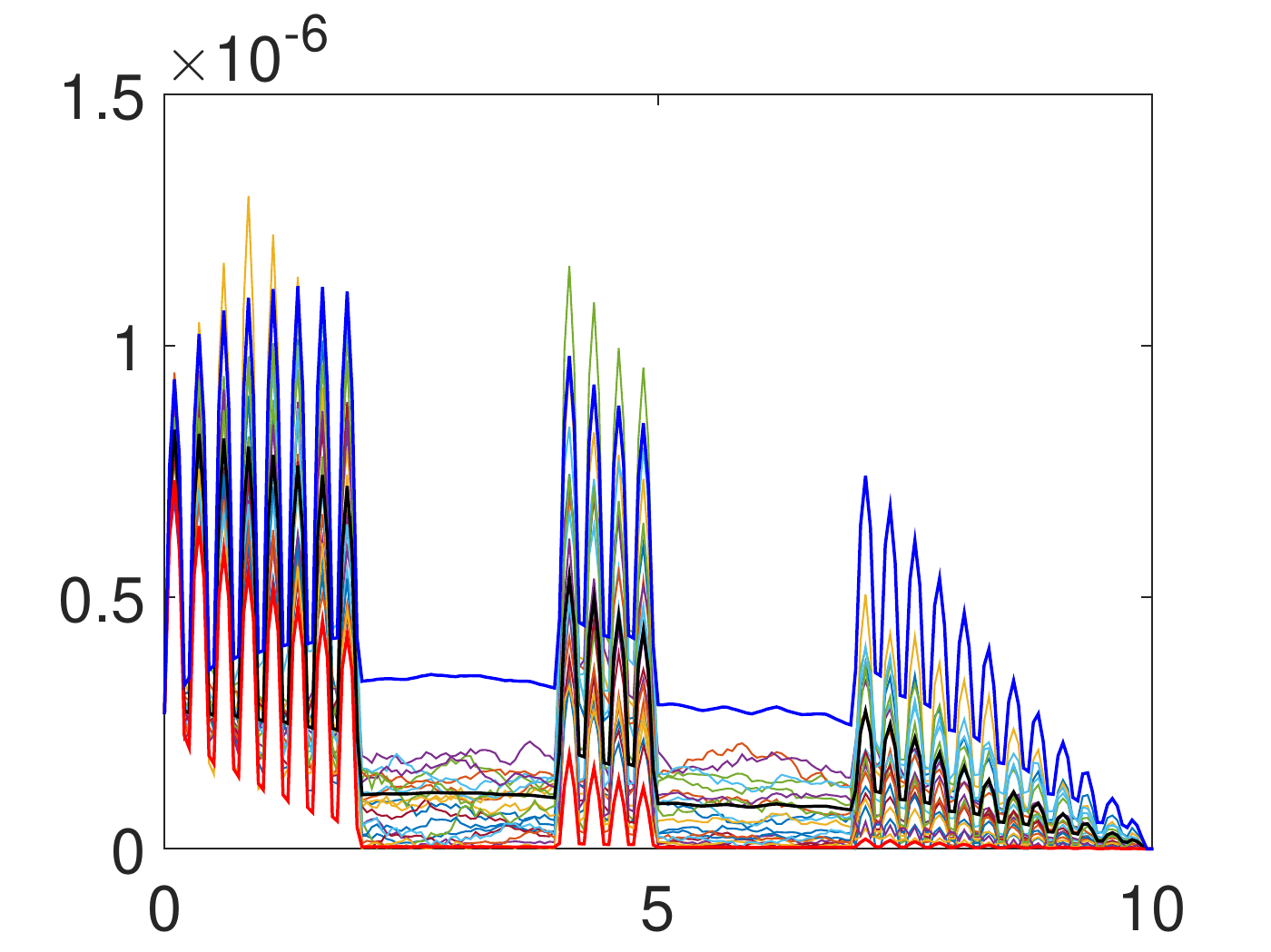}
\includegraphics[width=1.15\textwidth]{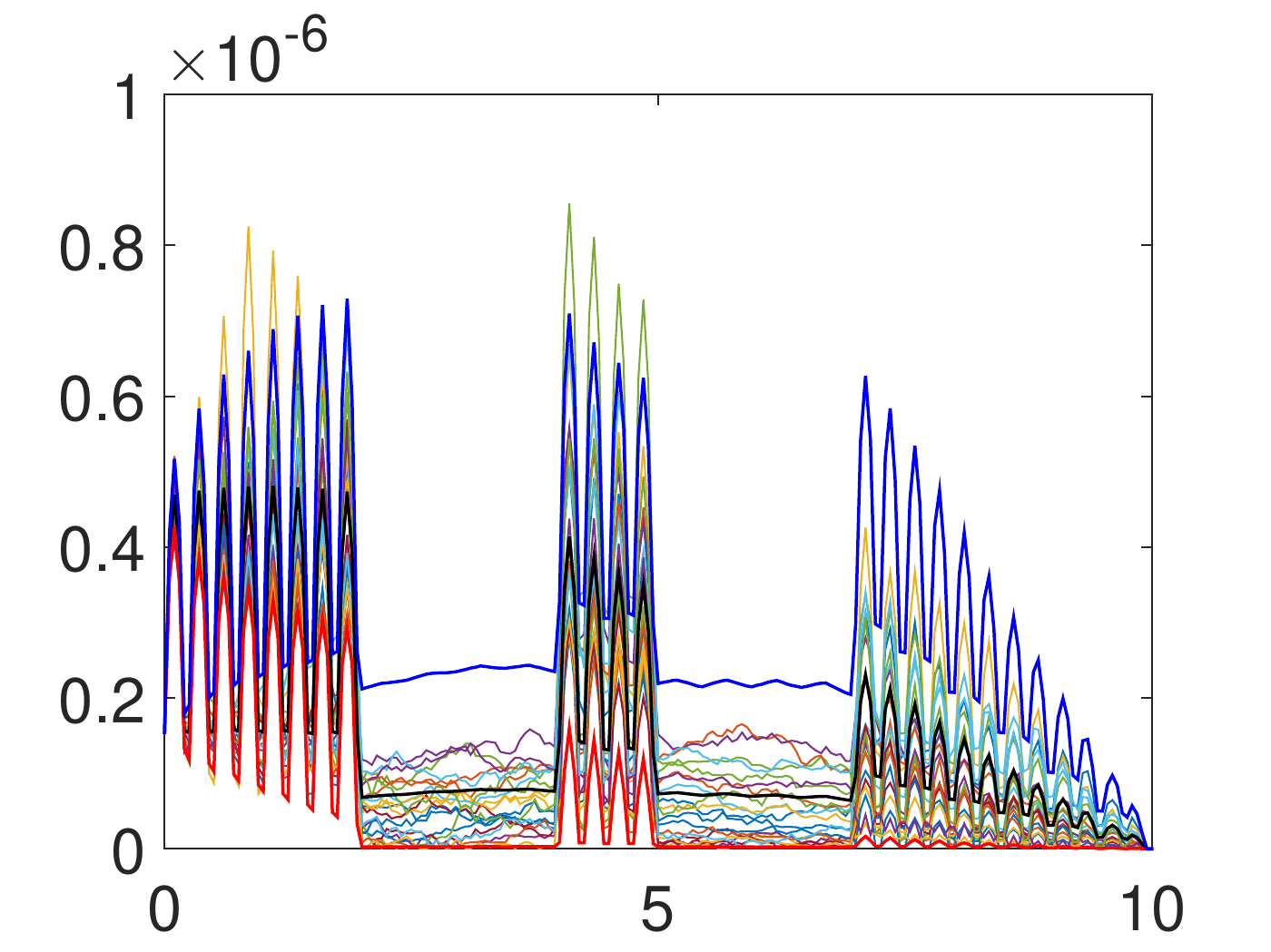}
\end{minipage}\caption{20 sample paths of the TVA process for the basis swap, together with 
the mean and percentiles, for the 5 different CSA specifications (left panels, 
top to bottom), and the difference between the TVA processes for the different 
interpolating functions (left to right).}
\label{fig:TVA-diffs}
\end{figure}


\subsection{Discussion}
\label{discussion}

The outcome of the numerical experiments is summarized in Figures \ref{fig:r-P-paths} until \ref{fig:TVA-diffs}. 
Starting with the top panel in Figure \ref{fig:r-P-paths}, we observe that there are significant structural differences in the dynamics of the short rate due to the different interpolating functions; this is mostly visible when looking at the averages and the percentile lines. 
The bottom panel in the same figure displays the price process of the 3M-6M basis swap for the different interpolating functions. 
As the differences are not as clearly visible as before, we have plotted the absolute differences in prices due to the different interpolating functions in Figure \ref{fig:r-P-differences}. 
There we see that notable differences in prices appear when using different interpolating functions (keep in mind, that the notional amount of the swap equals one, thus the deviations in prices are not negligible). 
As expected, the largest discrepancies between prices stemming from the first vs. second and the first vs. third interpolating functions occur on the curved section of the manifold used to construct $(u_k)$ and $(v_k)$. 
On the contrary, the discrepancies between prices from the second vs. third interpolating functions on the curved section of the manifold are zero, since both functions interpolate linearly in that segment.

The next Figure \ref{fig:TVA-diffs} depicts the sample path of the TVA process using the first interpolating function (left panels) for the five different CSA 
specifications (top to bottom), while the other figures show the differences in the TVA due to the different interpolating functions. 
The differences in the TVA are one order of magnitude smaller than the differences in prices, however the TVA itself is an order of magnitude smaller than the basis swap price. 
Reflecting the situation for the prices, the largest discrepancies between TVAs using the first and the other two interpolating functions occur around the 
curved section of the manifold. 
However, the discrepancies in the TVA in the flat sections of the manifold are more pronounced than the corresponding discrepancies in prices. 
The reason is that the interpolation affects value adjustments both via the basis swap price and via the short rate used for discounting, and its effect is propagated in different segments through the backward regression. 
This becomes clear when one looks at the differences between prices and value adjustments stemming from using the second and third interpolating functions; although the difference in prices is flat zero, the difference in value adjustments is far away from zero.

The numerical examples presented above show that the choice of the interpolating function entails significant model risk. 
The functions we chose are not especially far apart, in terms of their supremum norm, thus the differences above could become even higher. 
In fact, the coefficients of the short rate can become arbitrarily large due to the interpolating function. 
Therefore, both the manifolds on which the sequences $u$ and $v$ lie and the interpolating function have to be selected with caution, as they can fundamentally change the behaviour of the model.

\appendix
\section{Time-integration of Affine processes}

The following result is an extension of Theorem 4.10 in \citet{KellerRessel08}.

\begin{theorem}\label{thm:time-integral}
Let $\theta \colon [0,T] \rightarrow \R^d_{\geqslant0}$ be a bounded function and $(X_t)_{t\in[0,T]}$ be an affine process on $\R^d_{\geqslant0}$, with functional characteristics $F$ and  $R$. 
Then $\brac{X_t,\int_0^t \theta_u \circ X_u \du}_{t\in[0,T]}$ is a time-inhomogeneous affine process on $\R^d_{\geqslant0} \times \R^d_{\geqslant0}$ with functional characteristics
\[
\widetilde{F}\brac{t,u_X,u_Y}=F\brac{u_X}
	\quad {\text and}  \quad
\widetilde{R}\brac{t,u_X,u_Y}=\left( \begin{array}{c}  R\left(u_X\right)+\theta_t\circ u_{Y}\\0  \end{array}\right).
\]
Here $\circ$ denotes the componentwise multiplication between vectors.
\end{theorem}
 
\begin{proof}
Let $n\in\mathbb{N}$, $k\in\set{0,\dots,n}$, and define $s_k=\frac{k}{n}s$ and $h=\frac{s}{n}$, hence $s_0,\dots,s_n$ is an equidistant partition of $[0,s]$ with step size $h$. 
Approximating the integral with Riemann sums and using the dominated convergence theorem, we have
\begin{align*}
&\mathds{E} \cond{\exp\brac{\scal{u_X}{X_{t+s}}+\scal{u_Y}{\int_t^{t+s} \theta_r \circ X_r \dr}}}{\mathcal{F}_t}
\\
&\qquad=\lim_{n\rightarrow \infty}
\mathds{E} \cond{\exp\brac{\scal{u_X}{X_{t+s}}+\scal{u_Y}{h \sum_{k=0}^n \theta_{t+s_k}\circ X_{t+s_k}}}}{\mathcal{F}_t}
\\
&\qquad=:\lim_{n\rightarrow\infty} A_n.
\end{align*}
Using next the tower law of conditional expectations and the affine property of $X$, $A_n$ can be written as follows
\begin{align*}
A_n
&= \mathds{E}\cond{\mathds{E}\cond{\e^{\scal{u_X+h u_Y\circ \theta_{t+s_n}}{X_{t+s_n}}}}{\mathcal{F}_{t+s_{n-1}}} \e^{\scal{u_Y}{h\sum_{k=0}^{n-1}\theta_{t+s_k}\circ X_{t+s_k}}}}{\mathcal{F}_t}
\\
&= \mathds{E}\cond{\e^{\phi_{h}\brac{u_X+h u_Y\circ \theta_{t+s_n}}+\scal{\psi_{h}\brac{u_X+h u_Y\circ \theta_{t+s_n}}}{X_{t+s_{n-1}}}+\scal{u_Y}{h\sum_{k=0}^{n-1}\theta_{t+s_k}\circ X_{t+s_k}}}}{\mathcal{F}_{t}}
\\
&= \mathds{E}\left[\mathds{E}\cond{\e^{\brac{\scal{\psi_{h}\brac{u_X+h u_Y\circ \theta_{t+s_n}}+h u_Y\circ \theta_{t+s_{n-1}}}{X_{t+s_{n-1}}}}}}{\mathcal{F}_{t+s_{n-2}}}\right.
\\
&\qquad \left. \left.\times \ \e^{\scal{u_Y}{h\sum_{k=0}^{n-2}\theta_{t+s_k}\circ X_{t+s_k}}}\right|\mathcal{F}_t\right] \e^{\phi_h\brac{u_X+h u_Y\circ \theta_{t+s_n}}}
\\
&= \mathds{E}\cond{\e^{\scal{\psi_h\brac{\psi_{h}\brac{u_X+h u_Y\circ \theta_{t+s_n}}+h u_Y\circ \theta_{t+s_{n-1}}}}{X_{t+s_{n-2}}}+\scal{u_Y}{h\sum_{k=0}^{n-2}\theta_{t+s_k}\circ X_{t+s_k}}}}{\mathcal{F}_{t}}
\\
&\qquad \times \ \e^{\phi_h \brac{u_X+h u_Y \circ \theta_{t+s_n}}+\phi_h\brac{\psi_h \brac{u_X+h u_Y\circ \theta_{t+s_n}}+h u_Y\circ \theta_{t+s_{n-1}}}}.
\end{align*}
Iterating this procedure, we arrive at
\[
A_n=\exp\Big(p_n\brac{u_X,u_Y}+\scal{q_n\brac{u_X,u_Y}}{X_t} \Big),
\]
with $p_0\brac{u_X,u_Y}=0$, $q_0\brac{u_X,u_Y}=u_X+h u_Y\circ \theta_{t+s}$ and
\begin{align*}
p_{k+1}\brac{u_X,u_Y}&=p_k\brac{u_X,u_Y}+\phi_h\brac{q_k\brac{u_X,u_Y}}
\\
q_{k+1}\brac{u_X,u_Y}&=\psi_h\brac{q_k\brac{u_X,u_Y}}+h u_Y\circ \theta_{t+s_{n-(k+1)}}.
\end{align*}
Using the generalized Riccati equation \eqref{Riccati}, we can expand $\phi$ and $\psi$ linearly around the origin. 
Thus we get
\begin{align*}
p_{k+1}\brac{u_X,u_Y}&=p_k\brac{u_X,u_Y}+h F\brac{q_k\brac{u_X,u_Y}}+o\brac{h},
\\
q_{k+1}\brac{u_X,u_Y}&=q_{k}\brac{u_X,u_Y}+h\brac{R\brac{q_{k}\brac{u_X,u_Y}}+u_Y\circ \theta_{s_{n-(k+1)}}}+o\brac{h}.
\end{align*}
As $\theta$ is nonnegative and bounded, the second part of the proof of Theorem 4.10 in \citet{KellerRessel08} remains the same. 
Hence, the recursive scheme above is an Euler-type approximation, starting from the terminal time, to the ODE
\begin{align*}
\frac{\partial}{\partial s} p\brac{s,t,u_X,u_Y}&=F\brac{q\brac{s,t,u_X,u_Y}},
\\
\frac{\partial}{\partial s} q\brac{s,t,u_X,u_Y}&=R\brac{q\brac{s,t,u_X,u_Y}}+u_Y\circ \theta_s
\end{align*}
with initial conditions $p\brac{r,r,u_X,u_Y}=0$ and $q\brac{r,r,u_X,u_Y}=u_X$, for all $r\geq 0$.
\end{proof}

\bibliographystyle{plainnat}
\bibliography{references}

\end{document}